\documentclass[10pt,a4]{amsart}

\usepackage{amssymb}
\usepackage{amsmath}
\usepackage{verbatim}
\usepackage{enumerate}
\usepackage{hyperref}
\usepackage{bm}   
\usepackage{extarrows}
\usepackage{longtable}
\usepackage{amsthm}
\usepackage{color}
\usepackage{graphicx}
\usepackage{xcolor}
\usepackage{amscd}
\usepackage{dsfont}   
\usepackage{bbm}   

\usepackage{tabu}
\usepackage[sep=3pt, offset=1.2em]{simpler-wick} 

\newcommand{\RR}{\mathbb{R}}
\newcommand{\CC}{\mathbb{C}}
\newcommand{\ZZ}{\mathbb{Z}}

\newcommand{\dd}{\partial}
\newcommand{\db}{\bar\partial}
\newcommand{\tot}{\mathrm{tot}}
\newcommand{\simeom}{\underset{\mathrm{e.o.m.}}{\sim}}
\newcommand{\lan}{\left\langle}
\newcommand{\ran}{\right\rangle}
\newcommand{\mc}{\mathcal}
\newcommand{\mr}{\mathrm}

\newcommand{\ra}{\rightarrow}

\newcommand{\til}{\widetilde}
\newcommand{\OO}{\mathcal{O}}
\newcommand{\cF}{\mathcal{F}}
\newcommand{\reg}{\mathrm{reg.}}
\newcommand{\rg}{\mathrm{reg}^{(\infty)}}

\newcommand{\DD}{\mathsf{D}}
\newcommand{\II}{\mathbb{I}}
\newcommand{\JJ}{\mathbb{J}}

\newcommand{\cl}{\mathrm{cl}}
\newcommand{\g}{\mathfrak{g}}
\newcommand{\gm}{\pmb{\gamma}}
\newcommand{\ad}{\mathrm{ad}}
\newcommand{\nl}{:\!}
\newcommand{\nr}{\!:}
\newcommand{\soak}{\Theta}
\newcommand{\tsoak}{\widetilde{\Theta}}
\newcommand{\tDelta}{\widetilde{\Theta}}

\newcommand{\idg}{\mathds{1}}
\newcommand{\ff}{\mathbbm{f}}
\newcommand{\sg}{\mathsf{g}}

\theoremstyle{definition}
\newtheorem{remark}{Remark}[section]
\newtheorem{example}[remark]{Example}
\newtheorem{lemma}[remark]{Lemma}
\newtheorem{proposition}[remark]{Proposition}
\newtheorem{corollary}[remark]{Corollary}

\allowdisplaybreaks

\begin{document}

\title[Non-abelian $BF$ theory as a conformal field theory]{Two-dimensional non-abelian $BF$ theory in Lorenz gauge as  
a solvable
logarithmic TCFT  
}

\begin{abstract}
We study two-dimensional non-abelian $BF$ theory in Lorenz gauge and prove that it is a topological conformal field theory. This opens the possibility to compute topological string amplitudes (Gromov-Witten invariants). We found that the theory is exactly solvable in the sense that all correlators are given by finite-dimensional convergent integrals. Surprisingly, this theory turns out to be logarithmic in the sense that there are  correlators given by polylogarithms and powers of logarithms.
Furthermore, we found fields with ``logarithmic conformal dimension'' (elements of a Jordan cell for $L_0$). We also found certain vertex operators with anomalous dimensions that depend on the non-abelian coupling constant. The shift of dimension  of composite fields may be understood as arising from the dependence of subtracted singular terms on local coordinates. This 
generalizes the well-known explanation of anomalous dimensions of vertex operators in the free scalar field theory.
\end{abstract}

\author{Andrey S. Losev}
\address{Federal Science Centre ``Science Research Institute of System Analysis at Russian Science Academy''
(GNU FNC NIISI RAN), Moscow, Russia}
\address{Institute for Theoretical and Experimental Physics, Moscow, Russia}
\address{National Research University Higher School of Economics, Russian Federation}
\address{Moscow Institute of Physics and Technology (MIPT), Dolgoprudnyi, Russia}
\address{Wu Key Lab, USTC}
\email{aslosev2 @gmail.com}

\author{Pavel Mnev}
\address{
University of Notre Dame 
}
\address{
St. Petersburg Department of V. A. Steklov Institute of Mathematics of the Russian Academy of Sciences
}
\email{pmnev @nd.edu}

\author{Donald R. Youmans}
\address{
Universit\'e de Gen\`eve
}
\email{Donald.Youmans @unige.ch}

\thanks{The work of A. L. was accomplished in GNU
FNC NIISI RAN program No.6, theme 36.20,
and was partially supported by Laboratory of Mirror Symmetry NRU HSE, RF Government grant, ag. no. 14.641.31.0001.
 P. M. acknowledges partial support of RFBR Grant No. 17-01-00283a. Research of D. Y. was supported by the Grant 178794 
 and the NCCR SwissMAP of the Swiss National Science Foundation.}

\date{\today}
\maketitle

\setcounter{tocdepth}{3} 
\tableofcontents

\section*{Introduction and outline}

A two-dimensional conformal field theory is called topological if it contains an odd symmetry $Q$ satisfying $Q^2=0$ and such that the stress-energy tensor is $Q$-exact: $$T=Q(G),\; \bar{T}=Q(\bar{G})$$
Given a topological conformal field theory, one can consider so-called  ``coupling to topological gravity.'' This amounts to considering  correlators of fields $G$ and $\bar{G}$ (in the presence of vertex operators) as differential forms on moduli spaces of complex structures on surfaces with marked points. Periods of these differential forms are called the (generalized) amplitudes in topological string theory (generalized Gromov-Witten invariants), see \cite{WittenCSstring}.

In the study of TCFTs, mostly A-twisted and B-twisted (2,2)-superconformal field theories were considered in the literature.
In the  end of 1980s another class of topological theories were studied, coming from gauge-fixing of a gauge theory with topological (diffeomorphism-invariant) action -- for instance, Chern-Simons theory.
One may expect that in two dimensions, in a proper gauge-fixing (like Lorenz gauge), these topological theories would be also conformal.
In fact, in our previous work \cite{LMY} 
 we showed that the abelian $BF$ theory in Lorenz gauge is a type
 B-twisted (2,2)-superconformal theory with target being an odd complex plane ($\Pi \mathbb{C}$ or $\mathbb{C}[1]$).  

This work is devoted to the study of 
the two-dimensional non-abelian $BF$ theory in Lorenz gauge as a topological confromal field theory.

In Section \ref{s: classical} we consider the two-dimensional $BF$ theory for an arbitrary Lie algebra $\g$, with fields being a $\g$-valued one-form $A$ with curvature $F$  and a $\g^*$-valued function $B$.
We start by considering the classical action that appears after imposing  Lorenz gauge. In 
the gauge-fixed theory in dimension two we  have conformal invariance
on the classical level. Since the metric enters the action only through gauge-fixing, we find that the stress-energy tensor is classically $Q$-exact.
As it is clear from the form of the action, the non-abelian deformation violates accidental symmetries of the abelian theory. The only conserved currents in the deformed theory are:
the holomorphic piece of the stress-energy tensor $T$, its complex conjugate $\bar{T}$, the superpartner $G$ of $T$ and its complex conjugate $\bar{G}$ and the total BRST current $J$ (that is a sum of $J^{(1,0)}$
and $J^{(0,1)}$ pieces that are not conserved separately).

The characteristic feature of $BF$ theories in any dimension is the upper-triangular structure of the interaction. Thus, we expect to
get only tree level and one loop contributions in Feynman diagrams. In Section \ref{s: correlators} we show that this property is preserved by Lorenz gauge-fixing.
To our surprise, we find that for reductive Lie algebras (the exact condition is written in Subsection \ref{sec: general correlators and admissible Feynman graphs}) the one-loop contribution vanishes
due to cancellation between ghost and gauge fields. So, unexpectedly, on the level of correlators of fundamental fields, the theory
is classical and hence finite.\footnote{Therefore, it would be interesting to relate this theory to the instantonic theory (in the sense of \cite{FLN}) for instantonic
equations 
$
dA+A^2=0, \;\; d^*A=0
$
but we will leave this for further studies.}
Thus we conclude 
that the theory is conformal (since it does not need to have ultraviolet regularization and renormalization). We proceed by computing
simplest correlators on the complex plane. Here we meet another surprise -- the correlators involve logarithms, dilogarithms and so on. Thus, in this section we start to get evidence that the theory is logarithmic 
-- this will be confirmed in Section \ref{ss: vertex operators}. It would be interesting to compare this with the logarithmic theories arising as instantonic theories in \cite{FLN}.
We conclude this section by describing soaking observables (delta-functions of scalar fields) that allow one to pass from the plane to the sphere.
Note that Witten in \cite{Witten2dYM} had a different way to deal with the zero-modes of the  field $B$.
The insertion of delta-functions of scalar fields can be interpreted in terms of a modification of the moduli space of flat connections -- we are planning to return to this question in the nearest future.

In Section \ref{ss: OPE} we compute OPEs of fundamental fields and observe unusual coefficient functions like  $\log|z-w|$ and
 $\frac{\bar{z}-\bar{w}}{z-w}$. We think that such coefficient functions are characteristic features of a logarithmic conformal theory. We also find that not only correlators have finitely contributing diagrams,
 but also
  OPEs, which is a much stronger statement.

Since we would like to study correlators of $T$, $G$ and $J$ and they are composite fields, we extend our considerations  to composite fields in Section \ref{ss: composite fields}. We start by defining the composite field as a result of consecutive mergings of
fundamental fields accompanied by subtraction of singular parts. In this way the composite field depends on the order of mergings. Moreover,
we can define in a similar way the bilinear product of composite fields -- the result of merging of two composite fields accompanied by subtraction of the singular piece. Here we have an open question -- does this product satisfy the pre-Lie algebra identity (\ref{pre-Lie identity})?
Proceeding to the  fields $T,G,J$, we are surprised to find that these fields are independent of the order of merging and have zero singular subtractions. It would be interesting to understand this
more conceptually, not merely as a result of a long computation.
In this section we also present examples of correlators containing higher powers of logs.

Despite the fact that $T,G,J$ are independent of the order of merging and contain no singular subtractions, it is a priori not clear why they are conserved and why their OPE is the standard one.
To leave no doubt, we prove these properties directly in Section \ref{s: conf and Q-invariance}.
It would be 
interesting to understand if these properties of $T,G,J$ could be deduced from the cancellations of singular subtractions we found in Section \ref{ss: composite fields}.

In Section \ref{ss: vertex operators} we compute conformal dimensions of some composite fields.
We find that there are fields $V^{(n)}$ (\ref{V^n}), for $n\geq 0$, with logarithmic anomalous dimensions in the sense that
$$L_0 V^{(n)}=V^{(n)} + g V^{(n-1)}\quad, \quad \bar{L}_0 V^{(n)}= g V^{(n-1)}$$ where $g$ is the non-abelian coupling constant.
This confirms the logarithmic nature of the theory.
Moreover, we can build a vertex operator 
$$
V=\sum_{n}  V^{(n)}
$$
with anomalous dimension $(1+ g, g)$.
Like in the case of free scalar theory, the origin of the anomalous dimension may be explained 
as arising from the dependence of the singular subtraction on the local coordinate. 

In conclusion we should mention that we have constructed and studied a novel class of topological 
logarithmic conformal field theories. Our next step would be the construction
 of topological string amplitudes in such theories. We plan to do that in the nearest future.

\subsection*{Acknowledgements}
This work originated from discussions of A. L. with Nikita Nekrasov more than 20 years ago.
We thank Anton Alekseev, Nikita Nekrasov, Brant Pym and Konstantin Wernli for insightful comments.

\section{Classical non-abelian theory}\label{s: classical}
In this section we discuss the classical two-dimensional non-abelian $BF$ theory, paralleling the treatment of the abelian case in \cite{LMY}.

Fix a finite-dimensional 
Lie algebra $\g$.\footnote{ 
When discussing quantization, we will need to assume  that $\g$ is \textit{strongly unimodular}, see (\ref{trace identity}).  In particular, this assumption
holds for all semisimple and nilpotent Lie algebras, or sums of those.
}

We consider the non-abelian $BF$ theory on the complex plane $\CC$,\footnote{
Throughout this section 
we can everywhere replace $\CC$ by any surface $\Sigma$ equipped with a metric (needed for the gauge-fixing). We specialize to $\CC$ right away, as it will be the case of relevance in the discussion of 
quantization.
} defined classically by the action
\begin{equation}\label{S^cl}
S^\cl=\int_\CC \lan B, dA + \frac{g}{2}[A,A] \ran
\end{equation}
Here the classical fields are: a $\g$-valued $1$-form $A$ on $\CC$ and a $\g^*$-valued $0$-form $B$; $d$ is the de Rham operator; $\lan-,-\ran$ is the canonical pairing between the coefficient Lie algebra $\g$ and its dual $\g^*$; $g$ is a coupling constant (deformation parameter corresponding to the deformation of the abelian theory into the non-abelian one).

The equations of motion are: 
$$dA+\frac{g}{2}[A,A]=0\; \mbox{(flatness of $A$ as a connection),}\qquad  dB+g[A,B]=0$$  
Here $[A,B]=\mr{ad}^*_A(B)$ is a notation for the coadjoint action of $\g$ on $\g^*$; it is consistent with the case when $\g^*$ is identified with $\g$ via non-degenerate Killing form. The action $S^\cl$ is invariant  under infinitesimal gauge transformations 
$$A\mapsto A+d\alpha+g[A,\alpha],\qquad B\mapsto B+g[B,\alpha]$$ 
with generator $\alpha$ a $\g$-valued $0$-form.

\subsection{Gauge-fixing in BRST formalism}
We consider the non-abelian $BF$ theory in \textit{Lorenz gauge} $d*A=0$, with $*$ the Hodge star on $\CC$. The corresponding Faddeev-Popov gauge-fixed action is:
\begin{equation}\label{S}
S=\int_\CC \lan B, dA + \frac{g}{2}[A,A] \ran + \lan \lambda, d*A \ran + \lan b, d*d_A c \ran
\end{equation}
Here $\lambda$ is the Lagrange multiplier imposing the gauge condition and $b,c$ are Faddeev-Popov ghosts (anti-commuting scalar fields); $d_A=d+g[A,-]$ is the de Rham operator twisted by $A$. Action $S$ is a function on the space of BRST fields:
$$ \cF= \underbrace{\Omega^1(\g)}_A\oplus \underbrace{\Omega^0(\g^*)}_B\oplus \underbrace{\Omega^0(\g^*)}_\lambda\oplus \underbrace{\Omega^0(\g^*)[-1]}_b\oplus \underbrace{\Omega^0(\g)[1]}_c $$
where $\Omega^p(\cdots)$ stands for the space of $p$-forms on $\CC$ with coefficients in $\cdots$; $[\pm1]$ are homological degree shifts and correspond to assigning ghost degree $-1$ to $b$ and $+1$ to $c$. The BRST operator acts as follows:
\begin{equation}\label{Q}
Q:\quad A\mapsto d_A c,\quad B\mapsto g[c,B],\quad c\mapsto \frac{g}{2}[c,c] , \quad b\mapsto \lambda,\quad \lambda \mapsto 0
\end{equation}
Action (\ref{S}) is a shift of the classical action (\ref{S^cl}) by a $Q$-coboundary:
\begin{equation}\label{S=S^cl+Q Psi}
S=S^\cl+Q(\Psi)
\end{equation}
with 
$\Psi=\int_\CC \lan b, d*A\ran$ the gauge-fixing fermion. Euler-Lagrange equations for the action (\ref{S}) read:
\begin{multline}\label{eom in real fields}
dA+\frac{g}{2}[A,A]=0, \quad d*A=0,\quad d*d_A c=0,\quad d_A*d b=0,\\
d_A B-*d\lambda-g[c,*db]=0
\end{multline}

\begin{remark}[Superfields]
One can combine the fields $A,B,c,b$ into two \textit{superfields} (or, more precisely, ``gauge-fixed AKSZ superfields'') valued in non-homogeneous forms:
\begin{equation}\label{superfields}
\mc{A}=c+A,\qquad \mc{B}=B-*db
\end{equation}
Written in terms of superfields $\mc{A},\mc{B}$ and the Lagrange multiplier $\lambda$, the action (\ref{S}) is:
\begin{equation}
S=\int_\CC \lan \mc{B}, d\mc{A}+\frac{g}{2}[\mc{A},\mc{A}] \ran+\lan \lambda, d*\mc{A} \ran
\end{equation}
Here we are integrating the $2$-form component of the integrand. The integrand above differs from the integrand of (\ref{S}) by a total derivative $d(\cdots)$ which is inconsequential.
\end{remark}

\subsection{Complex fields} Let $x^1,x^2$ be the real coordinates on $\CC\sim \RR^2$ and $z=x^1+i x^2$ the complex coordinate. 

We split the $1$-form field $A$ into $(1,0)$ and $(0,1)$-form components: $A=dz\, a+d\bar{z}\, \bar{a}$ where 
$a,\bar{a}$ are $\g$-valued scalars. Also, we combine the field $B$ and the Lagrange multiplier $\lambda$ into a $\g^*$-valued complex scalar field $\gamma=\frac12 (\lambda+iB)$ and its complex conjugate $\bar\gamma =\frac12 (\lambda-iB)$.

Written in terms of fields $(a,\bar{a},\gamma,\bar\gamma,b,c)$, the action (\ref{S}) takes the following form:
\begin{multline}\label{S nonab in cx fields}
S=4\int_\CC d^2 x \Big(\lan \gamma, \db a \ran + \lan \bar\gamma,\dd \bar{a} \ran + \lan b, \dd \db c \ran -\\
-\frac{g}{2} \lan \gamma-\bar\gamma,[a,\bar{a}]\ran-\frac{g}{2}\lan \dd b, [\bar{a},c] \ran-\frac{g}{2} \lan \db b,  [a,c] \ran \Big)
\end{multline}
Here $d^2x =dx^1 dx^2=\frac{i}{2}dz\, d\bar{z}$ is the standard area form on $\CC$ and $\dd=\frac{\dd}{\dd z}$, $\db=\frac{\dd}{\dd \bar{z}}$ are the partial derivatives (not the holomorphic/anti-holomorphic Dolbeaux operators: we do not include $dz,d\bar{z}$ in $\dd,\db$ in our notations). 

Equations of motion (\ref{eom in real fields}) written in complex fields take the form
\begin{equation}\label{eom in cx fields}
\begin{aligned}
\db a-\frac{g}{2} [a,\bar{a}]&=0,& \dd \bar{a}+\frac{g}{2} [a,\bar{a}]&=0,\\
\db \gamma+\frac{g}{2} [\bar{a},\gamma-\bar\gamma]-\frac{g}{2} [c,\db b] &=0, & 
\dd \bar\gamma-\frac{g}{2} [a,\gamma-\bar\gamma]-\frac{g}{2} [c,\dd b] &=0,\\
\dd \db b+\frac{g}{2}[a,\db b]+\frac{g}{2}[\bar{a},\dd b]&=0,& \dd \db c+\frac{g}{2} \db [a,c]+\frac{g}{2} \dd [\bar{a},c]&=0
\end{aligned}
\end{equation}
Finally, the BRST operator $Q$ becomes the following:
\begin{equation}\label{Q in cx fields}
Q:\qquad 
\begin{aligned}
a&\mapsto -\dd c - g[a,c], & \bar{a} &\mapsto -\db c-g[\bar{a},c], \\
\gamma&\mapsto \frac{g}{2}[c,\gamma-\bar\gamma],& \bar\gamma&\mapsto -\frac{g}{2}[c,\gamma-\bar\gamma], \\
b&\mapsto \gamma+\bar\gamma,& c&\mapsto \frac{g}{2}[c,c]
\end{aligned}
\end{equation}

\subsection{BRST current}
The Noether current associated to BRST symmetry 
is
\begin{equation}\label{J^tot}
J^\mr{tot} 
=-2i(dz\, J - d\bar{z}\, \bar{J})
\end{equation}
where 
\begin{equation}\label{J}
\begin{aligned}
J&=\lan \gamma, \dd c \ran + g \lan \gamma,[a,c] \ran-\frac{g}{4} \lan \dd b,[c,c] \ran,\\
\bar{J}&=\lan \bar\gamma, \db c \ran + g \lan \bar\gamma,[\bar{a},c] \ran-\frac{g}{4} \lan \db b,[c,c] \ran
\end{aligned}
\end{equation}
The current $J^\mr{tot}$ is conserved: 
\begin{equation}\label{d J^tot = 0}
dJ^\mr{tot}\simeom 0
\end{equation} 
(where $\simeom$ means equivalence \textit{modulo equations of motion}). 

\textbf{Warning:} The 
$(1,0)$- and $(0,1)$-form components $J,\bar{J}$ of the current  are not conserved separately (unlike in abelian $BF$ theory): $\db J\not\sim 0$, $\dd \bar{J}\not\sim 0$.

In terms of real fields, the BRST current spells
$$ J^\mr{tot}= \lan B,d_A c \ran+ \lan \lambda,*d_A c \ran-\frac{g}{2}\lan *db,[c,c] \ran $$

\subsection{Classical conformal invariance and the $Q$-exact stress-energy tensor}
Action (\ref{S}) can be considered on any surface $\Sigma$ endowed with a Riemannian metric $\xi$, which enters the integrand via the Hodge star. Let us denote the action on $(\Sigma,\xi)$ by $S_{\Sigma,\xi}$. Since Hodge star acts in (\ref{S}) only on $1$-forms, 
the action $S_{\Sigma,\xi}$ is invariant under Weyl transformations -- rescaling of $\xi$ by a positive function on $\Sigma$. 
Thus, the action $S_{\Sigma,\xi}$ depends only on the conformal class of the metric.

One defines the stress-energy tensor $T^\tot=T_{\mu\nu}dx^\mu dx^\nu$ (a field-dependent section of the symmetric square of the cotangent bundle $\mr{Sym}^2 T^*\Sigma$) as the reaction of $S_{\Sigma,\xi}$ to an infinitesimal change of metric $\xi$. More precisely, $T^\tot$ is defined via
\begin{equation*}
\delta_\xi S_{\Sigma,\xi}=-\int_\Sigma \sqrt{\det\xi}\, d^2 x\, T_{\mu\nu}\,\delta \xi^{\mu\nu}
\end{equation*}
Here $x^1,x^2$ are local coordinates on $\Sigma$, $\xi^{\mu\nu}$ are the components of the inverse metric $\xi^{-1}$; $\delta_\xi$  stands for the variation w.r.t. a change of metric. 

Since the action can be written as $S^\cl+Q(\Psi_\xi)$ where $S^\cl$ and $Q$ are manifestly independent of the metric and only the gauge-fixing fermion $\Psi_\xi=\int \lan b,d*_\xi A \ran$ is metric-dependent (via the Hodge star), the stress-energy tensor is $Q$-exact: 
\begin{equation}\label{T=Q(G)}
T^\tot=Q(G^\tot)
\end{equation}
The primitive $G^\tot=G_{\mu\nu}dx^\mu dx^\nu$ is defined in terms of the variation of the gauge-fixing fermion w.r.t. the metric: $\delta_\xi \Psi_{\xi}=-\int_\Sigma \sqrt{\det\xi}\, d^2 x\, G_{\mu\nu}\,\delta \xi^{\mu\nu}$. The explicit calculation is the same as in the abelian theory \cite{LMY} (since $\Psi$ does not depend on the deformation parameter $g$) and yields the result
\begin{equation}\label{G^tot}
G^\tot=(dz)^2\underbrace{\lan a , \dd b \ran}_{G} + (d\bar{z})^2 \underbrace{\lan \bar{a},\db b \ran}_{\bar{G}}
\end{equation}
-- This result is valid for an arbitrary surface $\Sigma$, with $z,\bar{z}$  local complex coordinates (compatible with the conformal class of a given metric $\xi$), and, as a special case, for $\Sigma=\CC$ with standard metric and $z$ the global complex coordinate.

Next, we calculate the stress-energy tensor from (\ref{T=Q(G)}) and the explicit formula (\ref{G^tot}) for $G^\tot$:
\begin{equation}
T^\tot=(dz)^2 T+ (d\bar{z})^2 \bar{T}
\end{equation}
where the holomorphic component is:
\begin{equation}\label{T explicit}
T=Q(G)=\lan \dd b, \dd c+g[a,c] \ran + \lan a,\dd(\gamma+\bar\gamma) \ran\
\end{equation}
Modulo equations of motion, one can simplify it to an equivalent form 
\begin{equation}\label{T equivalent}
T\simeom \lan \dd b, \dd c \ran + \lan a,\dd\gamma\ran +\frac{g}{2}\lan \dd b,[a,c] \ran
\end{equation}
The anti-holomorphic component $\bar{T}$ of the stress-energy tensor is given by the complex conjugate of (\ref{T explicit}), (\ref{T equivalent}).  Note that the stress-energy tensor $T^\tot$ does not have a $dz \,d\bar{z}$ component, which is tantamount to conformal invariance of the action.

The
components of the stress-energy tensor and its primitive are holomorphic/anti-holomorphic
modulo equations of motion:
\begin{equation}\label{G,T conservation}
\db G\simeom 0,\quad \dd \bar{G}\simeom 0,\qquad \db T\simeom 0,\quad \dd \bar{T}\simeom 0
\end{equation}

\subsection{Non-abelian theory as a deformation of abelian theory by a 2-observable}
Setting $g=0$ in all the formulae, we get the \textit{abelian} $BF$ theory (with coefficients in $\g$ viewed as a vector space, or, equivalently, $\dim\g$ non-interacting copies of abelian $BF$ theory with coefficients in $\RR$). We will indicate objects corresponding to the abelian theory by a subscript ``$0$'': we have the action $S_0$, BRST operator $Q_0$, stress-energy tensor $T_0$ etc. In particular, the abelian action 
\begin{equation}\label{S_0 via cx fields}
S_0=4\int d^2x \left(\lan a, \db \gamma \ran+ \lan \bar{a},\dd \bar\gamma \ran + \lan b,\dd \db c \ran\right)
\end{equation}
is a sum of three 
conformal field theories: a holomorphic and an anti-holomorphic first-order $\beta\gamma$-system and a second-order ghost system. 
The three constituent conformal theories are free (quadratic) and do not interact between each other on the level of the action but are intertwined by the BRST operator $Q_0$.

We say that a sequence $\OO^{(0)},\OO^{(1)},\OO^{(2)}$ of composite fields (local expressions in terms of fundamental fields of the theory and their derivatives of finite order at a given point), such that $\OO^{(p)}$ is valued in $p$-forms, forms a Witten's hierarchy of observables if Witten's descent equation 
$$Q \OO^{(p)}\simeom d\OO^{(p-1)}$$
is satisfied (modulo equations of motion) for $p=0,1,2$; we understand the equation for $p=0$ as $Q\OO^{(0)}=0$. We then say that $\OO^{(p)}$ is a ``$p$-observable'' (in the sense that its integral over a $p$-cycle is a gauge-invariant expression), and we say that  $\OO^{(p)}$ is obtained from  $\OO^{(p-1)}$ via descent.
%
%
Starting from a given $0$-observable $\OO^{(0)}$, one can solve the descent equation for $\OO^{(1)}$ and then for $\OO^{(2)}$ directly: one constructs $\OO^{(p)}$ by
using the operator product expansion of $\OO^{(p-1)}$ with fields $G,\bar{G}$ - components of the BRST primitive of the stress-energy tensor. Explicitly (see \cite{LMY} for details):
$$ \OO^{(p)}_w=-\frac{1}{p} \left(dw \oint_{C_w\ni z} \frac{dz}{2\pi i}G_z \OO^{(p-1)}_w - 
d\bar{w} \oint_{C_w\ni z} \frac{d\bar{z}}{2\pi i}\bar{G}_z \OO^{(p-1)}_w
\right) $$
where the integration is over a simple closed contour $C_w$, going around $w$ once counterclockwise; subscripts $z,w$ are the points where the fields are inserted. The equality is understood as an equality under the correlator with any number of test fields inserted outside the integration contour $C_w$.

Within the abelian theory, starting from a $Q_0$-closed observable  $\OO^{(0)}$, 
one is interested in constructing the corresponding descents $\OO^{(1)}$ and  $\OO^{(2)}$. Then one can deform the action of the abelian theory by
\begin{equation}\label{S_0 deformation to S_0+g O^2 }
S_0\mapsto S_0+g\int_\CC \OO^{(2)}
\end{equation}
with $g$ a deformation parameter.

The non-abelian deformation of the abelian theory corresponds to choosing
$$ \OO^{(0)}=\frac12 \lan B,[c,c] \ran $$
The corresponding first and second descent are:
\begin{align}
\OO^{(1)}&= \lan B,[A,c] \ran -\frac12 \lan *db,[c,c] \ran \\
\OO^{(2)} &= \frac12 \lan B,[A,A] \ran-\lan *db,[A,c] \ran
\end{align}
-- see the calculation in \cite{LMY}. Thus, the deformed action is 
\begin{equation}\label{S non-ab as S_0+g O^2}
S=
\underbrace{\int_\CC \lan B,dA \ran+ \lan \lambda, d*A \ran+\lan b,d*d c \ran}_{S_0} + \underbrace{\frac{g}{2}\lan B,[A,A] \ran -g \lan *db,[A,c] \ran }_{g\, \OO^{(2)}}
\end{equation}
It coincides with the gauge-fixed non-abelian action (\ref{S}). 

We also note that the hierarchy of observables $\OO^{(0)},\OO^{(1)},\OO^{(2)}$ are in fact homogeneous components of form degree $0,1,2$ of the expression $\frac12 \lan \mc{B},[\mc{A},\mc{A}] \ran$ written in terms of superfields (\ref{superfields}).

Non-abelian action (\ref{S non-ab as S_0+g O^2}) is a deformation of the free conformal field theory defined by the abelian action $S_0$ (\ref{S_0 via cx fields}) by a $2$-observable, which is in fact an \textit{exact  marginal operator}, i.e.: 
\begin{itemize}
\item $\OO^{(2)}$ is a primary field of conformal dimension $(1,1)$, 
\item  the operator product expansion $\OO^{(2)}_z \OO^{(2)}_w$ does not contain the singularity $\frac{1}{|z-w|^2}$ (which would destroy the conformal invariance of the deformed theory on the quantum level).
\end{itemize}
See Section \ref{sss: ab OPEs with O^2} below for the check of these properties; in particular, the first property relies on unimodularity of $\g$. 

Alongside the deformation  (\ref{S_0 deformation to S_0+g O^2 }) of the action, the other relevant objects of the theory deform:
\begin{equation}\label{Q,T,J deformation}
Q_0 \mapsto Q=Q_0+g Q_1, \quad
T_0 \mapsto T=T_0+g T_1, \quad
J_0 \mapsto J=J_0+g J_1
\end{equation}
We read off the abelian part $(\cdots)_0$ and the deformation $(\cdots)_1$ (the subscript  corresponds to the order in Taylor expansion in $g$ of the object) as constant in $g$ and linear in $g$ terms in formulae (\ref{Q}), (\ref{T explicit}), (\ref{J}). Note that the BRST primitive of the stress-energy tensor is the one object which does not deform: $G_0=G=\lan a,\dd b \ran$.

\begin{remark} 
In the context of a general deformation of any $n$-dimensional gauge theory by an $n$-observable, $S_0\mapsto S=S_0+g\int \OO^{(n)}$,
Noether theorem gives the deformation 
of the BRST current in the form $J^\tot_0\mapsto J^\tot=J^\tot_0+g J^\tot_1+O(g^2)$ with 
\begin{equation}\label{J_1 = O^1+...}
J^\tot_1= \OO^{(n-1)}-\iota_{Q_0}\alpha_1
\end{equation}
Here $\alpha_1$ is the deformation of the Noether $1$-form $\alpha=\alpha_0+g\alpha_1$, viewed as a $1$-form on the space of fields valued in $(n-1)$-forms on the spacetime and defined from $\delta L=\sum_i EL_i \delta\phi^i+d\alpha$. Here $L$ is the Lagrangian density of the action, 
the summation is over species of fields $\phi^i$, $EL_i$ is the Euler-Lagrange equation arising from variation of the field $\phi^i$ and $\delta$ the de Rham operator on the space of fields (as opposed to $d$ -- the de Rham operator on the spacetime). 
In (\ref{J_1 = O^1+...}) the first term is the $(n-1)$-observable linked to the $n$-observable deforming the action by descent (in non-deformed theory): $Q_0 \OO^{(n)}=d \OO^{(n-1)}$. One finds (\ref{J_1 = O^1+...}) from the expression for the BRST current given by Noether theorem, 
$$J^\tot=\rho-\iota_Q\alpha$$ 
where $\rho$ is defined by $QL=d\rho$. Restricting to $O(g^1)$ terms in this formula, one finds $\rho_1=\OO^{(n-1)}+\iota_{Q_1}\alpha_0$ which leads to (\ref{J_1 = O^1+...}).
In our case -- non-abelian deformation (\ref{S non-ab as S_0+g O^2}) of 2D abelian $BF$ theory -- we have 
$$ \alpha=\underbrace{-\lan B,\delta A \ran-\lan \lambda,*\delta A \ran+\lan *db,\delta c \ran+\lan b,*d\delta c \ran}_{\alpha_0} \underbrace{- g \lan \delta b,[*A,c] \ran}_{g \alpha_1} $$
and formula (\ref{J_1 = O^1+...}) yields the $O(g^1)$ part of the current (\ref{J^tot}).
\end{remark}


\subsubsection{Symmetries not surviving in the deformed theory.}
A part of symmetries/conservation laws of the abelian theory gets destroyed by the non-abelian deformation. Most importantly, the left/right components $J,\bar{J}$ of the BRST current are conserved in the abelian theory but not in the deformed theory. In the abelian theory, conservation of $J,\bar{J}$ ultimately leads to the realization of abelian theory as a twisted $\mc{N}=(2,2)$ superconformal field theory \cite{LMY}, with $G,J,\bar{G},\bar{J}$ corresponding via type B-twist to the two pairs of supercurrents. This picture does not carry over to the deformed theory. In particular, 
abelian theory has the conserved $R$-symmetry current $\lan \gamma,a \ran$; in non-abelian theory this expression is not conserved (and does not correspond to a symmetry of the action).

In summary, we have the following table of conserved quantities on abelian vs. non-abelian side:

\noindent
\begin{tabular}{l||l|ll}
& abelian & non-abelian  & (notes) \\ \hline
stress-energy tensor & $T_0, \bar{T}_0$ & $T,\bar{T}$ & \\
BRST primitive for stress-energy & $G_0,\bar{G}_0$ & $G,\bar{G}$ & $G=G_0$ does not deform \\ 
BRST current &$J_0,\bar{J}_0$ & only $J^\tot$ \\
$R$-symmetry current & $\lan \gamma,a \ran$, $\lan \bar\gamma,\bar{a} \ran$ & --- &
\end{tabular}

Also, fields $a,\gamma,\dd b, \dd c$ are holomorphic in abelian theory (i.e. satisfy $\db (\cdots)\sim 0$ modulo equations of motion) but this property also does not carry over to the deformed theory (as one sees immediately from the equations of motion (\ref{eom in cx fields})).

\section{Correlators}\label{s: correlators}

In quantum theory, we are interested in the correlation functions of local fields $\Phi_1,\ldots,\Phi_n$ placed at points $z_1,\ldots,z_n\in \CC$. We assume the points to be pairwise distinct, $z_i\neq z_j$ for $i\neq j$. Such a correlator is formally defined by the path integral 
\begin{equation}\label{correlator PI}
\lan \Phi_1(z_1)\cdots \Phi_n(z_n) \ran = \frac{1}{Z}\int_\cF e^{-\frac{1}{4\pi}S}  \Phi_1(z_1)\cdots \Phi_n(z_n)
\end{equation}
where $Z$ is the normalization (partition function), such that $\lan 1 \ran=1$. Fields $\Phi_j$ are the fundamental fields of the theory $\gamma,\bar\gamma,a,\bar{a},b,c$ or their derivatives of arbitrary order. 
More generally, one can allow $\Phi_j$ to be a product of such objects -- a \textit{composite field}. For example, one can have $\Phi(z)=(\phi\psi)(z)$ with $\phi,\psi$ linear in fundamental fields. Such a product is understood as ``\textit{renormalized},''\footnote{Another possible term is the ``\textit{normally ordered}'' product. We do not use this term here as it is somewhat ambiguous in a non-free theory (our prescription has nothing to do with creation/annihilation operators).
} 
i.e., as a limit 
\begin{equation}\label{normal ordering}
\Phi(z)=\lim_{z'\ra z}\Big(\phi(z')\psi(z)-[\phi(z')\psi(z)]_\mr{sing}\Big)
\end{equation}
under the correlator with other fields; here the last term stands for the singular part, as $z'$ approaches $z$, of the operator product expansion $\phi(z')\psi(z)$, see Section \ref{ss: OPE} below. Furthermore, formula (\ref{normal ordering}) can be applied to (renormalized) composite fields $\phi,\psi$, to construct their renormalized product.
Thus, correlators of composite fields can be obtained from correlators of fundamental fields (or their derivatives), by merging some of the points $z_i$ (and subtracting the singular terms). We defer the detailed discussion of composite fields and the procedure of building them as renormalized products until Section \ref{sss: order of collapse ambiguity}.

In abelian $BF$ theory, one has correlators $\lan \Phi(z_1)\cdots \Phi(z_n) \ran_0$  defined as in (\ref{correlator PI}), but with $S$ replaced by the free action $S_0$. These free theory correlators are given by Wick's lemma with propagators
\begin{equation}\label{propagators}
\begin{gathered}
\lan c(w)\otimes b(z) \ran_0 =  \big(2\log|w-z| +C \big)\,  \idg ,\\
\lan a(w) \otimes \gamma(z) \ran_0 = \frac{\idg}{w-z}, \qquad 
\lan \bar{a}(w) \otimes \bar{\gamma}(z) \ran_0 = \frac{\idg}{\bar{w}-\bar{z}}
\end{gathered}
\end{equation}
These expressions are understood as valued in $\g \otimes\g^*$ 
(and tensor product sign in the correlator of $\g,\g^*$-valued fields refers to this);
 we denote by $\idg$  the identity in $\mr{End}(\g)\simeq \g\otimes\g^*$. Here $C$ is an undetermined constant.\footnote{
\label{footnote: C in bc propagator}
If the theory is regularized by an infrared cut-off, by imposing a Dirichlet boundary condition on $b,c$ on a circle of large radius $R$, then $C=-2\log R$.} Propagators for other pairs of fields from the list $\{a,\bar{a},\gamma,\bar\gamma,b,c\}$ are zero. Propagators (\ref{propagators}) are obtained as Green's functions for the operators $\dd \db, \db, \dd$ in quadratic action (\ref{S_0 via cx fields}). Parameterizing the space of fields by the superfields (\ref{superfields}) and the Lagrange multiplier $\lambda$, one has the following propagators:
\begin{equation}\label{propagators for superfields and lambda}
\lan \mc{A}(w)\otimes \mc{B}(z) \ran_0 =  2\,d \arg(w-z)\cdot\idg,\qquad \lan \mc{A}(w)\otimes\lambda(z) \ran_0 = 2\, d_w \log|w-z|\cdot\idg
\end{equation}
Here in the first formula, the propagator is understood as a $1$-form on the configuration space $\mr{Conf}_2(\CC)\subset \CC\times \CC$ of two distinct points $(w,z)$ on $\CC$; $d=d_w+d_z$ is the total de Rham operator on the configuration space, where $d_w,d_z$ are the de Rham operators on the first and second copy of $\CC$. 

Using $S=S_0+g\int \OO^{(2)}$ in (\ref{correlator PI}),
the correlator of the deformed theory is expressed in terms of correlators in the abelian theory with insertions of $N\geq 0$ copies of the deforming observable $\OO^{(2)}$ at points $u_1,\ldots,u_N$ integrated over $\CC$:
\begin{multline}\label{correlator - exp(O2) expansion}
\lan  \Phi_1(z_1)\cdots \Phi_n(z_n) \ran = \lan \Phi_1(z_1)\cdots \Phi_n(z_n)\;\; e^{-\frac{g}{4\pi}\int_{\CC\,\ni u} \OO^{(2)}(u)}\ran_0 \\
=\sum_{N\geq 0}\frac{1}{N!}\left(-\frac{g}{4\pi}\right)^N \int_{\CC^N\;\ni (u_1,\ldots,u_N)} 
\lan \Phi_1(z_1)\cdots \Phi_n(z_n)\;\; \OO^{(2)}(u_1)\cdots \OO^{(2)}(u_N) \ran_0
\end{multline}
Free theory correlator on the r.h.s. is evaluated using Wick's lemma, as a sum of Wick's contractions of fields $\Phi_1(z_1),\ldots,\Phi_n(z_n)$ and $N$ copies of $\OO^{(2)}$ and yields a sum of Feynman graphs with $N$ internal vertices  decorated by $\OO^{(2)}$  (their positions $u_i$ are integrated over $\CC$) and $n$ vertices decorated by $\Phi_1,\ldots,\Phi_n$ at fixed points $z_1,\ldots,z_n$. In the case when $\Phi_i$ are fundamental fields or their derivatives, the corresponding fixed vertices are uni-valent.

For Feynman graphs, we adopt the convention where the graphs are oriented, with half-edges decorated by fields $a,\bar{a},c$ (or their derivatives) oriented towards the incident vertex and $\gamma,\bar\gamma,b$ (or derivatives) oriented away from the vertex.
In particular, the interaction vertex (cf. the cubic part of (\ref{S nonab in cx fields})) is:
\begin{equation} \label{vertex}
\vcenter{\hbox{
\includegraphics[scale=0.4]{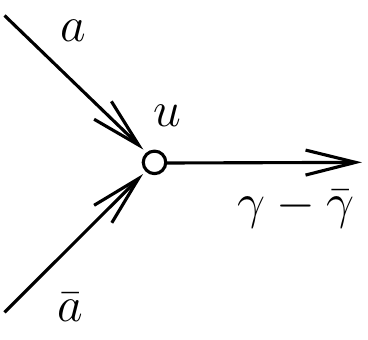}
}}\qquad + \qquad 
\vcenter{\hbox{
\includegraphics[scale=0.4]{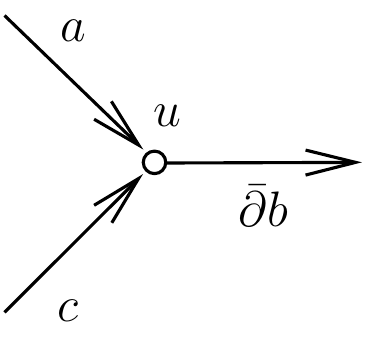}
}} \qquad + \qquad 
\vcenter{\hbox{\includegraphics[scale=0.4]{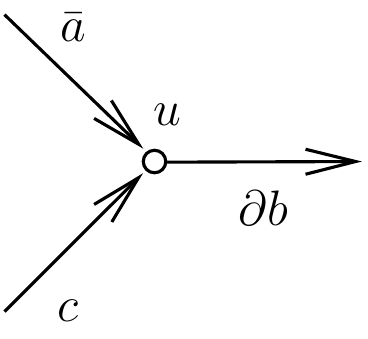}}}
\end{equation}
-- Here we list the possible decorations of half-edges by fields. The vertex is decorated by the expression  $g\,\mathbbm{f}\int_{\CC\,\ni u}\frac{d^2 u}{2\pi}$, where $\ff\in \g\otimes \g^*\otimes \g^*$ is the structure tensor of the Lie algebra $\g$.  In terms of the superfields $\mc{A},\mc{B}$, the vertex is simply 
\begin{equation} \label{vertex super}
\vcenter{\hbox{\includegraphics[scale=0.4]{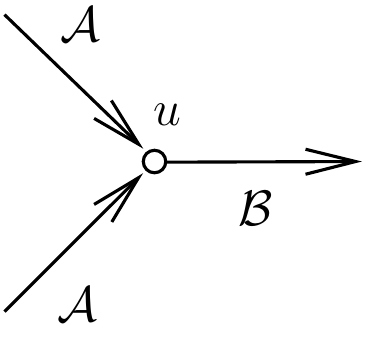}}} 
\end{equation}
and is  decorated by $\left(-\frac{g}{4\pi}\right)\,\frac12\,\ff\int_{\CC\,\ni u}$.

In our convention for Feynman graphs for the correlator (\ref{correlator - exp(O2) expansion}), we have two types of vertices: 
\begin{itemize}
\item black (fixed) vertices corresponding to fields $\Phi_1(z_1),\ldots,\Phi_n(z_n)$ we calculate the correlator of,
\item white (integrated) vertices corresponding to $\OO^{(2)}(u)$ where the point $u$ is integrated over.
\end{itemize}

\subsection{General correlators of fundamental fields: admissible Feynman diagrams}\label{sec: general correlators and admissible Feynman graphs}
Let us introduce a grading on fields -- the ``$\mc{AB}$-charge'' -- by assigning charge $+1$ to fields $\gamma,\bar\gamma,b$  and charge $-1$ to fields $a,\bar{a},c$. The convention is that the charge in unchanged when taking derivatives of a field  and  is additive under multiplication.\footnote{The name $\mc{AB}$-charge is due to the fact that all components of the superfield $\mc{A}$ have charge $-1$ and all components of $\mc{B}$ (plus the Lagrange multiplier $\lambda$) have charge $+1$.} Note that:
\begin{itemize}
\item A free theory correlator $\lan \cdots  \ran_0$ of a collection of fields can only be nonzero if the total charge of the fields vanishes. This is due to the form of propagators (\ref{propagators}) which only pair $+1$-fields to $-1$-fields (or, in other words, due to the fact that the abelian action (\ref{S_0 via cx fields}) has total charge zero).
\item The charge of the deforming observable $\OO^{(2)}$ is $-1$.
\end{itemize}
In particular, Feynman graphs contributing to the non-abelian correlator $\lan \Phi_1\cdots \Phi_n \ran$ must have exactly $N$ internal vertices, with $N$ the total charge of fields $\Phi_i$.  Thus, the correlator is proportional to $g^N$ with no other powers of $g$ present.

As follows from the form of the interaction vertex (\ref{vertex}), (\ref{vertex super}), with two incoming half-edges and one outgoing half-edge, Feynman graphs contributing to a correlation function $\lan \Phi_1 \cdots \Phi_n \ran$ of fundamental fields (or their derivatives) can have connected components of the following two types:
\begin{enumerate}[(i)]
\item \label{graph item: tree} Binary rooted trees with leaves (uni-valent vertices with outward orientation of the adjacent half-edge) decorated with fields $\gamma,\bar\gamma,b$  or their derivatives from the list $\{\Phi_i\}$ and the root (uni-valent vertex with inward orientation of the half-edge) decorated by $a,\bar{a},c$ or derivatives. 
For example:
$$ \vcenter{\hbox{\includegraphics[scale=0.5]{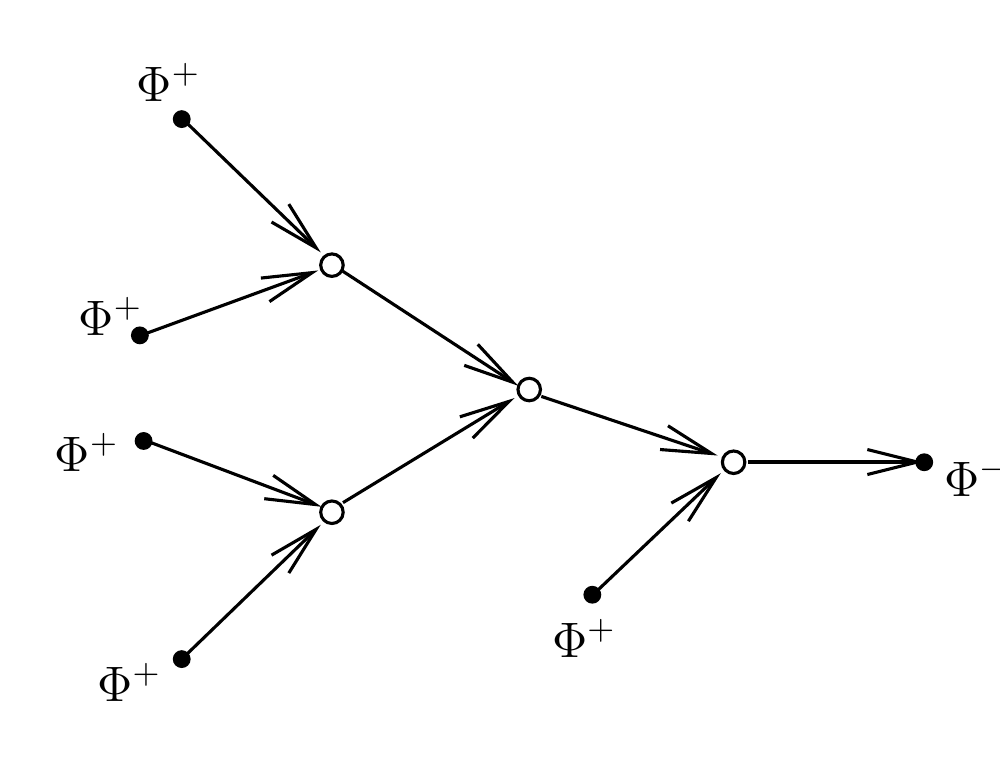}}}  $$
Here the superscript $\pm$ refers to fields of $\mc{AB}$-charge $\pm 1$.
\item \label{graph item: 1-loop} One-loop graphs, having the form of an oriented cycle with several binary trees rooted on the cycle (with leaves decorated by $\gamma,\bar\gamma,b$ or derivatives). 
For example:
$$ \vcenter{\hbox{\includegraphics[scale=0.5]{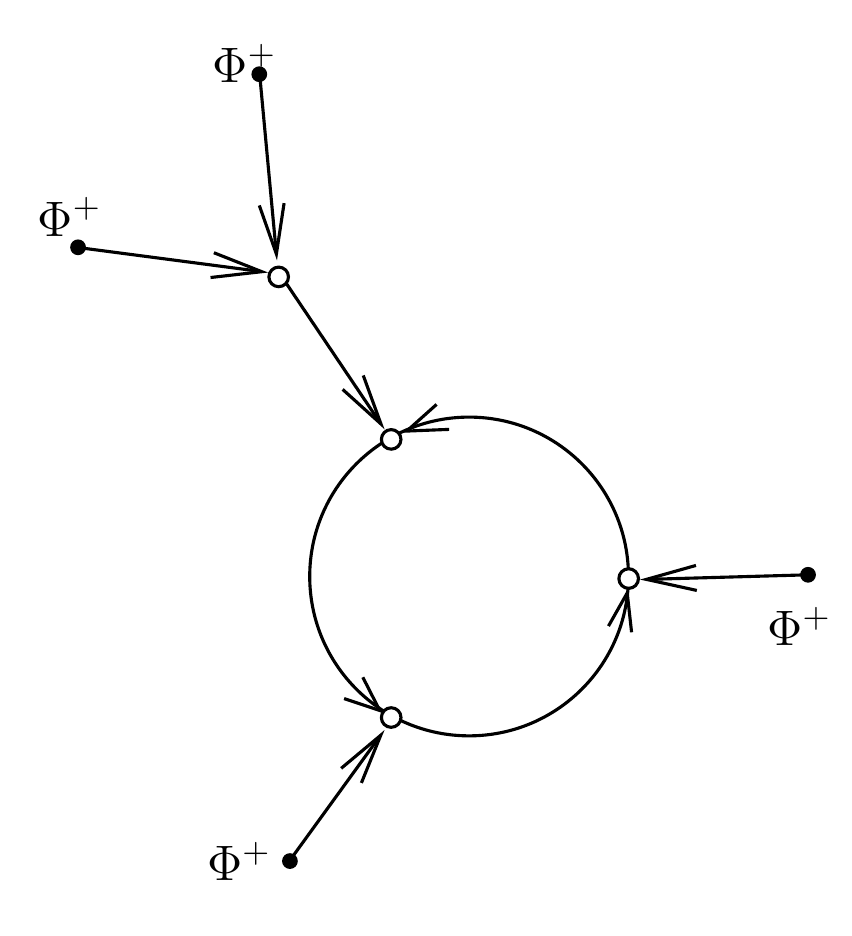}}}  $$
\end{enumerate}

We assume that Lie algebra $\g$ is such that one has the identity 
\begin{equation}\label{trace identity}
\mr{tr}_\g \left(\mr{ad}_{X_1}\cdots \mr{ad_{X_k}}\right) = (-1)^k \mr{tr}_\g \left(\mr{ad}_{X_k}\cdots \mr{ad_{X_1}}\right)
\end{equation}
for $X_1,\ldots,X_k\in \g$ arbitrary 
elements, for any $k\geq 1$.
This identity 
holds for the following classes of Lie algebras: 
\begin{itemize}
\item any semisimple $\g$,\footnote{Indeed, using the Killing form (which is nondegenerate due to semisimplicity) to identify $\g^*\simeq \g$, in the l.h.s. of (\ref{trace identity}) we have a trace of a product of $k$ \textit{anti-symmetric} matrices. Applying transposition under the trace, we get the r.h.s.
}
\item any nilpotent $\g$ (in a trivial way: the traces are zero),
\item a direct sum of a semisimple and a nilpotent Lie algebras, e.g., any reductive $\g$.
\end{itemize}
We call an algebra satisfying (\ref{trace identity}) \textit{strongly unimodular}, since $k=1$ case is equivalent to the usual unimodularity condition $\mr{tr}_\g [X,-]=0$.\footnote{
This condition appeared in \cite{AT} in the context of Kashiwara-Vergne problem.
}

\begin{lemma}[Boson-fermion cancellation in the loop]\label{lemma: boson-fermion cancellation}
Under assumption (\ref{trace identity}), 
graphs of type (\ref{graph item: 1-loop}) 
vanish, when summed over admissible decorations in the loop.
\end{lemma}


\begin{proof}
Given a one-loop graph $\Gamma$, there are two possible decorations of the half-edges in the loop -- by alternating fields $A$ and $B$ vs. by alternating $c$ and $-*db$ (for this argument, it is convenient to switch to real fields) -- and they give identical contributions of opposite sign:
\begin{equation}\label{loop cancellation}
\vcenter{\hbox{\includegraphics[scale=0.4]{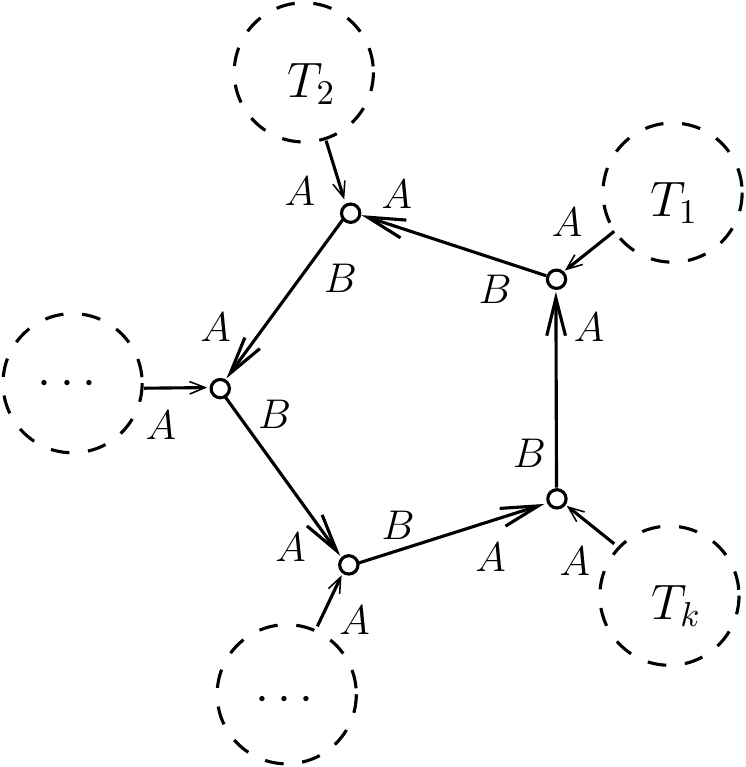}}}\qquad + \qquad
 \vcenter{\hbox{\includegraphics[scale=0.4]{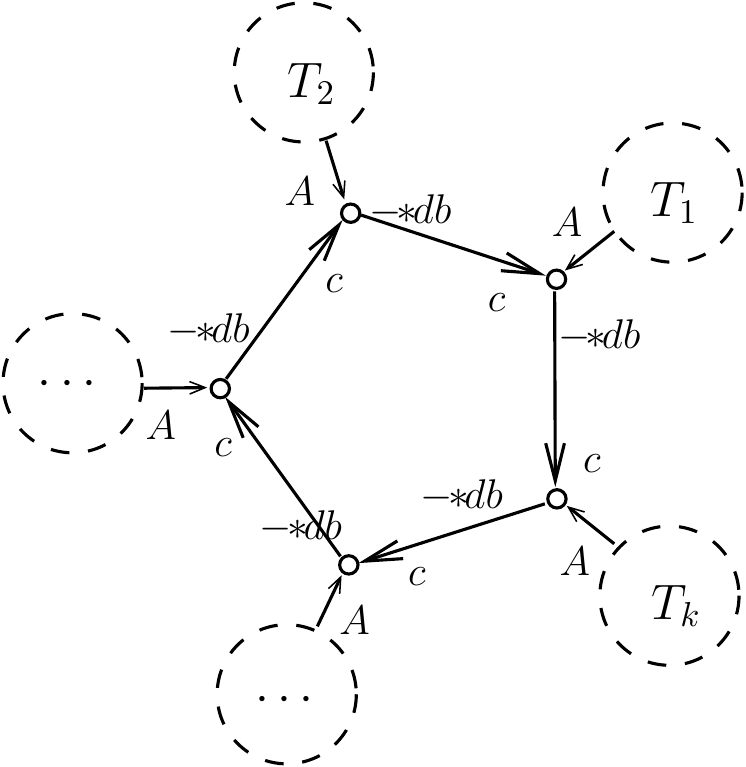}}} \qquad =0
\end{equation}
Here $T_1,\ldots,T_k$ are arbitrary trees rooted on the cycle; note that the 
orientation of the cycle
is switched between the two summands. 
%
%
%
%

To see the cancellation (\ref{loop cancellation}) explicitly, we observe that the first graph contains the expression
\begin{multline}\label{loop_boson}
\wick{
\langle \c3{B},[A,\c1{A}] \rangle_{u_k} \c1{ } \cdots \langle {B},[A,\c2{A}] \rangle_{u_2} \langle \c2{B},[A,\c3{A}] \rangle_{u_1}   
}
\\ = 
-   2^k \mr{tr}_\g \Big(d_1 \varphi_{1k} \  \mr{ad}_{A(u_k)} \  d_{k}\varphi_{k\,k-1}\  \mr{ad}_{A(u_{k-1})} \cdots d_{2} \varphi_{21}\ \mr{ad}_{A(u_1)} \Big)
\end{multline}
where remaining fields $A(u_i)$ are Wick-contracted with trees $T_1,\ldots, T_k$. Here $\varphi_{ij}=\arg(u_i-u_j)$ and $d_i$ is the de Rham differential in $u_i$. Likewise, the second graph in (\ref{loop cancellation}) contains the expression
\begin{multline}\label{loop_ghost}
\wick{
\langle -*d\c3{b},[A,\c1{c}] \rangle_{u_1}\langle -*d\c1{b},[A,c] \rangle_{u_2}  \cdots \c2{ } \langle -*d\c2{b},[A,\c3{c}] \rangle_{u_k}
}
\\
=-   2^k\mr{tr}_\g\Big(d_1\varphi_{1k}\ \mr{ad}_{A(u_1)}\ d_2\varphi_{21}  \ \mr{ad}_{A(u_2)}\cdots \mr{ad}_{A(u_k)}\  d_k\varphi_{k\,k-1} \Big)
\end{multline}
Using Lie algebra identity (\ref{trace identity}), one can see that expressions (\ref{loop_boson}) and (\ref{loop_ghost}) are the same, up to a minus sign.
\end{proof}


\begin{proposition}[\textit{Properties of correlators of fundamental fields}]\label{Prop: corr properties}
\mbox{}\\
A correlator $\lan \Phi_1\cdots \Phi_n \ran$ of fundamental fields or their derivatives satisfies the following properties:
\begin{enumerate}[(1)]
\item \label{corr properties (i)} It is given by finitely many diagrams $\Gamma$ which are unions of binary rooted trees $\Gamma=\sqcup_{j=1}^p T_j$. 
\item\label{corr properties (ii)} The number $N$ of interaction vertices (and thus the order of $\lan\cdots \ran$ in $g$)  equals the total $\mc{AB}$-charge of fields $\Phi_i$. 
\item\label{corr properties (iii)} The number $p$ of trees equals the number of $-1$-charged fields among $\{\Phi_i\}$. 
\item\label{corr properties (iv)} The contribution of each diagram is given by an integral over $\CC^N$ which is convergent if the field $b$ is always hit by derivatives in $\{\Phi_i\}$. 
\item\label{corr properties (v)} If the bare field $b$ occurs among $\{\Phi_i\}$, an infrared regularization may be necessary (see Remark \ref{rem: <gamma b c>} below). 
\end{enumerate}
\end{proposition}
\begin{proof}
Properties (\ref{corr properties (i)}--\ref{corr properties (iii)}) summarize the discussion above. We proceed to show the convergence properties (\ref{corr properties (iv)}--\ref{corr properties (v)}). 

First, consider the correlator $\lan \Phi_1(z_1)\cdots \Phi_n(z_n) \ran$ where all the fields $\Phi_i$ are from the list $\{a,\bar{a},\gamma,\bar\gamma,c, \dd b, \db b\}$ -- with no additional derivatives and no bare $b$ ghost. Corresponding Feynman diagrams are given by integrals of the type
\begin{equation} \label{Feynman diag convergence}
 \int_{\CC^N} d^2 u_1\cdots d^2 u_N \prod_{\{x,y\}\subset \{z_1,\ldots,z_n,u_1,\ldots,u_N\}} P(x,y) 
\end{equation}
where the product is over pairs of points corresponding to the edges of the Feynman graph $\Gamma$ and the propagator $P(x,y)$ is either $\frac{1}{x-y}$ or $\frac{1}{\bar{x}-\bar{y}}$, depending on which pair of fields are connected by the edge. We need to analyze the potential obstructions to convergence arising from a collision of $2$ or more points (ultraviolet problems) or from one or more points $u_i$ going to infinity (infrared problems). We have the following possibilities.
\begin{enumerate}[(a)]
\item \label{convergence item (a)} Collision 
of $r\geq 2$ interaction vertices. 
More precisely, consider the situation when points $u_{i_1},\ldots,u_{i_r}$ are at the distance between $C_1\epsilon$ and $C_2 \epsilon$ from each other, with $C_1<C_2$ some constants and $\epsilon$ arbitrarily small. The integrand of (\ref{Feynman diag convergence}) behaves at $\epsilon\ra 0$ as $O(\epsilon^{-(r-1)})$, since there are 
at most $(r-1)$ propagators connecting a pair of points from the set of $r$ colliding points, 
since the Feynman graph $\Gamma$ is a tree. Thus, fixing $u_{i_1}$ we have an integrable singularity for integration over $(u_{i_2},\ldots, u_{i_r})$ (which is a $2(r-1)$-fold integral). Therefore, there is no ultraviolet divergence in this case.
\item \label{convergence item (b)} Collision of $r\geq 1$ interaction vertices at $z_j$ -- the place of insertion of $\Phi_j$. I.e. we consider the situation where points $u_{i_1},\ldots,u_{i_r},z_j$ are at distance between $C_1\epsilon$ and $C_2\epsilon$ from each other. For the same reason as in (\ref{convergence item (a)}), there are at most $r$ propagators connecting pairs of points from the colliding set. Therefore, the integrand in (\ref{Feynman diag convergence}) behaves as $O(\epsilon^{-r})$ and we have an integrable singularity for integration over $u_{i_1},\ldots,u_{i_r}$. Thus, again we have no ultraviolet divergence.
\item \label{convergence item (c)} Situation where $r\geq 1$ points $u_{i_1},\ldots,u_{i_r}$ go to infinity. -- If these points are at a distance $> C_1 R$ from $z$'s, the rest of $u$'s, and from each other, the integrand of (\ref{Feynman diag convergence}) behaves as $O(\frac{1}{R^{2r+1}})$ 
at $R\ra \infty$, since there are $k\leq (r-1)$ propagators connecting pairs points in the set $\{u_{i_1},\ldots,u_{i_r}\}$ and $3r-2k$ propagators connecting points in this set to other points in the ``finite'' region (recall that the interaction vertices are trivalent), and thus overall $(3r-2k)+k\geq 2r+1$ propagators involving points $u_{i_1},\ldots,u_{i_r}$.
Thus, the $2r$-fold integral over $u_{i_1},\ldots,u_{i_r}$ is convergent and there is no infrared divergence.
\item One could have a potential mixed infrared/ultraviolet problem when several $u_i$'s collide in an $\epsilon$-neighborhood a large distance $R$ away from $z$'s and rest of $u$'s. This situation is treated by a combination of the arguments of (\ref{convergence item (a)}) and (\ref{convergence item (c)}) -- it also does not lead to a divergence.
\end{enumerate}

In the case of a correlator involving higher derivatives of fundamental fields, we simply take respective derivatives of the correlator of the fundamental fields, given by convergent integrals. 

In the case when bare ghost $b$ is present among $\Phi$'s, the potential ultraviolet problems become even milder (as the ghost propagator (\ref{propagators}) has just a $\log$ singularity, as opposed to a pole). However, the power counting in the case (\ref{convergence item (c)}) can fail, see an example in Remark \ref{rem: <gamma b c>}, thus such correlators may require an infrared regularization.
\end{proof}

\begin{corollary}
Since the theory is ultraviolet-finite and since the Lagrangian contains no dimensionful parameters, the theory is conformal.
\end{corollary}

\begin{remark}
In the case when fields $\Phi_1,\ldots,\Phi_n$ 
belong to the subset of
real fields $\{A,B,c,*db\}$ (but no $\lambda$ and no bare $b$ field),
the product of propagators (which are proportional to $d\arg (u_i-u_j)$) 
extends to a smooth form on the compactified configuration space of $n$ points, so the integral is automatically convergent, as in the case of perturbative Chern-Simons theory \cite{AS} and Poisson sigma model \cite{Kontsevich}. The argument we gave above is more general: it allows the $\lambda$ field (or equivalently, allows $\gamma$ and $\bar\gamma$ independently, not just in the combination $B=-i(\gamma-\bar\gamma)$).
\end{remark}

\subsubsection{Weights (naive conformal dimensions) of fields}\label{sss: conf dimension}
We assign the 
holomorphic/anti-holomorphic \textit{weight} $(h,\bar{h})$ to fields as follows: for $a$, we set $(h,\bar{h})=(1,0)$. For $\bar{a}$, we set $(h,\bar{h})=(0,1)$; for the remaining fundamental fields, $\gamma,\bar\gamma,b,c$, we set $(h,\bar{h})=(0,0)$.\footnote{
This assignment corresponds to $a\, dz$ being classically a $(1,0)$-form, $\bar{a}\, d\bar{z}$ being a $(0,1)$-form and $\gamma,\bar\gamma,b,c$ being scalars.
} Weight is additive with respect to multiplication of fields; applying $\dd$ to a field increases $h$ by $1$, while applying $\db$ increases $\bar{h}$ by $1$. These rules define the weight for any composite field.

These weights could be understood as the ``naive'' conformal dimensions of the fields. Later we will show that weights of  fundamental fields  coincide with their actual conformal dimensions -- see Section \ref{sss: T quantum, primary fields}. However, for composite fields there will be an interesting difference (see Section \ref{ss: vertex operators}). 

To summarize the various degrees of fields we introduced, we have the ghost degree, the $\mc{AB}$-charge and the weight. For fundamental fields they are as follows.

\begin{tabular}{l|cccccc}
& $a$ & $\bar{a}$ & $\gamma$ & $\bar\gamma$ & $b$ & $c$ \\ \hline
ghost degree & $0$ & $0$ & $0$ & $0$ & $-1$ & $1$ \\
$\mc{AB}$-charge & $-1$ & $-1$ & $1$ & $1$ & $1$ & $-1$ \\
weight $(h,\bar{h})$ & $(1,0)$ & $(0,1)$ & $(0,0)$ & $(0,0)$ & $(0,0)$ & $(0,0)$
\end{tabular}

\subsection{Example: 3-point function of fundamental fields}
Consider the $3$-point correlation function 
\begin{equation}\label{3-point fun}
\lan \gamma(z_1)\otimes \bar\gamma(z_2)\otimes a(z_3) \ran
\end{equation}
We have the following diagram:
$$\vcenter{\hbox{\includegraphics[scale=0.35]{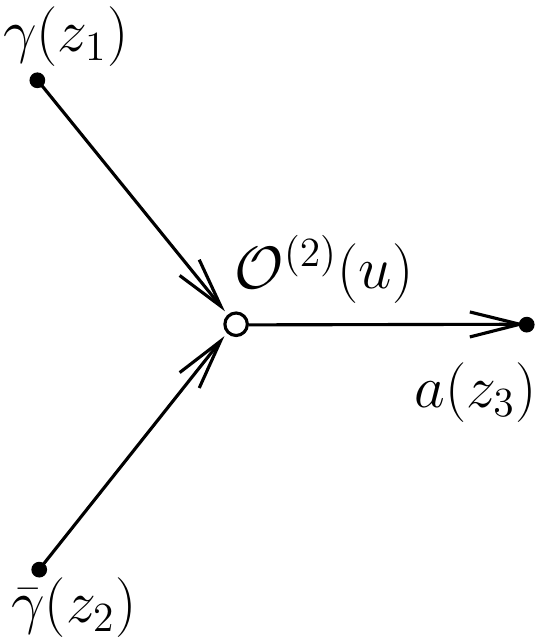}}} \qquad\qquad =\qquad 
\vcenter{\hbox{\includegraphics[scale=0.35]{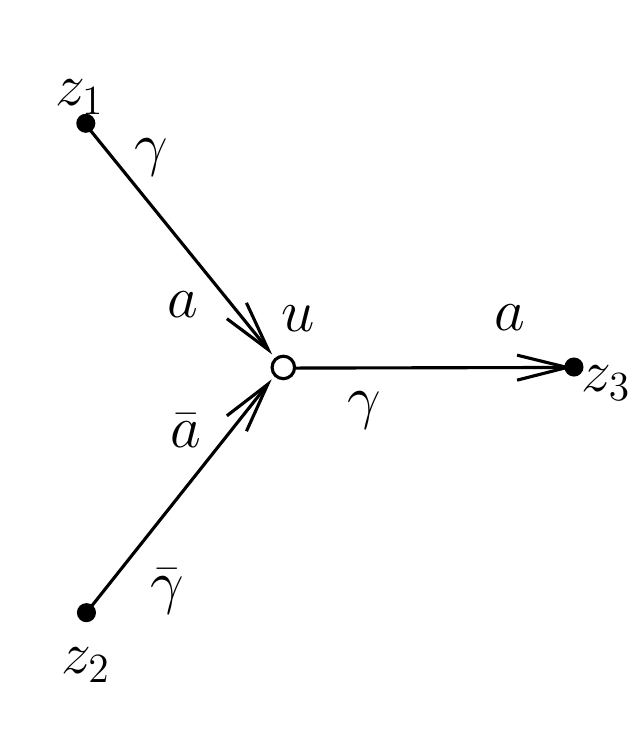}}}$$
The corresponding contribution to the correlator is the integral over $u$ (the place of insertion of the deforming observable $\OO^{(2)}$) of the product of three propagators:\footnote{
We originally introduced the structure tensor $\ff$ as an element of $\g\otimes\g^*\otimes\g^*$. By an abuse of notations, we use the same letter $\ff$ for the structure tensor with cyclically permuted $\g,\g^*,\g^*$ factors -- in this case (as is clear from the correlator we are looking at), $\ff$ is regarded as an element of $\g^*\otimes\g^*\otimes\g$.
}
\begin{multline}
\Big \langle \wick{\c1{\gamma(z_1)} \otimes \c3{\bar\gamma(z_2)} \otimes \c2{a(z_3)}\; g\int\frac{d^2 u}{2\pi}  \langle [\c1{a(u)},   \c3{\bar{a}(u)}], \c2{\gamma(u)}}\rangle \Big\rangle_0
\\  =
g\, \ff\cdot \int_{\CC\,\ni u}\frac{d^2 u}{2\pi} \frac{1}{(u-z_1)(\bar{u}-\bar{z}_2)(z_3-u)}
\end{multline}
In fact, the diagram above is the only contribution to the correlator (\ref{3-point fun}) -- the total $\mc{AB}$-charge of the fields $\gamma,\bar\gamma,a$ is $+1$ and thus a contributing diagram has to be a tree with a single interaction vertex.
Thus, evaluating the integral above (see (\ref{3-point integral})) we get the 
explicit result for the correlator:
\begin{equation}\label{3p_1}
\lan \gamma(z_1) \otimes \bar\gamma(z_2)\otimes a(z_3) \ran = g\, \ff\cdot \frac{1}{z_1-z_3}\log\left|\frac{z_1-z_2}{z_3-z_2}\right| 
\end{equation}
\begin{remark} The appearance of logs in correlators indicate that we are dealing with a logarithmic CFT \cite{Gurarie}, see Section \ref{ss: vertex operators}.
\end{remark}

By a similar calculation to (\ref{3p_1}), one finds
\begin{align*}
\lan \gamma(z_1)\otimes \dd b(z_2)\otimes c(z_3) \ran =& -  g\, \ff\cdot \frac{1}{z_1-z_2}\log\left|\frac{z_1-z_3}{z_2-z_3}\right| \\
\lan \bar\gamma(z_1)\otimes \dd b(z_2)\otimes c(z_3) \ran =& -  g\, \ff\cdot \frac{1}{z_2-z_3}\log\left|\frac{z_2-z_1}{z_3-z_1}\right|
\end{align*}

These $3$-point functions and their complex conjugates exhaust the nonvanishing $3$-point functions $\lan \phi_1(z_1)\phi_2(z_2) \phi_3(z_3)\ran$, with fields $\phi$ in the list $\{a,\bar{a},\gamma,\bar\gamma,c,\dd b,\db b\}$. 
By taking derivatives of these answers, one obtains $3$-point functions of arbitrary dervatives of the fundamental fields.

\begin{remark}\label{rem: <gamma b c>}
Note that here we did not consider $3$-point functions involving the ghost $b$ not hit by derivatives -- such correlators 
are given by more involved integrals of dilogarithmic type, which  contain an infrared divergence at $u\ra \infty$.
For instance:
\begin{equation}\label{3-point gamma b c}
\lan  \gamma(z_1)\otimes  b(z_2)\otimes c(z_3) \ran = g\, \ff\cdot \int_{\CC\,\ni u}
\frac{d^2 u}{2\pi} \;\;\frac{2\log|u-z_2|+C}{(u-z_1)(\bar{u}-\bar{z}_3)}
\end{equation}
One needs an infrared regularization, e.g. by 
restricting the integration domain to a disk of large radius $R$,
to have a convergent integral. Note that the constant $C$ in the $bc$ propagator (\ref{propagators}) also depends on the infrared regularization (see footnote \ref{footnote: C in bc propagator}). 
At $R\ra \infty$, the correlator  (\ref{3-point gamma b c}) behaves as $\sim -g\,\ff \log^2 R$.
\end{remark} 

\subsection{Example: $4$-point functions and dilogarithm}
Consider the $4$-point function
\begin{equation}\label{4-point <a gamma bargamma gamma>}
\Big\langle \langle a(z_1),\theta_1\rangle\, \langle\gamma(z_2),X_2\rangle\, \langle\bar\gamma(z_3),X_3\rangle \,\langle \gamma(z_4),X_4\rangle  \Big\rangle
\end{equation}
with $\theta_1\in\g^*$ and $X_2,X_3,X_4\in\g$ any four fixed vectors in the coefficient space.\footnote{
In a fashion similar to the previous section, we could consider the $\g\otimes(\g^*)^{\otimes 3}$-valued correlator $\langle a(z_1)\otimes \gamma(z_2) \otimes \bar\gamma(z_3)\otimes \gamma(z_4) \rangle$ -- however, that would require us to resort to the formalism of Jacobi diagrams to describe the relevant tensors in $\g\otimes(\g^*)^{\otimes 3}$ appearing in the answer. For presentation purposes, we chose instead to consider the 4-point correlator of scalar components of fields, fixed by a choice of vectors $\theta_1,X_2,X_3,X_4$; this is tantamount to pairing the tensor-valued correlator 
$\langle a(z_1)\otimes \gamma(z_2) \otimes \bar\gamma(z_3)\otimes \gamma(z_4) \rangle$
with $\gamma_1\otimes X_2\otimes X_3\otimes X_4$.
}
There are two contributing diagrams:
\begin{equation}\label{4-point diagrams}
\vcenter{\hbox{\includegraphics[scale=0.5]{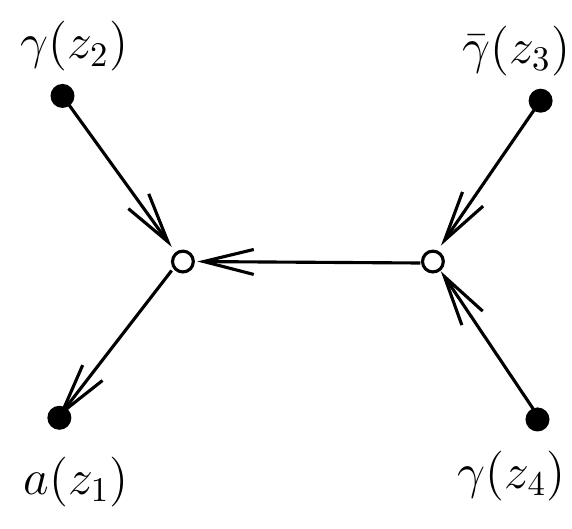}}} \qquad + \qquad \vcenter{\hbox{\includegraphics[scale=0.5]{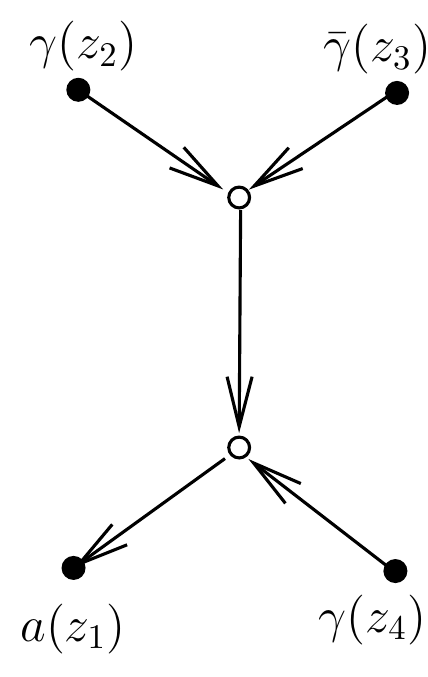}}}
\end{equation}
The first diagram yields
\begin{multline}\label{4-point fun computation}
{\scriptstyle
\wick{\Big\langle \langle \c1{a(z_1)},\theta_1\rangle\, \langle\c2{\gamma(z_2)},X_2\rangle\, \langle \c3{\bar\gamma(z_3)},X_3\rangle\, \langle\c4{\gamma(z_4)},X_4\rangle \cdot  g \int_{\CC\ni u}\frac{d^2 u}{2\pi}\langle\c1{\gamma}, [\c2{a},\c5{\bar{a}}]\rangle (u)\cdot g \int_{\CC\ni u'}\frac{d^2 u'}{2\pi}\langle-\c5{\bar\gamma}, [\c4{a},\c3{\bar{a}}]\rangle (u') \Big \rangle }
}
\\
=g^2 \langle\theta_1,[X_2,[X_3,X_4]] \rangle \int_{\CC^2\ni (u,u')}\frac{d^2 u}{2\pi} \frac{d^2 u'}{2\pi}\frac{1}{(z_1-u)(u-z_2)(\bar{u}-\bar{u}')(\bar{u}'-\bar{z}_3)(u'-z_4)}
\\
= -g^2  \langle\theta_1,[X_2,[X_3,X_4]] \rangle \int_{\CC\ni u}\frac{d^2 u}{2\pi} \frac{\log\left|\frac{z_3-z_4}{u-z_4}\right|}{(z_1-u)(u-z_2)(\bar{u}-\bar{z}_3)} \\
= g^2 \langle\theta_1,[X_2,[X_3,X_4]] \rangle\cdot \mathbb{I}(z_1,z_2,z_3,z_4)
\end{multline}
Here we introduced the notation
\begin{multline}\label{4-point I}
\mathbb{I}(z_1,z_2,z_3,z_4)=\\ 
= \frac{1}{2z_{12}}\left(i\DD\left(\frac{z_{34}}{z_{14}}\right)-i\DD\left(\frac{z_{34}}{z_{24}}\right)+
\log\left|\frac{z_{34}}{z_{14}} \right|\cdot \log\left|\frac{z_{23}}{z_{13}} \right| +
\log\left|\frac{z_{14}}{z_{24}} \right|\cdot \log\left|\frac{z_{23}}{z_{34}} \right|\right)
\end{multline}
where $z_{ij}=z_i-z_j$ and $\DD(-)$ is the Bloch-Wigner dilogarithm function \cite{Zagier} (see Appendix \ref{appendix: dilog integral} for a quick recap of the relevant properties). The integral over $u'$ in (\ref{4-point fun computation}) is evaluated using (\ref{3-point integral}) and the remaining integral over $u$ is evaluated using (\ref{3-point integral}), (\ref{3-point dilog integral}). The integral (\ref{4-point fun computation}) was considered in the literature, see \cite{Schnetz} (Section 5).


The full result for the $4$-point function (\ref{4-point <a gamma bargamma gamma>}) is:
\begin{multline}
\Big\langle \langle a(z_1),\theta_1\rangle\, \langle\gamma(z_2),X_2\rangle\, \langle\bar\gamma(z_3),X_3\rangle \,\langle \gamma(z_4),X_4\rangle  \Big\rangle = \\
=g^2 \Big(\langle\theta_1,[X_2,[X_3,X_4]] \rangle\cdot \mathbb{I}(z_1,z_2,z_3,z_4) + \langle\theta_1,[X_4,[X_3,X_2]] \rangle\cdot\mathbb{I}(z_1,z_4,z_3,z_2) \Big)
\end{multline}
The two terms here corresponds to the two diagrams (\ref{4-point diagrams}). Note that they are obtained from one another by interchanging points $z_2$ and $z_4$ and indices $b$ and $d$.

By a similar computation, one finds the $4$-point function
\begin{multline}\label{4-point a bargamma gamma bargamma}
\Big\langle \langle a(z_1),\theta_1\rangle\, \langle\bar\gamma(z_2),X_2\rangle\, \langle\gamma(z_3),X_3\rangle \,\langle \bar\gamma(z_4),X_4\rangle  \Big\rangle = \\
=g^2 \Big(\langle\theta_1,[X_2,[X_3,X_4]] \rangle\cdot \mathbb{J}(z_1,z_2,z_3,z_4) + \langle\theta_1,[X_4,[X_3,X_2]] \rangle\cdot\mathbb{J}(z_1,z_4,z_3,z_2)\Big)
\end{multline}
where 
\begin{multline}\label{4-point J}
\mathbb{J}(z_1,z_2,z_3,z_4)= -\int_{\CC^2}\frac{d^2 u}{2\pi}\frac{d^2 u'}{2\pi}\frac{1}{(z_1-u)(\bar{u}-\bar{z}_2)(u-u')(u'-z_3)(\bar{u}'-\bar{z}_4)} \\
=\frac{1}{2z_{13}}\left(i\DD\left(\frac{z_{14}}{z_{24}}\right)-i\DD\left(\frac{z_{34}}{z_{24}}\right)+ \log\left|\frac{z_{34}}{z_{24}}\right|\cdot \log\left|\frac{z_{12}}{z_{23}}\right| + 
\log\left|\frac{z_{34}}{z_{14}}\right|\cdot \log\left|\frac{z_{12}}{z_{24}}\right|
\right)
\end{multline}

We also have $4$-point functions involving ghosts which are computed similarly and are also expressed in terms of functions $\mathbb{I},\mathbb{J}$:
\begin{multline}
\label{4-point c gamma db gamma}
\Big\langle \langle c(z_1),\theta_1\rangle\, \langle\gamma(z_2),X_2\rangle\, \langle\dd b(z_3),X_3\rangle \,\langle \gamma(z_4),X_4\rangle  \Big\rangle = \\
=-g^2 (\langle\theta_1,[X_2,[X_3,X_4]] \rangle\cdot \mathbb{I}_{3412} + \langle\theta_1,[X_4,[X_3,X_2]] \rangle\cdot \mathbb{I}_{3214}), 
\end{multline}
\begin{multline}
\label{4-point c bargamma db bargamma}
\Big\langle \langle c(z_1),\theta_1\rangle\, \langle\bar\gamma(z_2),X_2\rangle\, \langle\dd b(z_3),X_3\rangle \,\langle \bar\gamma(z_4),X_4\rangle  \Big\rangle = \\
=-g^2 (\langle\theta_1,[X_2,[X_3,X_4]] \rangle\cdot \mathbb{J}_{1234} + \langle\theta_1,[X_4,[X_3,X_2]] \rangle\cdot \mathbb{J}_{1432}),
\end{multline}
\begin{multline}
\label{4-point c db bargamma gamma}
\Big\langle \langle c(z_1),\theta_1\rangle\, \langle\bar\gamma(z_2),X_2\rangle\, \langle\dd b(z_3),X_3\rangle \,\langle \gamma(z_4),X_4\rangle  \Big\rangle 
=\\
=-g^2 \Big(\langle\theta_1,[X_3,[X_2,X_4]] \rangle\cdot (\mathbb{I}_{1324}+\mathbb{J}_{3142})
+  \\
+\langle\theta_1,[X_2,[X_3,X_4]] \rangle\cdot \II_{3421}+\langle\theta_1,[X_4,[X_3,X_2]] \rangle\cdot\JJ_{4132} \Big)
\end{multline}
where for brevity we denoted $\II_{ijkl}=\II(z_i,z_j,z_k,z_l)$ and $\JJ_{ijkl}=\JJ(z_i,z_j,z_k,z_l)$. Note that (\ref{4-point c bargamma db bargamma}) is simply minus the correlator (\ref{4-point a bargamma gamma bargamma}).

We remark that $\II$ and $\JJ$, our building blocks for $4$-point functions,  have the following symmetries:
$$ \II_{2134}= \II_{1234},\quad \II_{1243}= \frac{\overline{z}_{12}}{z_{12}}\,\overline{\II_{1234}},\quad \JJ_{3412}=-\JJ_{1234} $$

Formulae (\ref{4-point <a gamma bargamma gamma>}, \ref{4-point a bargamma gamma bargamma}, \ref{4-point c gamma db gamma}, \ref{4-point c bargamma db bargamma}, \ref{4-point c db bargamma gamma}) and their complex conjugates exhaust all nonzero $4$-point function of fields from the set $\{a,\bar{a},\gamma,\bar\gamma, c, \dd b, \db b\}$ with total $\mc{AB}$-charge of fields under the correlator equal to $+2$. The other possibility is to have total $\mc{AB}$-charge zero; in this case the correlator coincides with the abelian one and is the sum of products of propagators, e.g. 
$$ \Big\langle \langle a(z_1),\theta_1\rangle\, \langle a(z_2),\theta_2\rangle\,  
\langle\gamma(z_3),X_3\rangle\,\langle \gamma(z_4),X_4\rangle \Big\rangle = 
\frac{\langle\theta_1,X_3\rangle\, \langle \theta_2,X_4\rangle}{z_{13}z_{24}}+\frac{\langle\theta_1,X_4\rangle\, \langle \theta_2,X_3\rangle}{z_{14}z_{23}}$$
with $\theta_1,\theta_2\in\g^*$, $X_3,X_4\in\g$ fixed vectors.

\subsection{Aside: from correlators on the plane to correlators on the sphere. Restoring M\"obius-invariance}\label{sec: sphere}
Consider the two-point function 
\begin{equation}\label{<a gamma>}
 \lan a(z_1)\otimes  \gamma(z_2) \ran = \frac{\idg}{z_1-z_2}  
\end{equation}
-- it coincides with the abelian propagator (\ref{propagators}), as there are no admissible Feynman graphs apart from the edge connecting $z_1,z_2$. As we will see (Section \ref{sss: T quantum, primary fields}, Proposition \ref{Prop T}),
field $a$ is primary, of conformal dimension $(\Delta,\bar{\Delta})=(1,0)$ and $\gamma$ is primary of dimension $(0,0)$. Global conformal  invariance  implies that two-point functions of primary fields of non-matching dimensions must vanish (see e.g. \cite{DMS}). Thus, (\ref{<a gamma>}) seems to be in contradiction with conformal invariance.

The explanation to the apparent paradox is that (\ref{<a gamma>}) is indeed not compatible with the global conformal symmetry of the sphere $\CC P^1$ (the group of M\"obius transformations), but is compatible with the global conformal symmetry of the plane $\CC$ (translations, rotations and scaling). Indeed, on a sphere the kinetic operators $\dd,\db,\dd \db$ appearing in (\ref{S_0 via cx fields}) have zero-modes, which we have killed when constructing propagators (\ref{propagators}) by imposing conditions on fields at $z=\infty$. In other words, a correlator (\ref{correlator PI}) on $\CC$ can be written as a correlator on the sphere with an additional field\footnote{
We refer the reader to Witten \cite{Witten_betagamma} (chapter 10) for details on soaking zero-modes and working with delta-functions of fields in $\beta\gamma$ systems.
}
$$\soak=\delta(\gamma)\delta(\bar\gamma)\delta(b)\delta(c)$$ 
inserted at $z=\infty$, which effectively imposes the necessary conditions on fields at infinity:
\begin{align*} 
\lan \Phi_1(z_1)\cdots \Phi_n(z_n) \ran_\CC =& \frac{1}{Z}\int e^{-\frac{1}{4\pi}S_{\CC P^1}}  \Phi_1(z_1)\cdots \Phi_n(z_n) \soak(\infty) \\
= &
\lan \Phi_1(z_1)\cdots \Phi_n(z_n) \soak(\infty) \ran_{\CC P^1}
\end{align*}
Here $\delta(\gamma)=\prod_a\delta(\gamma_a)$ and similarly for the other delta-functions; here $\gamma_a$, with $a=1,\ldots,\dim\g$ are the components of the field $\gamma=\gamma_a t^a$ with respect to some basis $\{t^a\}$ in $\g^*$. Moreover, we have $\delta(c)=\prod_a c^a$, $\delta(b)=\prod_a b_a$ since $b,c$ are odd, with $b_a$ the components of $b=b_a t^a$ and with  $c^a$ the components of $c=c^a t_a$ w.r.t. the dual basis $\{t_a\}$ in $\g$.\footnote{
For the normalization of the fields $\delta(\gamma)$, $\delta(\bar\gamma)$, $\delta(c)$, $\delta(b)$ to be well-defined, we should assume that the coefficient space $\g$ is endowed with a fixed volume element $\rho\in \wedge^{\dim\g}\g^*$ and the basis $\{t_a\}$ must be compatible with it: $t^1\wedge\cdots\wedge t^{\dim\g}=\rho$.
}

The version of the two-point function (\ref{<a gamma>}) on the sphere is the $3$-point function
\begin{equation}\label{<a gamma Delta>}
\begin{aligned}
\lan dz_1\, a(z_1)\otimes  \gamma(z_2) \;\soak(z_0) \ran_{\CC P^1} = & dz_1\,\left(\frac{1}{z_1-z_2} -\frac{1}{z_1-z_0} \right)\cdot\idg\qquad \\ 
= & dz_1\; \frac{z_2-z_0}{(z_1-z_2)(z_1-z_0)}\cdot\idg
\end{aligned}
\end{equation}
Here we included the factor $dz_1$ with $a(z_1)$ for convenience of tracking invariance properties. This answer on the sphere has 
the following properties: 
\begin{itemize}
\item It reduces to (\ref{<a gamma>}) in the limit $z_0\ra\infty$ and is invariant under the M\"obius group $PSL_2(\CC)$. (In fact, this property fully characterizes the answer.)
\item Asymptotic behavior at $z_2 \ra z_1$ (with $z_0$ fixed) is given by the pole (\ref{<a gamma>}).
\item At $z_2\ra z_0$ the result vanishes, which is consistent with $(\gamma \delta(\gamma))(z_0)=0$, cf. \cite{Witten_betagamma}.
\end{itemize}

One can also express (\ref{<a gamma Delta>}) in terms of the \textit{Szeg\"o kernel}
$$ \mu_{wz}= \frac{(dw)^{\frac12}\, (dz)^{\frac12}}{w-z}$$
-- a M\"obius-invariant holomorphic half-differential on the configuration space of two points on $\CC P^1$. Indeed, one has
\begin{equation}\label{<a gamma> on CP^1}
\lan dz_1\, a(z_1)\otimes \gamma(z_2)\; \soak(z_0) \ran_{\CC P^1} =  \frac{\mu_{z_1 z_2} \mu_{z_1 z_0}}{\mu_{z_2 z_0}}\cdot\idg
\end{equation}
The benefit of this form of the answer is that it is manifestly M\"obius-invariant.

Likewise, for instance, the $3$-point function (\ref{3p_1}) on the plane arises as the limit $z_0\ra \infty$ of a M\"obius-invariant $4$-point function on the sphere:
$$ \lan \gamma(z_1)\otimes \bar\gamma(z_2)\otimes dz_3\,a(z_3)\; \soak(z_0 )\ran_{\CC P^1} = -g\, \ff\cdot \frac{\mu_{z_3 z_1}\,\mu_{z_3 z_0}}{\mu_{z_1 z_0}} \log\left| \frac{(z_1-z_2)(z_3-z_0)}{(z_3-z_2)(z_1-z_0)} \right|  $$
We have again included the factor $dz_3$ with $a(z_3)$ for convenience. Note that the expression in $\log |\cdots|$ is the cross-ratio of the quadruple of points $(z_1,z_3;z_2,z_0)$ -- an invariant of the M\"obius group. Also, note that the the first factor in the r.h.s. vanishes at $z_1=z_0$ and the factor $\log|\cdots|$ vanishes at $z_2=z_0$.

Conformally invariant version of the two-point function $\lan c(z_1) b(z_2) \ran$ (\ref{propagators}) is the following $4$-point function on the sphere:
\begin{equation}\label{<cb> on CP^1}
\lan c(z_1)\otimes b(z_2)\; \tsoak(z_0)\; \delta(c(z'_0)) \ran_{\CC P^1}=
2\log \left|\frac{(z_1-z_2)(z'_0-z_0)}{(z_1-z_0)(z'_0-z_2)}\right|\cdot\idg
\end{equation}
Here we have split the field $\soak$ into $\tsoak=\delta(\gamma)\delta(\bar\gamma)\delta(b)$ and $\delta(c)$.\footnote{
The insertion of $\delta(c)$ at a point corresponds to requiring the gauge transformations to be trivial at that point. 
} The splitting of $\soak$ at a point $z_0$ into $\tsoak$ at $z_0$ and $\delta(c)$ at a different ``nearby'' point $z'_0$ is a version of the ``infrared regularization'' that we needed to define the $bc$ propagator on the plane (cf. footnote \ref{footnote: C in bc propagator}).

\begin{remark} Note that in (\ref{<a gamma> on CP^1}) we could also split $\soak(z_0)$ as $\tsoak(z_0)\delta(c(z_0'))$. The resulting $4$-point function on the sphere 
can be written in the form
\begin{equation}
\lan dz_1\, a(z_1)\otimes  \gamma(z_2)\; \tsoak(z_0)\; \delta(c(z_0')) \ran_{\CC P^1} =  dz_1\, \dd_{z_1} \Big(2 \log \left|\frac{(z_1-z_2)(z'_0-z_0)}{(z_1-z_0)(z'_0-z_2)}\right|\cdot\idg \Big)
\end{equation}
It does not depend on $z_0'$ and coincides with (\ref{<a gamma> on CP^1}).
\end{remark}

As another example, $3$-point function (\ref{3-point gamma b c}) becomes the following $5$-point function on the sphere, with added insertions of $\widetilde{\Theta}$ at $z_0$ and $\delta(c)$ at $z'_0\neq z_0$:
\begin{multline*}
\lan \gamma(z_1)\otimes b(z_2)\otimes c(z_3)\;\tsoak(z_0)\; \delta(c(z'_0)) \ran_{\CC P^1} = \\
= g\, \ff\cdot\frac{i}{4\pi}\int_{\CC\ni u} \frac{\mu_{u z_1}\mu_{u z_0}}{\mu_{z_1 z_0}}\cdot \frac{\bar\mu_{u z_3} \bar\mu_{u z'_0}}{\bar \mu_{z_3 z'_0}}\cdot 2 \log\left|\frac{(u-z_2)(z'_0-z_0)}{(u-z_0)(z'_0-z_2)} \right|
\end{multline*}
Here the three factors under the integral are the conformally invariant replacements of the three propagators constituting the integrand of (\ref{3-point gamma b c}). Note that the integral above is convergent; it can be computed explicitly in terms of dilogarithms, using (\ref{dilog-integral over D_R}).

In summary: every $n$-point correlator $\lan \Phi_1\cdots \Phi_n \ran$ on the plane not containing an infrared divergence (no bare $b$ field among $\Phi_1,\ldots,\Phi_n$) has a unique M\"obius-invariant extension as an $(n+1)$-point function on $\CC P^1$, with an added insertion $\soak(z_0)$. This extension is written in terms of Szeg\"o kernels and cross-ratios. In the case of a plane $n$-point correlator requiring infrared regularization (case when bare field $b$ occurs among $\Phi_1,\ldots,\Phi_n$), the M\"obius-invariant extension on $\CC P^1$ is an $(n+2)$-point function with added insertions of $\tsoak(z_0)$ and $\delta(c)$ at $z'_0\neq z_0$. It is also written in terms of Szeg\"o kernels and cross-ratios, via replacing the propagators in the Feynman diagram expansion of the plane correlator with their $\CC P^1$ counterparts (\ref{<a gamma> on CP^1}), (\ref{<cb> on CP^1}).

\begin{remark}
Fields $\tsoak=\delta(\gamma)\delta(\bar\gamma)\delta(b)$ and $\delta(c)$ which we use to ``soak'' the zero-modes satisfy the following:
\begin{itemize}
\item Both $\tsoak$ and $\delta(c)$ are $Q$-closed. Indeed:
\begin{multline*}
Q\tsoak 
=-\frac{g}{2}\big\langle \gamma-\bar\gamma,\big[c,\frac{\dd}{\dd \gamma}\delta(\gamma)\big] \big\rangle \delta(\bar\gamma) \delta(b)+
(-1)^{\dim\g} \delta(\gamma) \frac{g}{2}  \big\langle \gamma-\bar\gamma,\big[c,\frac{\dd}{\dd \bar\gamma}\delta(\bar\gamma)\big] \big\rangle \delta(b) + \\
+
\big\langle \gamma+\bar\gamma, \delta(\gamma)\delta(\bar\gamma)\frac{\dd}{\dd b}\delta(b)  \big\rangle
=g \, \mr{tr}_\g[ c,-]\;  \tsoak
= 0
\end{multline*}
Here we use that $\gamma\delta(\gamma)= \bar\gamma\delta(\bar\gamma)=0$. In the last step, we use
unimodularity of the Lie algebra $\g$.
 Also, 
$$Q\delta(c)=\frac{g}{2}\big\langle [c,c], \frac{\dd}{\dd c}\delta(c)\big\rangle=0$$
-- vanishes as a product of $\dim\g+1$ ghosts at a point (recall that $\delta(c)=c^{\dim\g}\cdots c^2 c^1$ is the product of all components of the $c$-ghost) or in other words because $\wedge^{\dim \g+1}\g^*=0$.\footnote{Note that $Q$-closedness would fail if we would have split $\soak$ instead into $\delta(\gamma)\delta(\bar\gamma)\delta(c)$ and $\delta(b)$.} 
We further note that $\tsoak$ can be split further into $Q$-cocylces:
\begin{equation}
\tsoak=\prod_{a=1}^{\dim\g} \Big(\delta(\lambda_a)b_a\Big)\cdot \delta(B)
\end{equation}
(with appropriately normalized delta-functions).
Here fields $\delta(\lambda_a)b_a$ are $Q$-closed for each $a$ and $\delta(B)$ is $Q$-closed due to unimodularity of $\g$. 
\item The operator product expansions $\OO^{(2)}(u)\, \tsoak(z)$ and $\OO^{(2)}(u)\, \delta(c(z))$ both have an integrable singularity in $u$ at $u=z$ (see Section \ref{sss: ab OPEs with O^2}).
\item The operator product expansion between the fields $\delta(c)$ and $\tsoak$ in the abelian theory (i.e. at $g=0$) has the form
\begin{equation}\label{soaking fields OPE}
\delta(c(z))\; \tsoak(w) \sim 
\delta(\gamma)\delta(\bar\gamma)\sum_{p=0}^{\dim\g-1} \frac{1}{p!} \Big(2\log|z-w|\Big)^{\dim\g-p}\cdot \lan c, b\ran^p +\reg
\end{equation}
where all fields on the r.h.s. are at $w$ and composite field $\lan c, b\ran^p
$ is understood as \textit{renormalized}, cf. (\ref{normal ordering}). In the non-abelian theory, there are  additional terms of order $\geq 1$ in $g$, which also come with powers of $\log|z-w|$.
\item Note that our way of soaking zero-modes is different from the way proposed by Witten in \cite{Witten2dYM}, by using $\exp(-g_0^2\int\mu\,\mr{tr}B^2)$ with $\mu$ an area form and $g_0$ the standard coupling constant in two-dimensional Yang-Mills theory. 
We will explain the geometrical meaning of our soaking operators in terms of the moduli space of flat connections elsewhere.
\end{itemize}
\end{remark}

\section{Operator product expansions}\label{ss: OPE} 
Given two fields $\Phi_1$, $\Phi_2$, we are interested in the singularity of the correlator
\begin{equation}\label{OPE correlator}
\lan \Phi_1(z) \Phi_2(w)\; \phi_1(x_1)\cdots \phi_n(x_n) \ran
\end{equation}
in the asymptotics $z\ra w$; here $\phi_1,\ldots,\phi_n$ are arbitrary test fields inserted at finite distance away from $z,w$. Operator product expansion (OPE) is an expression of the form 
\begin{equation}\label{OPE}
\Phi_1(z) \Phi_2(w) \sim \sum_{i=1}^s \sigma_i(z-w) \widetilde{\Phi}_i(w)+\mr{reg.}
\end{equation}
with $\widetilde{\Phi}_i$ some fields and $\sigma_i(z-w)$ some singular coefficient functions, typically of form $(z-w)^{-p} (\bar{z}-\bar{w})^{-q}\log^r|z-w|$ with $p+q\geq 0, r\geq 0$; $\mr{reg.}$ stands for terms which are regular (continuous) at $z\ra w$. The number $s$ of singular terms on the r.h.s. depends on fields $\Phi_1,\Phi_2$. Expression (\ref{OPE}) means that one can replace the product $\Phi_1(z) \Phi_2(w)$ with the right hand side in a correlator (\ref{OPE correlator}) with  arbitrary test fields $\phi_1,\ldots,\phi_n$ inserted away from $z,w$, reproducing the correct behavior of the correlator at $z\ra w$, modulo terms having a well-defined limit at $z\ra w$.

Test fields $\phi_1,\ldots,\phi_n$ in (\ref{OPE correlator}) can be assumed to be fundamental fields without loss of generality. 

Let $\Phi_1,\Phi_2$ be two (possibly, composite) fields. 
The OPE is given by a sum of Feynman graphs $\gm$ with loose half-edges decorated by respective fundamental fields (or derivatives)
-- their product over the loose half-edges yields the composite field $\widetilde{\Phi}(w)$ in the 
term of the OPE corresponding to $\gm$. Graphs $\gm$ contributing to the OPE have the following properties:
\begin{enumerate}[(i)]
\item\label{OPE subgraph (i)} Graph $\gm$ contains one vertex decorated by $\Phi_1(z)$, one vertex decorated by $\Phi_2(w)$ and $k\geq 0$ interaction vertices decorated by $\OO^{(2)}(u_1),\ldots, \OO^{(2)}(u_k)$, with $u_1,\ldots,u_k$ integrated over $\CC$.
\item\label{OPE subgraph (ii)} Cutting any single edge in $\gm$, we do not create a connected component which contains neither vertex $\Phi_1(z)$, nor vertex $\Phi_2(w)$. This is an analog of the one-particle irreducibility; by an abuse of terminology, we will call graphs with this property 1PI graphs.
\end{enumerate}
Graphs $\gm$ arise as subgraphs of Feynman graphs $\Gamma$ contributing to the correlator (\ref{OPE correlator}). Loose half-edges correspond to edges of $\Gamma$ 
that are severed when cutting out the subgraph.
1PI requirement for $\gm$ is imposed in order to avoid overcounting: for a graph $\Gamma$ contributing to (\ref{OPE}),  there is a unique way to single out the OPE subgraph $\gm$ satisfying (\ref{OPE subgraph (i)}), (\ref{OPE subgraph (ii)}) above. The contribution of the quotient graph $\Gamma/\gm$ (i.e. $\Gamma$ with the subgraph $\gm$ collapsed into a single vertex) to the correlator of the term $i=\gm$ in the r.h.s. of (\ref{OPE}) with the test fields $\phi_1,\ldots,\phi_n$ 
is the same as the contribution of $\Gamma$ to (\ref{OPE correlator})
(up to regular terms at $z\ra w$).

\subsection{
OPEs of fundamental fields
}\label{sss: OPE of fundamental fields}
For example, consider the OPE
\begin{equation}\label{a^a gamma_b OPE}
 a(z)\otimes \gamma(w) 
\end{equation}
The only potentially contributing Feynman graphs $\gm$ are  graphs of ``branch'' type
\begin{equation}\label{OPE branch}
\vcenter{\hbox{\includegraphics[scale=0.5]{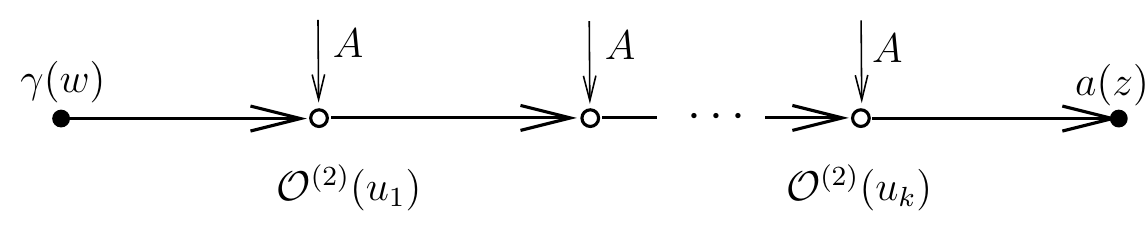}}} 
\end{equation}
with $k\geq 0$ interaction vertices. They arise as subgraphs of trees (or disjoint unions of trees) 
$\Gamma$ contributing to a correlator $\lan a(z)\otimes \gamma(w)\;\phi_1(x_1)\cdots \phi_n(x_n)  \ran $:
$$\vcenter{\hbox{\includegraphics[scale=0.4]{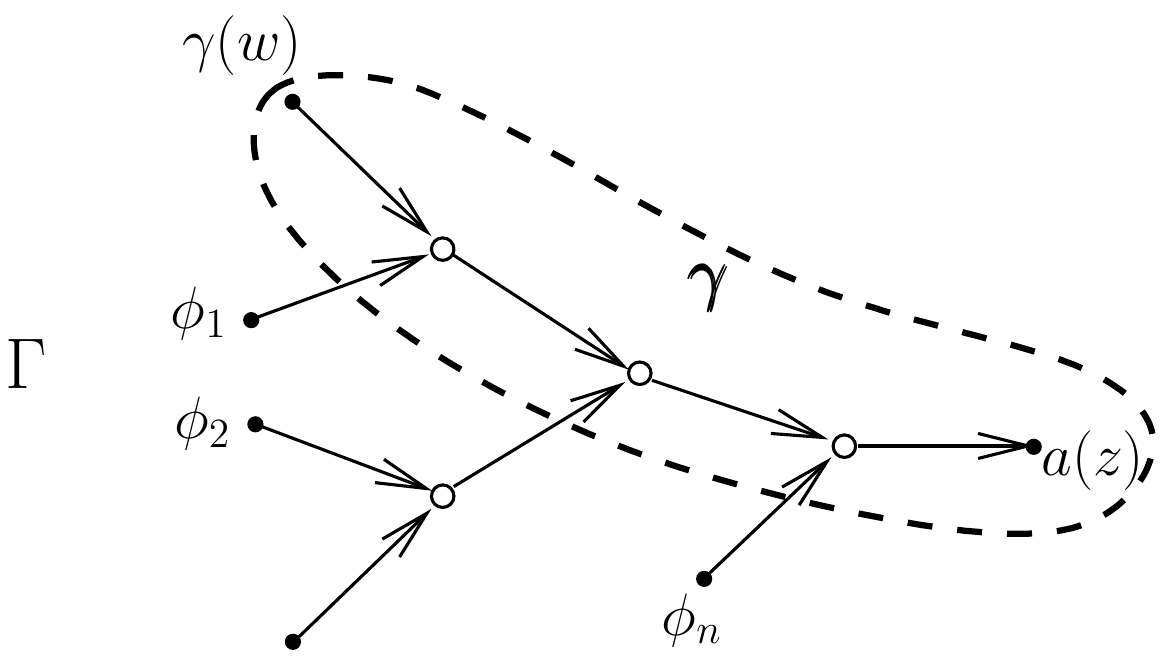}}} $$

For instance, for $k=0$, the branch (\ref{OPE branch}) is a single edge connecting  $\gamma(w)$ and $a(z)$, its contribution to the OPE  (\ref{a^a gamma_b OPE}) is simply the propagator $\frac{\idg}{z-w}$ (times the identity field suppressed in the notation).

For $k=1$, the contribution of the branch graph 
$$\vcenter{\hbox{\includegraphics[scale=0.5]{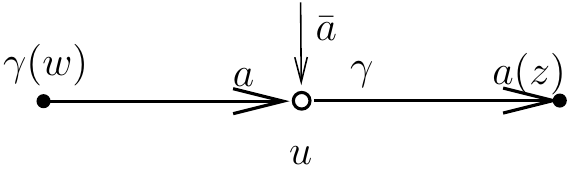}}} $$
to the OPE is:
\begin{equation}\label{a^a gamma_b OPE k=1}
 -g \int_{\CC} \frac{d^2 u}{2\pi}\,\frac{\mr{ad}(\bar{a}(u))}{(z-u)(u-w)} 
\end{equation}
with $\mr{ad}(\cdots)$ the adjoint action. In this expression, we replace the field $\bar{a}^c(u)$ with its Taylor expansion around $w$, 
\begin{equation}\label{abar(u)=abar(w)+R}
\bar{a}(u)=\sum_{i,j\geq 0}\frac{1}{i!j!}\dd^i \db^j \bar{a}(w)\; (u-w)^i (\bar{u}-\bar{w})^j\quad = \bar{a}(w)+\mc{R}(u,w) 
\end{equation}
where split the Taylor expansion into the zeroth term and the remainder (error term) $\mc{R}(u,w)$ behaving as $O(|u-w|)$. Under the correlator with test fields, the term in (\ref{a^a gamma_b OPE k=1}) with $\bar{a}$ replaced by $\mc{R}(u,w)$ is continuous as $z\ra w$: setting $z=w$, we get an integrable singularity of the integrand. Thus, up to a regular term, (\ref{a^a gamma_b OPE k=1}) is equivalent to
\begin{equation}
-g\int_{\CC} \frac{d^2 u}{2\pi}\,\frac{\mr{ad}(\bar{a}(w))}{(z-u)(u-w)} = 
-\frac{g}{2}\,  \frac{\bar{z}-\bar{w}}{z-w}\; \mr{ad}(\bar{a}(w))
\end{equation}
-- cf. (\ref{int 1/((u-z)(u-x))}) for the evaluation of the integral.

Finally, for $k\geq 2$, branch graphs (\ref{OPE branch}) give regular contributions to the OPE: setting $z=w$, we get an integrable singularity of the integrand as any subset of $u_1,\ldots,u_k$ approaches $z=w$ (by power counting arguments of the proof of Proposition \ref{Prop: corr properties}).

Thus, we have a complete result for the OPE (\ref{a^a gamma_b OPE}):
\begin{equation}\label{a gamma OPE answer}
a(z)\, \gamma(w)\sim \frac{\idg}{z-w}-\frac{g}{2}\;  \frac{\bar{z}-\bar{w}}{z-w}\cdot \mr{ad}(\bar{a}(w))\; + \mr{reg.}
\end{equation}

As a check of this result, we can take the correlator of left and right side of (\ref{a gamma OPE answer}) with the test field $\bar\gamma(x)$. We obtain
$$ 
\lan a(z)\otimes \gamma(w)\otimes \bar\gamma(x)\ran
\stackrel{?}{\sim} -\frac{g}{2}  \; \frac{\bar{z}-\bar{w}}{z-w}\; 
\lan \mr{ad}(\bar{a}(w))\otimes\bar\gamma(x) \ran
+\mr{reg.} $$ 
The $3$-point function on the left, known from (\ref{3p_1}), can be written as $\frac{g}{2} \ff\cdot \frac{1}{z-w}{\log|1+\frac{z-w}{w-x}|^2}$ and is indeed equivalent to r.h.s., $\frac{g}{2}\ff\cdot \frac{\bar{z}-\bar{w}}{z-w}\frac{1}{\bar{w}-\bar{x}}$, as $z\ra w$.

As another example, consider the OPE
$$ a(z)\otimes \bar\gamma(w) $$
As in the previous case, 
we have branch graphs similar to (\ref{OPE branch}), and graphs with $k\geq 2$ don't contribute to the singular part of the OPE by a power counting argument. Case $k=0$ is now also absent: propagator between $a$ and $\bar\gamma$ is zero. Thus, we only have the contribution of the  $k=1$ graph
$$\vcenter{\hbox{\includegraphics[scale=0.5]{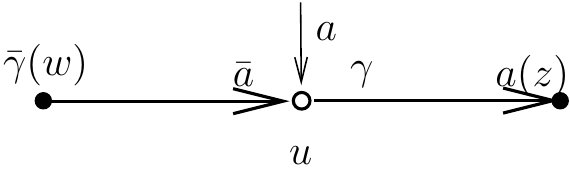}}} $$
This results in the following OPE:
\begin{equation}\label{a bargamma OPE}
\begin{aligned}
a(z)\otimes \bar\gamma(w)\sim & 
g \int_\CC \frac{d^2 u}{2\pi} \frac{\mr{ad}(a(u))}{(z-u)(\bar{u}-\bar{w})}+\mr{reg.} \\
\sim & g \int_\CC \frac{d^2 u}{2\pi} \frac{\mr{ad}(a(w))}{(z-u)(\bar{u}-\bar{w})}+\mr{reg.} \\
\sim &
g  \log |z-w|\cdot \mr{ad}(a(w)) + \mr{reg.}
\end{aligned}
\end{equation}
Here we use the same argument as above to replace $a(u)$ with $a(w)$. The resulting integral over $u$ is logarithmically divergent at $u\ra\infty$ and needs an infrared regularization $|u|<R$.\footnote{
More precisely: we split $a(u)=a(w)+\mc{R}(u,w)$ as in (\ref{abar(u)=abar(w)+R}). Then we have 
$\int d^2 u\frac{a(u)}{(z-u)(\bar{u}-\bar{w})}=\int d^2 u\frac{a(w)}{(z-u)(\bar{u}-\bar{w})}+
\int d^2 u\frac{\mc{R}(u,w)}{(z-u)(\bar{u}-\bar{w})}$. Here the integral on the left is convergent at $u\ra\infty$ when placed under a correlator with a test field. On the right, it is split into two integrals which are both infrared-divergent, but their behavior (after imposing a cutoff $|u|<R$) at $z\ra w$ is easier to analyze.
} 
The regularized integral is given by (\ref{log-integral over D_R}). Changing the cutoff $R$ does not affect the singular part of the result.

By similar computations, we have the following OPEs:
\begin{align}
\gamma(z)\otimes\bar\gamma(w)\sim & - g\, \log|z-w| \cdot \langle \ff, (\gamma-\bar\gamma)(w)\rangle +\reg
\label{OPE gamma bargamma} \\ \label{OPE c db}
c(z)\otimes \dd b(w) \sim& -\frac{\idg}{z-w}+\frac{g}{2}\, \frac{\bar{z}-\bar{w}}{z-w}\cdot \mr{ad}(\bar{a}(w)) +\\
\nonumber &+ g\, \log|z-w|\cdot \mr{ad}(a(w))+\reg  
\\
c(z)\otimes \gamma(w)  \sim & \; g\, \log|z-w|\cdot \mr{ad}(c(w))+\reg  
\label{OPE c gamma}\\
\dd b(z)\otimes \gamma(w) \sim & -\frac{g}{2}\, \frac{\bar{z}-\bar{w}}{z-w}\;\langle\ff,\db b(w)\rangle + \reg
\label{OPE db gamma}\\
\dd b(z)\otimes  \bar\gamma(w) \sim & -g\,  \log|z-w| \cdot \langle\ff, \dd b(w)\rangle+\reg 
\label{OPE db bargamma}
\end{align}
For each OPE, there is also the complex conjugate one.

Let us denote $\mr{reg}^{(p)}$ a remainder term in an OPE which has continuous derivatives of order $\leq p$ at $z=w$. In particular, by default we write OPEs up to $\reg = \mr{reg}^{(0)}$ terms.

For the OPE $c(z)\, b(w)$ only the branch graph with $k=0$ (i.e. just a single edge) contributes:
\begin{equation} \label{OPE cb}
 c(z)\otimes b(w) \sim 2\,\log|z-w|\cdot\idg+\reg 
\end{equation}
Note that the ``regular'' part here is just continuous but not differentiable at $z=w$, as implied by presence of $O(g)$ terms in $c(z) \dd b(w)$ OPE (\ref{OPE c db}). 
The latter imply that (\ref{OPE cb}) can be refined to
$$ c(z)\otimes b(w) \sim \log|z-w| \Big( 2\cdot\idg- g\, (z-w)\cdot \mr{ad}(a(w))- g\,  (\bar{z}-\bar{w})\cdot\mr{ad}(\bar{a}(w)) \Big) + \mr{reg}^{(1)}
$$
where 
the remainder term is differentiable (but not twice differentiable) at $z=w$. The $O(g)$ contribution here can be seen as coming from $k=1$ branch graph for $cb$ OPE which is continuous but not differentiable.


\textbf{OPEs which are trivial due to Feynman diagram combinatorics.} The OPEs of the following pairs of fundamental fields are purely regular:
\begin{equation}\label{regular OPEs}
\begin{gathered} a(z)\otimes a(w)\sim \mr{reg}^{(\infty)},\quad 
a(z)\otimes \bar{a}(w) \sim \mr{reg}^{(\infty)}, \quad
a(z)\otimes c(w)\sim \mr{reg}^{(\infty)}, \\
c(z)\otimes c(w)\sim \mr{reg}^{(\infty)},\quad
b(z)\otimes b(w)\sim \mr{reg}^{(\infty)},\quad
a(z)\otimes b(w)\sim \mr{reg}^{(\infty)}
\end{gathered}
\end{equation}
In each of these cases one can also take arbitrary derivatives of the first and second field and the OPE is still regular -- there are no contributing Feynman graphs (cases $bb$ and $ab$ is slightly more subtle: there is an admissible orientation of branch graphs but no admissible decoration of half-edges by fields). In other words, right hand sides in OPEs (\ref{regular OPEs}) are infinitely-differentiable in $z,\bar{z},w,\bar{w}$ at $z=w$.

On the other hand, we have
$$ \gamma(z)\otimes \gamma(w)\sim\reg,\quad  
b(z) \otimes\gamma(w)\sim \reg 
$$
-- with continuous but non-differentiable r.h.s. at $z=w$. 
In fact,  OPEs (\ref{OPE db gamma}), (\ref{OPE db bargamma}) 
imply that
$$ b(z) \otimes\gamma(w)\sim -g\,(\bar{z}-\bar{w})\log|z-w|\cdot\langle \ff,\db b(w)\rangle + \mr{reg}^{(1)} $$
with a remainder term which is 
differentiable but not twice differentiable at $z=w$. Similarly, $\gamma(z)\gamma(w)$ is non-differentiable because there are contributions of $k=2$ branch diagrams to OPEs $\dd\gamma(z) \gamma(w)$, $\db \gamma(z)\gamma(w)$.

\subsection{OPEs of derivatives of fundamental fields}
For an OPE of general derivatives of fundamental fields, branch graphs with $k>1$ can contribute -- but only finitely many of them: for $k$ sufficiently large, the limit $z=w$ of the integral over $u_1,\ldots,u_k$ is convergent. 
One can also find a bound on $k$ from a weight counting argument (here ``weight'' is understood as in Section \ref{sss: conf dimension}), as follows. 

As an example, consider the OPE 
\begin{equation}\label{OPE d^p a d^q gamma}
\dd^p a(z)\otimes \dd^q  \gamma(w)
\end{equation} 
with some $p,q\geq 0$. The weight of this expression is $(h,\bar{h})=(p+q+1,0)$. A contribution of a branch graph with $k$ interaction vertices to the OPE is a sum of terms of form
\begin{equation}\label{OPE d^p a d^q gamma term}
\prod_{i=1}^r (\dd^{\mu_i} \db^{\nu_i} a)(w)\cdot \prod_{j=1}^s (\dd^{\rho_j} \db^{\sigma_j} \bar{a})(w)\cdot (z-w)^l (\bar{z}-\bar{w})^m \log^\alpha|z-w|  
\end{equation}
with $r+s=k$ (fields $a$, $\bar{a}$ should be appropriately contracted via structure constants); note that derivatives of $a,\bar{a}$ arise from expanding a field inserted at $u$ in a Taylor series centered at $w$. The weight of this expression is $(h,\bar{h})=(r+\sum \mu_i+\sum \rho_j-l, s+\sum \nu_i+\sum\sigma_j-m)$.  It has to coincide with the weight of (\ref{OPE d^p a d^q gamma}). In particular, for the total weight $h+\bar{h}$, we have
$$ \underbrace{r+s}_k+\underbrace{\sum \mu_i + \sum \nu_i + \sum \rho_j+\sum \sigma_j}_{\geq 0}-l-m=p+q+1 $$
In particular, we have $l+m\geq k-(p+q+1)$. If $l+m\geq 1$, the term (\ref{OPE d^p a d^q gamma term}) is non-singular (continuous). 
Thus, one can only have singular terms in the OPE if 
\begin{equation}\label{OPE d^p a d^q gamma k estimate}
k\leq p+q+1
\end{equation}

Similarly, for the OPE $\dd^p \db^{p'}a(z)\otimes \dd^q\db^{q'}\gamma(w)$, only branch graphs with $k\leq p+p'+q+q'+1$ can contribute to the singular part.

As another example, for the OPE $\dd^p \db^{p'}\gamma(z)\otimes \dd^q\db^{q'}\gamma(w)$, there are terms containing  $c$, a derivative of $b$ and $(k-2)$ fields $a,\bar{a}$ (each field can come with more derivatives) 
-- for these terms one obtains the estimate $k\leq p+p'+q+q'+1$. There are also terms containing $\gamma-\bar\gamma$ 
and $(k-1)$ fields $a,\bar{a}$ 
-- these yield the same estimate for $k$. 

\textit{Convergence argument.} As we mentioned above, alternatively to going the route of weight counting,
one can prove that branch graphs with large $k$ do not contribute to the singular part of the OPE by checking convergence of the integral over $u_1,\ldots,u_k$ in the limit $z=w$. For instance, consider the OPE (\ref{OPE d^p a d^q gamma}). At $z=w$, the corresponding contribution contains an integral of the form
\begin{equation}\label{OPE d^p a d^q gamma integral}
 \int_{u_1,\ldots,u_k} \dd_w^p P_{wu_1}\cdot P_{u_1 u_2}\cdot\; \cdots\; \cdot P_{u_{k-1} u_k} \cdot \dd^q_w P_{u_k w} 
\end{equation}
where each propagator $P_{x y}$ is either $\frac{1}{x-y}$ or $\frac{1}{\bar{x}-\bar{y}}$. 
We should analyze the potential ultraviolet problems:
\begin{enumerate}
\item If some of the $u_i$'s collapse together (but not at $w$), we have an integrable singularity in (\ref{OPE d^p a d^q gamma integral}) by the argument of (\ref{convergence item (a)}) of the proof of Proposition \ref{Prop: corr properties}.
\item If a proper subset of $u_i$'s, with indices $i\in S\subset \{1,\ldots,k\}$, collapses on $w$, we consider the integral over $\{u_i\}_{i\in S}$ in a disk $D_{w,\epsilon}$ centered at $w$ of small radius $\epsilon$, with the non-collapsing points $u_i$ fixed outside of the disk and regarded as parameters. This gives a product of convergent integrals  (by (\ref{convergence item (b)}) of the proof of Proposition \ref{Prop: corr properties}), one integral per each string of consecutive integers in $S$;\footnote{In other words, we have one integral per connected component of the collapsing subgraph of the branch graph $\gm$. Here the ``collapsing subgraph'' is the full subgraph of $\gm$ with vertices $\{u_i\}_{i\in S}\cup \{w,z\}$}  two of these integrals can be equipped with derivatives $\dd^p_w$ and $\dd^q_w$ which does not affect convergence.
\item \label{OPE d^p a d^q gamma convergence (3)} If \textit{all} $u_i$'s collapse at $w$, the integrand of (\ref{OPE d^p a d^q gamma integral}) behaves as $O(\frac{1}{\epsilon^{p+q+k+1}})$ when all $u_i$'s are at the distance of order $\epsilon$ from $w$ and from each other. Thus, the $2k$-fold integral (\ref{OPE d^p a d^q gamma integral}) may be divergent if $p+q+k+1\geq 2k$ (i.e. when $k\leq p+q+1$) but is convergent otherwise.
\end{enumerate}
Here are the typical collapsing subgraphs of $\gm$ corresponding to these three cases (we draw them on the graph $\gm$ with vertices $z$ and $w$ identified).
$$\vcenter{\hbox{\includegraphics[scale=0.5]{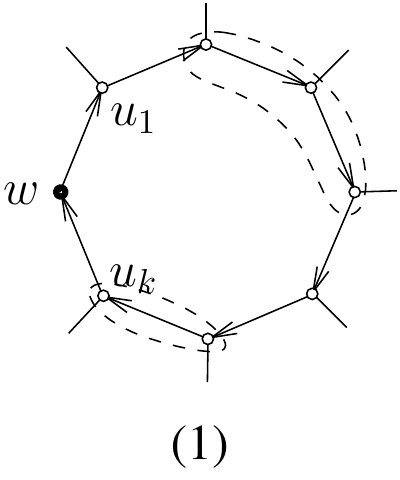}}}\qquad 
\vcenter{\hbox{\includegraphics[scale=0.5]{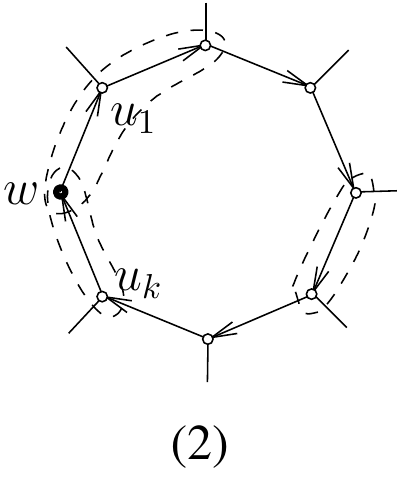}}}\qquad 
\vcenter{\hbox{\includegraphics[scale=0.5]{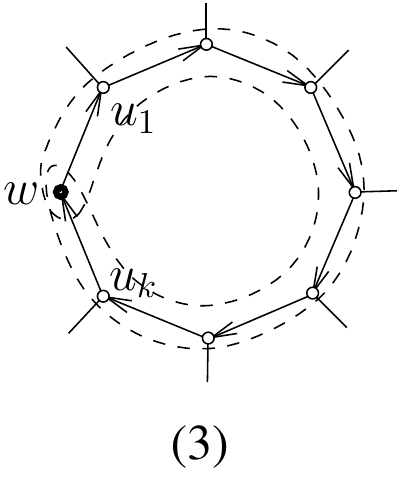}}}
$$

Thus, only the situation (\ref{OPE d^p a d^q gamma convergence (3)}) can lead to an ultraviolet problem at $z=w$, which gives us an upper bound for which $k$ can contribute to the singular part of the OPE. Note that the bound is the same as the one we obtained above from weight counting (\ref{OPE d^p a d^q gamma k estimate}).

\begin{example} For instance, the following OPEs contain the contributions of branch diagrams with $k\leq 2$:
\begin{gather}
a(z)\otimes  \dd \gamma(w) \sim \label{a  dd gamma OPE} \\ \nonumber
\sim \frac{\idg}{(z-w)^2}
%
-\frac{g}{2}\Big(  \frac{\bar{z}-\bar{w}}{(z-w)^2}\cdot \ad(\bar{a}) +
 \frac{\bar{z}-\bar{w}}{z-w}\cdot \ad(\dd\bar{a}) +
 \frac{1}{2}\frac{(\bar{z}-\bar{w})^2}{(z-w)^2} \cdot \ad(\db\bar{a})
\Big)
 \\ \nonumber
+\frac{g^2}{4}\Big(\frac{\bar{z}-\bar{w}}{z-w}\cdot \ad(a)\ad(\bar{a})-  \frac12 \frac{(\bar{z}-\bar{w})^2}{(z-w)^2}\cdot \ad(\bar{a})\ad(\bar{a}) \Big) +\reg \\
\dd c(z)\otimes \dd b(w) \sim   \label{dd c  dd b OPE} 
\frac{\idg}{(z-w)^2}+
\\ \nonumber 
+
\frac{g}{2}\Big(\frac{1}{z-w}\cdot \ad(a)+
\frac{\bar{z}-\bar{w}}{z-w}\cdot \ad(\db a)-
\frac{\bar{z}-\bar{w}}{(z-w)^2}\cdot \ad(\bar{a})-\frac12 \frac{(\bar{z}-\bar{w})^2}{(z-w)^2}\cdot\ad(\db \bar{a})\Big)
\\ \nonumber
-\frac{g^2}{4}\Big( 
2 \log|z-w|  \cdot\ad(a)\ad(a) +\frac{\bar{z}-\bar{w}}{z-w}  \cdot\ad(a)\ad(\bar{a}) \\ \nonumber
 + \frac{\bar{z}-\bar{w}}{z-w} \cdot\ad(\bar{a})\ad(a)- \frac12 \frac{(\bar{z}-\bar{w})^2}{(z-w)^2}  \cdot\ad(\bar{a})\ad(\bar{a})
\Big) + \reg
\\
\dd c(z)\otimes  \gamma(w) \sim  \label{dd c  gamma OPE} 
\frac{g}{2}\Big(\frac{1}{z-w}\cdot\ad(c)+\frac{\bar{z}-\bar{w}}{z-w}\cdot\ad(\db c) \Big)-
\\ \nonumber
-\frac{g^2}{4}\Big(
2\log|z-w| \cdot\ad(a)\ad( c)+ \frac{\bar{z}-\bar{w}}{z-w} \cdot\ad(\bar{a})\ad( c) 
\Big)+\reg
\end{gather}
Here 
all the fields on the r.h.s. are at $w$. These particular OPEs will be important when studying the stress-energy tensor $T$ and BRST current $J$ as composite fields in Section \ref{ss: composite fields}.
\end{example}

\subsection{Some important OPEs in abelian theory involving $\OO^{(2)}$}\label{sss: ab OPEs with O^2}
Limit $g=0$ of an OPE between two composite fields $\Phi_1$, $\Phi_2$ is given by a sum of Wick contractions of some of the constituent fundamental fields of $\Phi_1$ with some of the constituent fundamental fields of $\Phi_2$, i.e., by Feynman graphs with two vertices corresponding to $\Phi_1$ and $\Phi_2$, with no interaction vertices and with loose half-edges allowed. Short loops are not allowed (which corresponds to the assumption that $\Phi_1$, $\Phi_2$ are renormalized/normally ordered\footnote{
An implicit assumption here is that the order in which fundamental fields (or their derivatives) are merged when building $\Phi_1,\Phi_2$ is such that $g=0$ limit coincides with the usual normal ordering prescription in a free theory; one can always choose such an order, see Section \ref{sss: order of collapse ambiguity}.
}). Such OPEs in abelian $BF$ theory were studied in \cite{LMY}. Here, for the study of non-abelian theory as a deformation of the abelian one, we are interested in several OPEs involving the deforming observable $\OO^{(2)}$.

As an example, consider the OPE $\OO^{(2)}(z)\, \OO^{(2)}(w)$ in the free theory. We  have the following
diagrams: 
$$ \vcenter{\hbox{\includegraphics[scale=0.5]{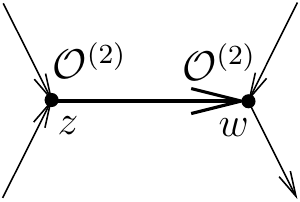}}}\qquad, \qquad 
\vcenter{\hbox{\includegraphics[scale=0.5]{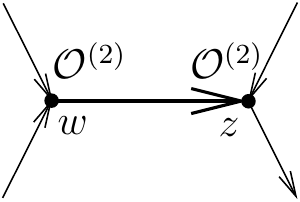}}}\qquad,\qquad
\vcenter{\hbox{\includegraphics[scale=0.5]{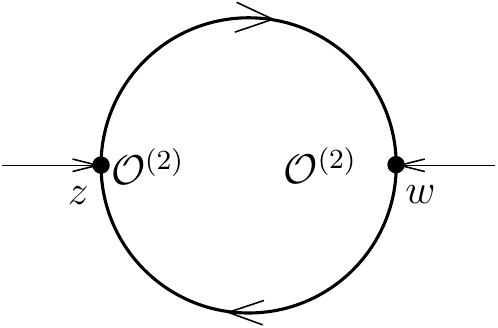}}}
$$
The singular contribution of the first one, taking into account all possible decorations of half-edges is:
$$\Big(\frac12 \lan  B , [A,[A,A]] \ran + \Big\langle -*db, [A,[A,c]]  + \frac12 [c,[A,A]] \Big\rangle\Big)\; 2\, d\mr{arg}(z-w) = 0 $$
-- vanishes due to Jacobi identity in $\g$. The second diagram is similar. Contribution of the third diagram vanishes by boson-fermion cancellation in a loop mechanism. Thus, the free theory OPE is trivial:
\begin{equation}\label{O^2 O^2 ab OPE}
\OO^{(2)}(z)\, \OO^{(2)}(w)  \underset{g=0}{\sim } \reg
\end{equation}

Next, consider the free theory OPE of the (abelian) stress-energy tensor $T_0=\lan a,\dd \gamma \ran+\lan \dd b, \dd c \ran$ with $\OO^{(2)}$. We have the following contributing diagrams:
$$ \vcenter{\hbox{\includegraphics[scale=0.5]{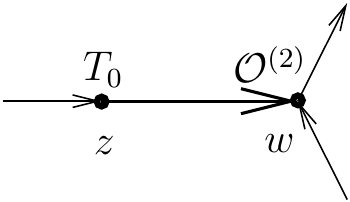}}}\qquad, \qquad
\vcenter{\hbox{\includegraphics[scale=0.5]{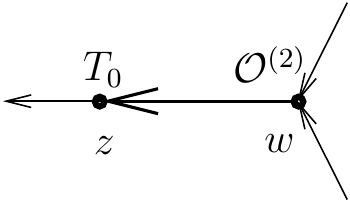}}} \qquad,\qquad
\vcenter{\hbox{\includegraphics[scale=0.5]{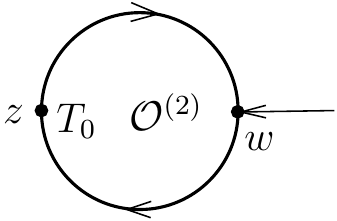}}}
$$
Here the last diagram gives
$$ 4\,d^2 w\,\frac{ \mr{tr}_\g\mr{ad}(\bar{a}(w))}{(z-w)^3} $$
-- a potential third order pole contribution to the OPE, which vanishes due to unimodularity. Thus, the OPE is given by the first two graphs, which yield
\begin{equation}\label{T_0 O^2 ab OPE}
T_0(z)\, \OO^{(2)}(w) \underset{g=0}{\sim } \frac{\OO^{(2)}(w)}{(z-w)^2}+\frac{\dd \OO^{(2)}(w)}{z-w}+\reg
\end{equation}
Together with the complex conjugate OPE, this implies that $\OO^{(2)}$ is a primary field of conformal dimension $
(1,1)$ in the abelian theory.\footnote{If $z$ and $w$ are allowed to collide,   
one must include an additional
contact term in the OPE (\ref{T_0 O^2 ab OPE}), see (\ref{T_0 O^2 with contact term}) below.}

As another example, consider the OPE of $\OO^{(2)}$ with the ``soaking field'' $\tsoak=\delta(\gamma)\delta(\bar\gamma)\delta(b)$ which appeared in Section \ref{sec: sphere}. We have:
\begin{multline}
\OO^{(2)}(z) \tDelta(w) \underset{g=0}{\sim } -2 d^2 z\, \Big(\frac{1}{|z-w|^2}  \langle(\gamma-\bar\gamma)(z),[\frac{\dd}{\dd\gamma},\frac{\dd}{\dd\bar\gamma}]\rangle\, \tDelta(w)+\\
+2\frac{\log|z-w|}{z-w}\langle \db b(z),[\frac{\dd}{\dd\gamma},\frac{\dd}{\dd b}]\rangle\, \tDelta(w) + c.c.
 \Big) + O\big(\frac{1}{|z-w|}\big) 
\end{multline}
Here $c.c.$ stands for the complex conjugate of the second term. Note that the first term contains $\frac{1}{|z-w|^2}$ (coming from Wick contractions of $a,\bar{a}$ from $\OO^{(2)}$ with $\gamma,\bar{\gamma}$ from $\tDelta$) times a sum of two expressions vanishing as $(z-w)$ and as $(\bar{z}-\bar{w})$ respectively -- these zeroes arise from $\gamma(z) \delta(\gamma(w))$ and $\bar\gamma(z) \delta(\bar\gamma(w))$. Therefore the worst singularities in this OPE are in fact $\frac{\log|z-w|}{z-w}$ and $\frac{\log|z-w|}{\bar{z}-\bar{w}}$ coming from the second term and its complex conjugate -- these terms arise from the pair of Wick contractions of $a,c$ from $\OO^{(2)}$ with $\gamma,b$ from $\tDelta$ and the conjugate situation. In particular, this OPE has an integrable singularity 
in $z$ at $z=w$. 

By a similar argument, the OPE of $\OO^{(2)}$ with the second ``soaking field'' $\delta(c)$ behaves as $O(1)$:
\begin{equation}
\OO^{(2)}(z)\, \delta(c(w)) \underset{g=0}{\sim } -2 d^2 z  \Big( \frac{\bar{z}-\bar{w}}{z-w}\; \langle [\bar{a}, \db c], \frac{\dd}{\dd c} \delta(c)\rangle\, (w) + c.c. \Big) \; +\reg
\end{equation}

\subsection{A remark on OPEs of composite fields}\label{sss: a remark on OPEs of composite fields}
Consider the OPE $\Phi_1(z)\Phi_2(w)$  for $\Phi_1, \Phi_2$ two composite fields. It is given, according to the general principle, as a sum of Feynman graphs  $\gm$ with leaves satisfying properties (\ref{OPE subgraph (i)}), (\ref{OPE subgraph (ii)}) above. Part of the contributions come from branch graphs connecting one constituent fundamental field (or derivative of a fundamental field) $\phi_1$ from $\Phi_1$ and one (derivative of) fundamental field $\phi_2$ from $\Phi_2$ (one can think of such a contribution as a ``dressed Wick contraction'' of $\phi_1(z)$ and $\phi_2(z)$). Let us call the sum of these diagrams the ``tree part'' of the OPE, $[\Phi_1(z)\Phi_2(w)]_\mr{tree}$ -- it is readily calculated from OPEs of (derivatives of) fundamental fields.  

Generally, in addition to tree diagrams there are loop diagrams with $l\geq 1$ loops. Let us focus on the case when $\Phi_1$ and $\Phi_2$ are at most linear in fields $b, \gamma, \bar\gamma$ or their derivatives (fields of $\mc{AB}$-charge $+1$). This case is of particular relevance, since several important composite fields in the theory -- $G$, $T$, $J$, $\OO^{(2)}$ -- have this property. Under this assumption, OPE $\Phi_1(z)\Phi_2(w)$ cannot contain diagrams with $\geq 2$ loops but can contain $1$-loop diagrams of form
\begin{equation}\label{1-loop graphs for Phi_1 Phi_2 OPE linear in B}
\vcenter{\hbox{\includegraphics[scale=0.5]{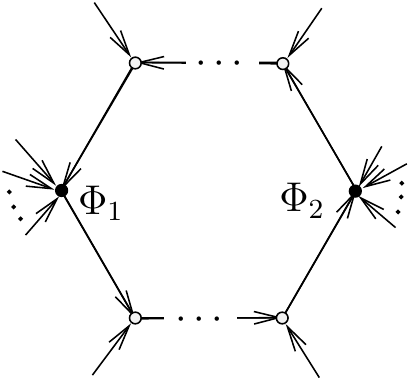}}}\quad, \quad 
\vcenter{\hbox{\includegraphics[scale=0.5]{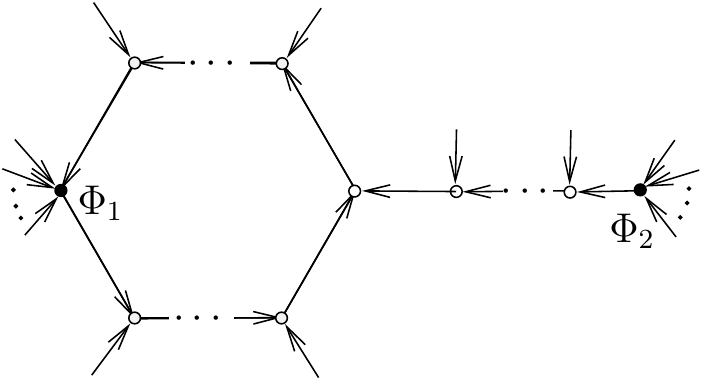}}}
\quad, \quad 
\vcenter{\hbox{\includegraphics[scale=0.5]{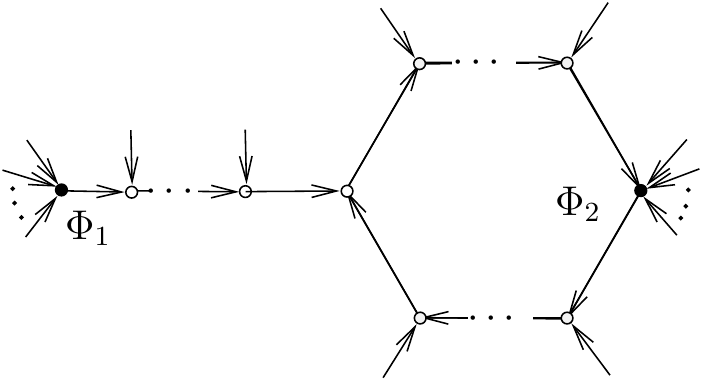}}}
\end{equation}
Note that diagrams where all vertices in the loop are $\OO^{(2)}$ cancel out by Lemma \ref{lemma: boson-fermion cancellation}.

In case when $\Phi_1$ is linear in fields $b, \gamma, \bar\gamma$ or derivatives while $\Phi_2$ does not contain them, we have 
\begin{equation}\label{OPE = tree OPE}
\Phi_1(z)\Phi_2(w) = [\Phi_1(z)\Phi_2(w)]_\mr{tree}
\end{equation}


\section{
Composite fields
} \label{ss: composite fields}

\subsection{
Building composite fields via renormalized products. Order-of-merging ambiguity. 
}\label{sss: order of collapse ambiguity}
Given two fields $\Phi_1,\Phi_2$, we define their renormalized product by the prescription (\ref{normal ordering}):
\begin{equation}\label{renormalized product Phi1 Phi2}
 (\Phi_1\Phi_2)(z)=\lim_{z'\ra z}\left(\Phi_1(z')\Phi_2(z)- \big[\Phi_1(z')\Phi_2(z) \big]_\mr{sing}\right) =\til{\lim_{z'\ra z}}\Phi_1(z')\Phi_2(z) 
\end{equation}
where we introduced the notation $\til\lim$ meaning ``subtract the singularity, then take the limit.'' 

Generally, fields $\Phi_1$, $\Phi_2$ have an OPE of the form
\begin{equation}\label{Phi_1(z') Phi(z) OPE general}
 \Phi_1(z') \Phi_2(z) \sim \sum_{p,q,r} \sigma_{pqr}(z'-z)
\, \til\Phi_{pqr}(z)+\reg 
\end{equation}
with 
\begin{equation}\label{sigma_pqr}
\sigma_{pqr}(z'-z)=(z'-z)^{-p}(\bar{z}'-\bar{z})^{-q}\log^r|z'-z|
\end{equation}
where the sum is over $p,q,r$ with $p+q\geq 0, r\geq 0$ and $(p,q,r)\neq (0,0,0)$. The singular subtraction $\big[\Phi_1(z')\Phi_2(z) \big]_\mr{sing}$ in (\ref{renormalized product Phi1 Phi2}) is defined uniquely as the r.h.s. of (\ref{Phi_1(z') Phi(z) OPE general}) without ``$\reg$'' term. 

\begin{remark}
Note that the singular subtraction in (\ref{renormalized product Phi1 Phi2}) is defined 
with respect to a local coordinate $z$, using the explicit basis (\ref{sigma_pqr}). 
In other words, renormalized product is not a diffeomorphism-invariant operation. In Section \ref{ss: vertex operators} we will see how this coordinate-dependence of the subtraction may lead to a nontrivial scaling behavior of composite fields.
\end{remark}

If instead of having $\Phi_1$ approach $\Phi_2$ in (\ref{renormalized product Phi1 Phi2}), we do the opposite and make $\Phi_2$ approach $\Phi_1$, we can get a different finite part! For instance, if 
$$ \Phi_1(z') \Phi_2(z) \sim \frac{\til\Phi(z)}{z'-z}+\Psi(z)+o(1)_{z'\ra z} = \frac{\til\Phi(z')}{z'-z}-\dd\til\Phi(z')+\Psi(z')+o(1)_{z\ra z'} $$
Then merging $\Phi_1$ with $\Phi_2$ yields $\Psi$ while merging $\Phi_2$ with $\Phi_1$ yields a different field $\Psi-\dd\til\Phi$. That is, we have an \textit{order-of-merging ambiguity}
\begin{equation}\label{ambiguity 1st order pole}
 \til{\lim_{z'\ra z}}\Phi_1(z')\Phi_2(z) - \til{\lim_{z'\ra z}}\Phi_1(z)\Phi_2(z') = \dd\til\Phi(z) 
\end{equation}
given by the derivative of the residue in the OPE between the constituent fields $\Phi_1$ and $\Phi_2$.

As an explicit example of this phenomenon, already in free (abelian) theory, at $g=0$, we have 
\begin{multline}\label{ambiguity a (gamma gamma)}
\til{\lim_{z'\ra z}}\langle a(z'),\theta\rangle  \langle (\gamma\otimes\gamma)(z),X\otimes Y\rangle-  \til{\lim_{z'\ra z}}\langle a(z),\theta\rangle  \langle (\gamma\otimes\gamma)(z'),X\otimes Y\rangle=
\\
= \langle \theta,X \rangle\, \langle \dd \gamma(z),Y\rangle + \langle \theta,Y\rangle\, \langle \dd \gamma(z),X\rangle
\end{multline}
Here $\theta\in\g^*$, $X,Y\in\g$ are arbitrary fixed vectors in the coefficient space. Note that the first term on the left corresponds to the standard normal ordering prescription in free theory -- the field $\langle :a\otimes \gamma\otimes  \gamma:\,,\theta\otimes X\otimes Y \rangle$ -- e.g., its correlator with a test field $a(x)$ is vanishing, while it is nonvanishing for the second term on the l.h.s.  
Note that in non-abelian theory, for $g\neq 0$, the result (\ref{ambiguity a (gamma gamma)}) still holds: although the OPE $a(z')(\gamma\gamma)(z)$ acquires an additional term $O(\frac{\bar{z}'-\bar{z}}{z'-z})$, it does not contribute to the ambiguity.

As another example, we have 
\begin{equation}\label{ambiguite (da) bargamma}
 \til {\lim_{z'\ra z}}\dd a(z')\otimes \bar\gamma(z) -  \til {\lim_{z'\ra z}}\dd a(z)\otimes\bar\gamma(z')= \frac{g}{2} \mr{ad}( \dd a(z) )
\end{equation}

One has the following generalization of (\ref{ambiguity 1st order pole}) for a general pair $\Phi_1,\Phi_2$, obtained by the same logic.
\begin{lemma}
For $\Phi_1,\Phi_2$ any pair of composite fields with OPE given by  (\ref{Phi_1(z') Phi(z) OPE general}), the order-of-merging ambiguity in the product $\Phi_1\Phi_2$ is:
\begin{equation}\label{ambiguity - general}
 \til{\lim_{z'\ra z}}\Phi_1(z')\Phi_2(z) - \til{\lim_{z'\ra z}}\Phi_1(z)\Phi_2(z') =\sum_{p,q\geq 0,\, (p,q)\neq (0,0)} \frac{(-1)^{p+q-1}}{p!q!}\, \dd^p \db^q \til\Phi_{pq0}(z)
\end{equation}
\end{lemma}
Note that terms in the OPE involving logarithms or involving positive powers, like $\frac{\bar{z}'-\bar{z}}{z'-z}$, do not contribute to the ambiguity.

\noindent \textbf{Open question.}
Does the \textit{pre-Lie algebra} identity hold
\begin{multline}\label{pre-Lie identity}
  \Phi_1*_R (\Phi_2 *_R \Phi_3 )- (-1)^{|\Phi_1|\cdot|\Phi_2|}\Phi_2*_R (\Phi_1 *_R \Phi_3 ) 
  \\
  \stackrel{?}{=}
(\Phi_1*_R \Phi_2 - (-1)^{|\Phi_1|\cdot|\Phi_2|}\Phi_2 *_R \Phi_1)*_R \Phi_3
\end{multline}
for any triple $\Phi_1,\Phi_2,\Phi_3$? Here we denoted $(\phi*_R \psi)(z)=\displaystyle{\til{\lim_{z'\ra z}}}\phi(z')\psi(z)$ the renormalized product merging the left factor onto the right factor; $|\phi|$ is the ghost number of the field $\phi$. Note that the field $1$ serves as left- and right-unit for the product $*_R$. Identity (\ref{pre-Lie identity}) holds in any \textit{chiral} CFT, see Appendix 6.C in \cite{DMS}; of course, our case of non-abelian $BF$ is non-chiral and we cannot use that result.

The following is a special case of (\ref{pre-Lie identity}) which easy to prove independently;  we will need it for our analysis of conservation laws under the correlator in Section \ref{ss: quantum conservation laws}.
\begin{lemma}\label{lemma: merging A-fields onto a B-field}
Let $\Phi_1,\ldots,\Phi_n$ be a collection of fundamental fields (or their derivatives) of $\mc{AB}$-charge $-1$  and $\Psi$ a fundamental field (or derivative) of $\mc{AB}$-charge $+1$. Then
\begin{equation}\label{merging several A-type fields onto a B-type field}
\til{\lim_{z_1\ra z}} \cdots \til{\lim_{z_n\ra z}}  \Phi_1(z_1)\cdots \Phi_n(z_n)\Psi(z)=
\til{\lim_{z'\ra z}}\big(\Phi_1\cdots\Phi_n \big)(z') \Psi(z)
\end{equation}
In particular, the resulting composite field is independent of the order in which one merges fields $\Phi_i$ onto $\Psi$. Field $(\Phi_1\cdots\Phi_n)$ appearing in the r.h.s. is independent of the order of merging, since fields $\Phi_i$ have regular OPE with each other.
\end{lemma}
\begin{proof}
We give a proof for the case $n=2$; the case of general $n$ is similar.  Consider the correlator 
$$F(z_1,z_2,z;x_1,\ldots,x_m)=\lan \Phi_1(z_1) \Phi_2(z_2) \Psi(z) \phi(x_1)\ldots \phi(x_m) \ran $$
with $\{\phi_i\}$  an arbitrary collection of test fields. The correlator is a sum of
\begin{enumerate}
\item diagrams where $\Phi_1$ and $\Psi$ belong to the same tree and $\Phi_2$ belongs to another tree,
\item diagrams where $\Phi_2$ and $\Psi$ belong to the same tree and $\Phi_1$ belongs to another one,
\item diagrams where $\Phi_1$, $\Phi_2$ and $\Psi$ belong to $3$ different trees.
\end{enumerate}  
Thus, the correlator has the following structure:
\begin{equation}\label{F(z_1,z_2,z)}
F(z_1,z_2,z) = \sum_{k}G_k(z_1,z) H_k(z_2) +\sum_{l}\til{G}_l(z_2,z) \til{H}_l(z_1) + K(z_1,z_2,z)
\end{equation}
where $K$ has no singularities when any pair among $z_1,z_2,z$ collides; we are suppressing the dependence on $x_1,\ldots,x_m$ in the notation. Merging first $z_2$ onto $z$ and then $z_1$ onto $z$, we obtain
$$\til{\lim_{z_1\ra z}}\,\til{\lim_{z_2\ra z}}F(z_1,z_2,z) =  \sum_{k}\til{\lim_{z_1\ra z}}G_k(z_1,z)\,H_k(z) + \sum_{l}\til{\lim_{z_2\ra z}} \til{G}_l(z_2,z) \,\til{H}_l(z)+K(z,z,z)  $$
Setting $z_1=z_2=z'$ in (\ref{F(z_1,z_2,z)}) and then evaluating $\displaystyle{\til{\lim_{z'\ra z}}}$, we obtain the same result. Thus, we checked (\ref{merging several A-type fields onto a B-type field}) for $n=2$ by probing both sides by a correlator with a collection of  test fields.
\end{proof}

The derivative of a renormalized product is defined in the natural way:
\begin{equation}\label{derivative of a renormalized product}
\dd \left(\til{\lim_{z'\ra z}} \Phi_1(z')\Phi_2(z)\right) = \til{\lim_{z'\ra z}} \Big(\dd \Phi_1(z')\Phi_2(z)+\Phi_1(z')\dd\Phi_2(z)\Big)
\end{equation}
and similarly for $\db$ of a product. Here it is crucial that the terms on the right, arising from Leibnitz rule, respect the order of merging in the product $\Phi_1\Phi_2$ we take the derivative of. The following property is immediate from this definition.
\begin{lemma}\label{lemma: ambiguity of a derivative}
Given two fields $\Phi_1,\Phi_2$, we have
\begin{multline}\label{ambiguity of a derivative}
\til{\lim_{z'\ra z}} \left(\dd \Phi_1(z')\Phi_2(z)+\Phi_1(z')\dd\Phi_2(z)\right) - 
\til{\lim_{z'\ra z}} \left(\dd \Phi_1(z)\Phi_2(z')+\Phi_1(z)\dd\Phi_2(z')\right)  \\
=
\dd \left(\til{\lim_{z'\ra z}} \Phi_1(z')\Phi_2(z) - \til{\lim_{z'\ra z}} \Phi_1(z)\Phi_2(z') \right)
\end{multline}
i.e., the ambiguity in the derivative $\dd(\Phi_1\Phi_2)=\dd\Phi_1\,\Phi_2+\Phi_1\,\dd\Phi_2$ is the derivative of the ambiguity of the product $\Phi_1\Phi_2$.
The same holds if we replace $\dd$ with $\db$.
\end{lemma}

In summary, we have the following.
\begin{itemize}
\item A composite field built as a renormalized product of several fundamental fields (or their derivatives) 
$$ (\phi_1\cdots \phi_n)_{\mu} $$
must be decorated with order-of-merging data $\mu$, prescribing in which fields merge onto which and in what order. Generally, such data can be given by a planar binary rooted tree with $n$ leaves decorated by some permutation $\sigma$ of $\phi_1,\ldots,\phi_n$, where at each vertex the left incoming field merges onto the right one (as a possible convention). 
$$\vcenter{\hbox{\includegraphics[scale=0.65]{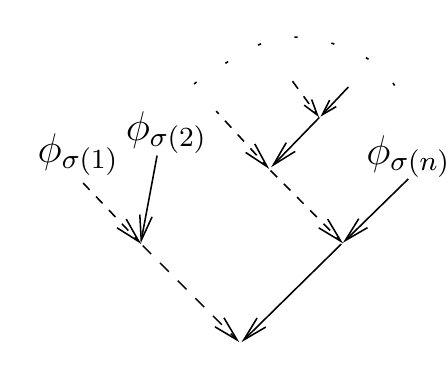}}}$$
Here a solid incoming edge at a vertex represents the field onto which the merging occurs.
\item As a special case of order-of-merging data, one may pick one of $\phi_k$'s as a ``base'' and consecutively merge other fields onto it. 
The limit $g=0$  of such a renormalized product coincides with the normally ordered product $:\phi_1\cdots\phi_n:$ of the free theory (and in particular is independent of the order in which fields are merged onto the ``base''; at $g\neq 0$ the result can depend on the order).
\item There are many examples of composite fields which turn out to be independent of the order of merging. For instance:
\begin{itemize}
\item The product of any two fields from the list $\{a,\bar{a},\gamma,\bar\gamma,c,b,\dd b, \db b\}$. (However, taking further derivatives can create a dependence on the order, as in (\ref{ambiguite (da) bargamma})).
\item Fields $\OO^{(2)}$, $J$, $G$, $T$ and complex conjugates, see Proposition \ref{prop: vanishing subtractions in G,T,J} below.
\item Expressions vanishing by equations of motion -- left hand sides of (\ref{eom in cx fields}).
\item If a field $\Phi$ is independent of the order of merging, then any derivative $\dd^p\db^q \Phi$ is independent too, by Lemma \ref{lemma: ambiguity of a derivative}, \textit{as long as the order of merging is the same in all terms of $\dd^p\db^q \Phi$ produced by Leibnitz rule}.\footnote{
To illustrate the importance of the last condition, consider the derivative $\dd(a\otimes \bar\gamma)$ of an ordering-independent field $a\otimes \bar\gamma$. If we choose an inconsistent order of merging between the two terms, $\displaystyle{\til{\lim_{z'\ra z}}}\big(\dd a(z')\otimes \bar\gamma(z)-a(z)\otimes\dd\bar\gamma(z')\big)$, then it differs by a defect $\frac{g}{2}\ad(\dd a)$ from the consistent ordering and by twice that defect from the opposite inconsistent one.
}
\end{itemize}
\end{itemize}

\subsection{$G,T,J$ as composite fields}

When we consider fields $G,T,J$ as composite fields, the corresponding singular subtractions miraculously vanish.

Indeed, consider the regularization of the stress-energy tensor by splitting the constituent fields:
\begin{equation}\label{T_split}
T^\mr{split}(z',z)= \langle \dd \gamma(z'), a(z)\rangle + \langle\dd b(z'), \dd c(z)\rangle +\frac{g}{2} \langle \dd b(z'), [a(z), c(z)] \rangle 
\end{equation}
Note that, since the OPE between $a$ and $c$ is regular, we can put them in the same point. The singular part of (\ref{T_split}) at $z'\ra z$, as calculated using the OPEs (\ref{a  dd gamma OPE}), (\ref{dd c  dd b OPE}), is:
\begin{multline}
\big[T^\mr{split}(z',z)\big]_\mr{sing}=\\
=\frac{\dim\g}{(z'-z)^2}+\frac{g^2}{4}\Big(-\frac{\bar{z}'-\bar{z}}{z'-z}\,K(a,\bar{a})+\frac12 
\frac{(\bar{z}'-\bar{z})^2}{(z'-z)^2}\,K(\bar{a},\bar{a})\Big)-\\
-\frac{\dim\g}{(z'-z)^2}+\frac{g^2}{4}\Big( 2\log|z'-z|\,K(a, a) + 2\frac{\bar{z}'-\bar{z}}{z'-z}\,K(a, \bar{a})- \frac12 
\frac{(\bar{z}'-\bar{z})^2}{(z'-z)^2}\,K(\bar{a},\bar{a}) \Big) -\\
-\frac{g^2}{4}\Big(\frac{\bar{z}'-\bar{z}}{z'-z}\,K(a ,\bar{a})+2\log|z'-z|\,K(a, a) \Big)\qquad = 0
\end{multline}
Here $K(X,Y)=\mr{tr}_\g \mr{ad}(X)\mr{ad}(Y)$ is the Killing form; all $O(g)$ terms vanish by unimodularity. All fields on the right are at $z$. Thus, the \textit{total singular subtraction in (\ref{T_split}) vanishes} and the renormalized stress-energy tensor is simply
$$ T(z)= \lim_{z'\ra z} T^\mr{split}(z',z)$$

Likewise, we regularize $J$ as
\begin{equation}
J^\mr{split}(z',z)=  \langle\gamma(z'),\dd c(z)\rangle+g  \langle\gamma(z'), [a(z), c(z)]\rangle-\frac{g}{4} \langle \dd b(z'), [c(z) c(z)]\rangle
\end{equation}
Here the singular subtraction is calculated using (\ref{dd c  gamma OPE}):
\begin{multline}
\big[J^\mr{split}(z',z)\big]_\mr{sing}= -\frac{g^2}{4}\Big(2\log|z'-z|\,K(a, c)+\frac{\bar{z}'-\bar{z}}{z'-z}\,K(\bar{a}, c) \Big)+ \\
+
\frac{g^2}{4}\Big(2 \frac{\bar{z}'-\bar{z}}{z'-z}\,K(\bar{a}, c)+4\log|z'-z|\,K(a, c) \Big)-\\
-\frac{g^2}{4}K\Big(2\log|z'-z|\,K(a, c)+\frac{\bar{z}'-\bar{z}}{z'-z}\,K(\bar{a}, c) \Big)\qquad =0
\end{multline}
The total singular subtraction vanishes again and thus the renormalized $J$ field is just
$$ J(z) =\lim_{z'\ra z} J^\mr{split}(z',z) $$

The case of the field 
$$G(z)=\langle a(z), \dd b(z)\rangle$$ 
is trivial: the OPE between $a$ and $\dd b$ is regular, so we can safely put the fields at the same point. I.e., again we have a vanishing singular subtraction, but in the case of $T,J$ the vanishing was a nontrivial cancellation between subtractions for different terms in the composite field, while for $G$ it vanishes on the nose.

Furthermore, consider the field $\OO^{(2)}$. We regularize it as
\begin{multline*}
\OO^{(2)\mr{split}}(z',z)=\\
=-2 d^2 z  \Big(
\langle (\gamma-\bar\gamma)(z'),[a(z),\bar{a}(z)]\rangle+
\langle\dd b(z'), [\bar{a}(z), c(z)]\rangle+
\langle \db b(z'),[ a(z), c(z)]\rangle
\Big)
\end{multline*}
One finds the singular subtraction to be
\begin{multline*}
\big[\OO^{(2)\mr{split}}(z',z)\big]_\mr{sing} = 
-g\, d^2 z \Big(
\big(\frac{z'-z}{\bar{z}'-\bar{z}}K(a, a)+4\log|z'-z|K(a,\bar{a}) + \frac{\bar{z}'-\bar{z}}{z'-z}K(\bar{a}, \bar{a}) \big)- \\
-\big(2\log|z'-z|K(a,\bar{a})+\frac{\bar{z}'-\bar{z}}{z'-z}K(\bar{a}, \bar{a})\big)-
\big(\frac{z'-z}{\bar{z}'-\bar{z}}K(a, a)+2\log|z'-z|K(a,\bar{a})\big)
\Big)\qquad =0
\end{multline*}

Finally, consider the equations of motion -- left hand sides in (\ref{eom in cx fields}) -- as composite fields. They all have zero singular subtractions on the nose except for the field $\db\gamma+\cdots$ and its complex conjugate. In this case, we have
\begin{multline} 
\Big[ \big\langle \db \gamma(z')-\frac{g}{2}[(\gamma-\bar\gamma)(z'),\bar{a}(z)]-\frac{g}{2} [\db b(z'), c(z)] ,X\big\rangle \Big]_\mr{sing}=\\ =
\frac{g^2}{4}K\Big(\frac{z'-z}{\bar{z}'-\bar{z}}\,a+2\log|z'-z|\,\bar{a},X\Big)-\frac{g^2}{4}K\Big(\frac{z'-z}{\bar{z}'-\bar{z}}\,a+2\log|z'-z|\,\bar{a},X\Big)\qquad =0
\end{multline}
-- and again we have a cancellation for the singular subtraction. Here $X\in\g$ is an arbitrary vector. Thus, left hand sides in (\ref{eom in cx fields}) all have zero singular subtractions as composite fields. 

Note that in all the cases we considered here we did not encounter terms of form $\frac{\til\Phi}{(z'-z)^p(\bar{z}'-\bar{z})^q}$ among the terms in the singular subtractions, with $p,q\geq 0$ and $\til\Phi$ a non-constant field. This implies that all these composite fields are independent of the order of merging.

In summary, we have proved the following.
\begin{proposition}\label{prop: vanishing subtractions in G,T,J}
Fields $G,T,J$ (and their complex conjugates) viewed as composite fields have the following properties:
\begin{enumerate}[(a)]
\item They are independent of the order of merging of the constituent fundamental fields.
\item The total singular subtraction vanishes.
\end{enumerate}
The same applies to $\OO^{(2)}$ and to equations of motion -- left hand sides of (\ref{eom in cx fields}).
\end{proposition}

\subsection{Examples of correlators and OPEs of composite fields}
As a first example, consider the $2$-point correlation function 
\begin{equation}\label{<(a bargamma) gamma>}
\lan (a\otimes \bar\gamma)(z)\otimes \gamma(w) \ran 
\end{equation}
The composite field $a\otimes \bar\gamma$ is defined by the prescription (\ref{normal ordering}) -- by placing the two constituent fundamental fields into distinct nearby points and subtracting the singular part of their OPE (\ref{a bargamma OPE}):
\begin{equation}
\begin{aligned}
\big(a\otimes \bar\gamma\big)(z)&=\lim_{z'\ra z} \left(a(z')\otimes \bar\gamma(z) - \big[a(z')\otimes \bar\gamma(z) \big]_\mr{sing}\right) \\ 
&=\lim_{z'\ra z} \left( a(z')\otimes \bar\gamma(z)-g  \log|z'-z|\, \mr{ad}(a(z)) \right)
\end{aligned}
\end{equation}
Thus, the correlator (\ref{<(a bargamma) gamma>}) is:
\begin{equation}\label{<(a bargamma) gamma> result}
\begin{aligned}
\langle (a\otimes \bar\gamma)(z)&\otimes \gamma(w) \rangle  =\\
=& \lim_{z'\ra z}\left(\lan a(z')\otimes \bar\gamma(z)\otimes \gamma(w) \ran - g \log|z'-z| \lan \mr{ad}(a(z))\otimes  \gamma(w) \ran \right) \\
=&\lim_{z'\ra z} \left( g \ff\cdot \frac{1}{z'-w}\log\left|\frac{w-z}{z'-z}\right|+g \ff\cdot \log|z'-z|\frac{1}{z-w} \right) \\
=& g \ff\cdot \frac{\log|z-w|}{z-w}
\end{aligned}
\end{equation}
Here we used the result (\ref{3p_1}) for the $3$-point function of fundamental fields.

Similarly, for the correlator $\lan (a\otimes \gamma)(z)\otimes \bar\gamma(w) \ran$ we find
\begin{equation}
\begin{aligned}
\langle (a\otimes \gamma)(z)&\otimes \bar\gamma(w) \rangle  = \\
=& \lim_{z'\ra z} \lan \left(a(z')\otimes \gamma(z)-\frac{\idg}{z'-z}+\frac{g}{2}\,\frac{\bar{z}'-\bar{z}}{z'-z}\mr{ad}(\bar{a}(z))\right)\otimes \bar\gamma(w) \ran \\
=& \lim_{z'\ra z} \left( g\ff\cdot\frac{1}{z-z'}\log\left|\frac{z-w}{z'-w}\right| -\frac{g}{2} \ff\cdot \frac{\bar{z}'-\bar{z}}{z'-z}\,\frac{1}{\bar{z}-\bar{w}}\right) \\
=& \frac{g}{2} \ff\cdot\frac{1}{z-w}
\end{aligned}
\end{equation}

As the next example, consider the following correlator of two composite fields: 
$$\big\langle \langle (a\otimes \bar\gamma)(z),\theta\otimes X\rangle\, \langle (\gamma\otimes \bar\gamma)(w) , Y\otimes Z\rangle \big\rangle$$
with $\theta\in \g^*$, $X,Y,Z\in\g$ four fixed coefficient vectors.
We can obtain this correlator from the $4$-point function (\ref{4-point a bargamma gamma bargamma}) by collapsing the first pair of points and the last pair of points (and subtracting the singularities). Collapsing $a$ and $\bar\gamma$, we get the $3$-point function
\begin{multline}\label{<(a bargamma) gamma bargamma>}
\big\langle \langle(a\otimes\bar\gamma)(z),\theta\otimes X\rangle\, \langle \gamma(w_1),Y\rangle\,\langle \bar\gamma(w_2),Z\rangle \big\rangle   = \\
=  \lim_{z'\ra z} \big\langle \langle a(z')\otimes \bar\gamma(z)-g \log|z'-z|\, \mr{ad}(a(z)),\theta\otimes X\rangle\, \langle \gamma(w_1),Y\rangle\,\langle \bar\gamma(w_2),Z\rangle \big\rangle \\
= \frac{g^2 \langle\theta,[X,[Y,Z]]\rangle}{2(z-w_1)}\left(-i\DD\left(\frac{w_1-w_2}{z-w_2}\right)-\log\left|\frac{w_1-w_2}{z-w_2}\right|\cdot \log|(z-w_1)(z-w_2)|\right)\\
+
\frac{g^2 \langle \theta,[Z,[Y,X]] \rangle}{2(z-w_1)}\left(i\DD\left(\frac{w_1-w_2}{z-w_2}\right)-\log\left|\frac{w_1-w_2}{z-w_2}\right|\cdot \log\left|\frac{z-w_1}{z-w_2}\right|\right)
\end{multline}
Then, collapsing $w_1$ and $w_2$, 
we get 
\begin{multline}\label{<(a bargamma) (gamma bargamma)>}
\big\langle \langle(a\otimes \bar\gamma)(z),\theta\otimes X\rangle\,\langle (\gamma\otimes \bar\gamma)(w) , Y\otimes Z\rangle\big\rangle = \\
=  \lim_{w'\ra w} \big\langle \langle(a\otimes \bar\gamma)(z),\theta\otimes X\rangle \,\langle\gamma(w')\otimes \bar\gamma(w)+g\log|w'-w|\cdot \langle\ff, (\gamma-\bar\gamma)(w)\rangle,Y\otimes Z\rangle \big\rangle \\
= g^2 \langle\theta,[X,[Y,Z]]\rangle\cdot \frac{\log^2 |z-w|}{z-w}
\end{multline}
Feynman diagrams corresponding to (\ref{<(a bargamma) gamma bargamma>}) are (\ref{<(a bargamma) (gamma bargamma)>}):
$$
\vcenter{\hbox{\includegraphics[scale=0.5]{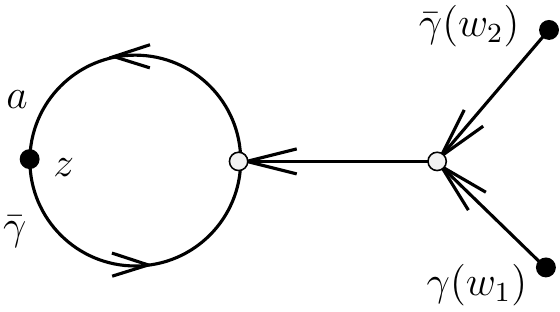}}}\quad, \quad 
\vcenter{\hbox{\includegraphics[scale=0.5]{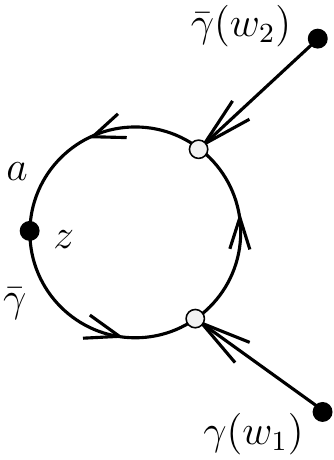}}}\quad;\quad
\vcenter{\hbox{\includegraphics[scale=0.5]{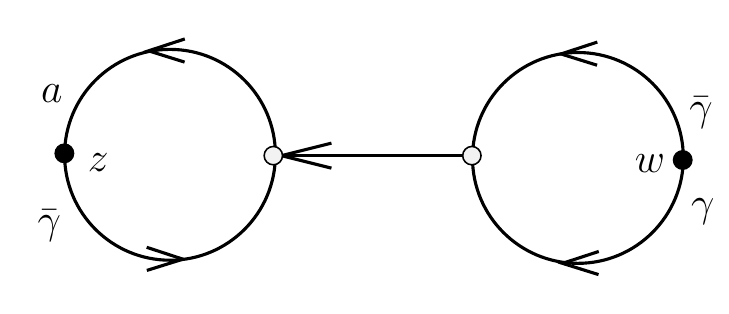}}}\quad,\quad
\vcenter{\hbox{\includegraphics[scale=0.5]{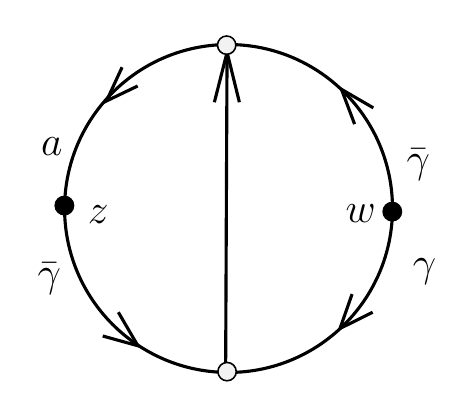}}}
$$
The last diagram here is, in fact, vanishing. Note that (\ref{<(a bargamma) (gamma bargamma)>}) corresponds to a \textit{two-loop} diagram.

Next, consider the OPE 
$$ \langle(a\otimes \bar\gamma)(w),\theta\otimes X\rangle \,\langle \bar\gamma(z),Y\rangle $$
with $\theta\in \g^*$ and $X,Y\in\g$.
There are the following contributing Feynman diagrams: 
$$
 \vcenter{\hbox{\includegraphics[scale=0.5]{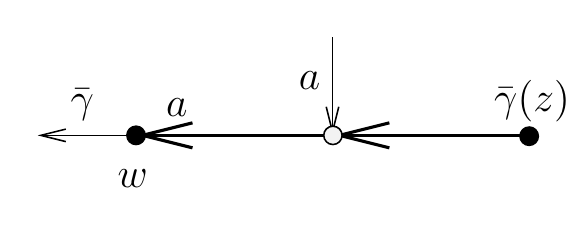}}} \quad , \quad 
\vcenter{\hbox{\includegraphics[scale=0.5]{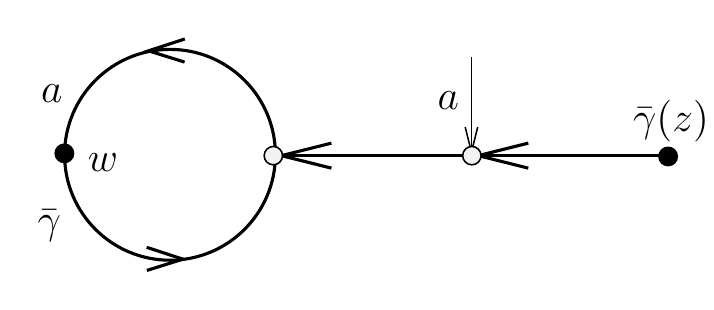}}} \quad, \quad
\vcenter{\hbox{\includegraphics[scale=0.5]{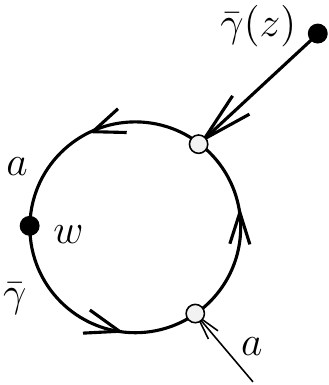}}}
$$
They 
give the following result:
\begin{multline}
\langle(a\otimes \bar\gamma)(w),\theta\otimes X\rangle\,\langle\bar\gamma(z),Y\rangle 
\sim  g \log|z-w|\, \big(\langle a, [Y,\theta]\rangle\, \langle\bar\gamma,X  \rangle\big)(z)+\\
+\frac{g^2}{2} \log^2|w-z|\,\big(\underbrace{-\langle \theta,[X,[Y,a(z)]]  \rangle+\langle \theta,[Y,[X,a(z)]]  \rangle}_{\langle a(z), [[X,Y],\theta]\rangle} \big)
+\reg
\end{multline}
Here the contribution of the two last diagrams simplifies by Jacobi identity. This OPE gives a singular subtraction needed to define the composite field with three constituent fundamental fields
\begin{equation}
(a\otimes \bar\gamma\otimes \bar\gamma)(z)=\lim_{z'\ra z} \left((a\otimes \bar\gamma)(z')\otimes\bar\gamma(z)-\big[(a\otimes \bar\gamma)(z')\otimes\bar\gamma(z)\big]_\mr{sing} \right)
\end{equation}
Its correlator with $\gamma$ is obtained by collapsing $z$ with $w_2$ in (\ref{<(a bargamma) gamma bargamma>}):
\begin{multline}\label{<(a bargamma bargamma) gamma>}
\big \langle \langle (a\otimes \bar\gamma\otimes \bar\gamma)(z),\theta\otimes X\otimes Y\rangle\, \langle\gamma(w),Z\rangle \big\rangle = \\
=\frac{g^2}{2}
(\langle\theta,[X,[Y,Z]] \rangle+\langle\theta,[Y,[X,Z]]\rangle)
\frac{\log^2|z-w|}{z-w}
\end{multline}

\subsection{Correlators involving the field $\bar\gamma\cdots \bar\gamma$}\label{sss: correlators with bargamma...bargamma}
Here we give some examples of correlators containing an arbitrary power of $\log$. These results will be the starting point for the construction of ``vertex operators'' -- composite fields with a quantum correction to conformal dimension -- in Section \ref{ss: vertex operators}.

\begin{lemma}
The $3$-point correlation function of the composite field $\bar\gamma\cdots\bar\gamma$ with $a$ and $\gamma$ is:
\begin{multline}\label{3-point a (bargamma ... bargamma) gamma}
\big\langle \langle (\bar\gamma\otimes\cdots \otimes\bar\gamma)(z),X_1\otimes\cdots \otimes X_n \rangle\; a(w_1)\otimes \gamma(w_2) \big\rangle = \\
=  \frac{g^n}{n!} \left(\sum_{\sigma\in S_n}
\ad(X_{\sigma(1)})\ad(X_{\sigma(2)})\cdots \ad(X_{\sigma(n)})
\right) \frac{1}{w_1-w_2} \log^n\left|\frac{z-w_2}{z-w_1}\right|
\end{multline} 
where $X_1,\ldots,X_n\in\g$ are arbitrary fixed vectors, the correlator is understood as valued in $\g\otimes\g^*$ and the sum on the right  goes over permutations $\sigma$.
\end{lemma}

\begin{proof}
One proves this by first considering the correlator
\begin{multline}\label{<(a bargamma^n) gamma> split}
\big\langle  \langle \bar\gamma(z_1),X_1\rangle\cdots \langle \bar\gamma (z_n),X_n\rangle\; a(w_1)\otimes\gamma(w_2) \big\rangle =\\
=g^n\left(\sum_{\sigma\in S_n}
\ad(X_{\sigma(1)})\ad(X_{\sigma(2)})\cdots \ad(X_{\sigma(n)})
 \right) \mathbb{F}_n(w_1,z_1,\ldots, z_n,w_2)  
 \end{multline}
 where 
\begin{equation}\label{F_n from <(a bargamma^n) gamma>}
 \mathbb{F}_n(w_1,z_1,\ldots, z_n,w_2) = \int \frac{d u_1}{2\pi}\cdots \frac{du_n}{2\pi} 
\frac{(-1)^n}{
\prod_{k=1}^n (u_{k-1}-u_k) (\bar{u}_k-\bar{z}_k)\cdot
(u_n-w_2)}
\end{equation}
where we set $u_0:=w_1$. 
Here the contributing diagrams are:
$$ 
 \vcenter{\hbox{\includegraphics[scale=0.5]{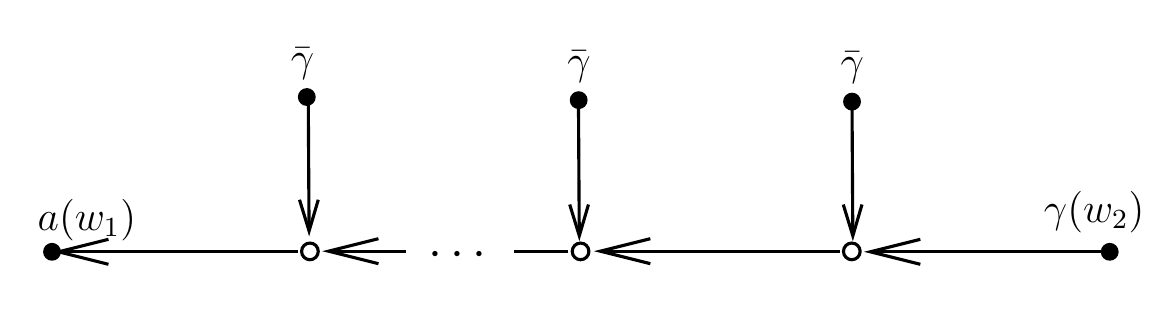}}}
$$
where we need to sum over orders in which $\bar\gamma$'s are connected (hence the sum over $\sigma\in S_n$ above).
Next, we set $z_1=\cdots=z_n=z$ (note that the integral is convergent in this limit -- there are no singularities to be subtracted when merging $z_i$'s):
$$\mathbb{F}^\mr{merged}_n(w_1,z,w_2) =  \mathbb{F}_n(w_1,z,\ldots, z,w_2)$$ 
We note that functions $\mathbb{F}^\mr{merged}_n$ satisfy a recursion in $n$:
$$   \mathbb{F}^\mr{merged}_n (w_1,z,w_2) =- \int \frac{d^2 u}{2\pi}\frac{\mathbb{F}^\mr{merged}_{n-1} (w_1,z,u)}{(\bar{u}-\bar{z})(u-w_2)} $$
as follows from the form of the integrals (\ref{F_n from <(a bargamma^n) gamma>}). This allows us to check by induction in $n$ that 
\begin{equation}\label{F_n^collapsed}
\mathbb{F}^\mr{merged}_n (w_1,z,w_2) =\frac{1}{n!}\,\frac{1}{w_1-w_2} \log^n\left|\frac{z-w_2}{z-w_1}\right|
\end{equation}
\end{proof}

Merging the field $\bar\gamma\cdots \bar\gamma$ with either $a$ or $\gamma$ in (\ref{3-point a (bargamma ... bargamma) gamma}) and subtracting the singularity results in following the $2$-point functions:
\begin{gather}
\big\langle \langle(a\otimes \bar\gamma\otimes\cdots\otimes \bar\gamma)(z), X_1\otimes\cdots\otimes X_n \rangle\, \gamma(w) \big\rangle = \label{2-point (a bargamma...bargamma) gamma}\\ 
\nonumber =\frac{g^n}{n!} \left(\sum_{\sigma\in S_n}
\ad(X_{\sigma(1)})\ad(X_{\sigma(2)})\cdots \ad(X_{\sigma(n)})
 \right) \frac{\log^n|z-w|}{z-w},  \\
\big\langle  a(z)\, \langle (\bar\gamma\otimes \cdots\otimes  \bar\gamma\otimes \gamma)(w),X_1\otimes\cdots\otimes X_n\rangle \big\rangle =
\label{2-point a (bargamma...bargamma gamma)}\\
\nonumber =
(-1)^n\frac{g^n}{n!} \left(\sum_{\sigma\in S_n} 
\ad(X_{\sigma(1)})\ad(X_{\sigma(2)})\cdots \ad(X_{\sigma(n)})
\right) \frac{\log^n|z-w|}{z-w}
\end{gather}
Here we understand that the $n$ $\bar\gamma$-factors are contracted with the vectors $X_i\in\g$ while $a$ and $\gamma$ are left non-contracted. Correlator (\ref{2-point (a bargamma...bargamma) gamma})  is a generalization of the results  (\ref{<(a bargamma) gamma> result}), (\ref{<(a bargamma bargamma) gamma>}).

\section{Conformal and $Q$-invariance on the quantum level}\label{s: conf and Q-invariance}

\subsection{Equations of motion under the correlator and contact terms}
Consider the correlator
$$ \lan \db a(z) \phi_1(x_1)\cdots \phi_n(x_n) \ran $$
with $\phi_1,\ldots,\phi_n$ some test fields (assumed to be fundamental) inserted at points $x_1,\ldots,x_n$ distinct from $z$. Contributing Feynman graphs are binary trees with $\db a(z)$ at the root and $\phi_1(x_1),\ldots,\phi_n(x_n)$ decorating the leaves. The edge connecting the root with $\gamma$ from the interaction vertex $\OO^{(2)}(u)$ gets assigned 
$$
\db_z \frac{1}{z-u}\cdot \idg =\pi \delta(z-u) \cdot \idg$$ 
This implies that
\begin{multline}\label{db a under correlator}
\lan \db a(z) \phi_1(x_1)\cdots \phi_n(x_n) \ran = \\ =\sum_{N\geq 0}\frac{(-g/4\pi)^N}{(N-1)!} \Big\langle \wick{ \c1{\db a(z)} \big(\int_u \c1{\OO^{(2)}(u)}\big)\prod_{i=1}^{N-1} \big(\int_{u_i} \OO^{(2)}(u_i)\big)  \phi_1(x_1)\cdots \phi_n(x_n) } \Big\rangle_0
\\=
\frac{g}{2} \lan [a, \bar{a}](z)\; \phi_1(x_1)\cdots \phi_n(x_n) \ran
\end{multline}
Graphically:
$$ \vcenter{\hbox{\includegraphics[scale=0.4]{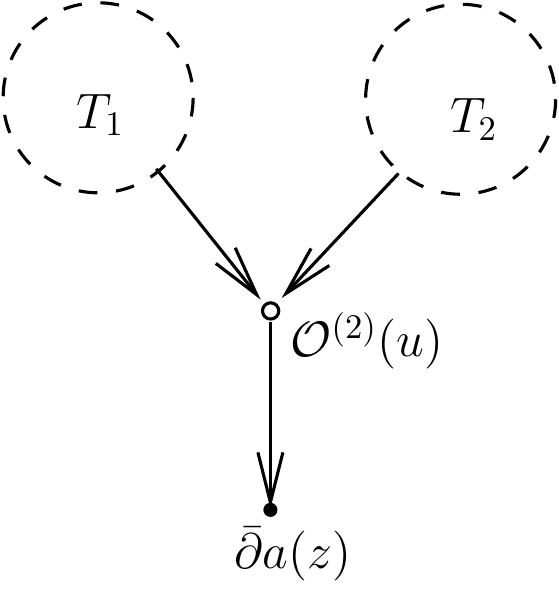}}} \quad = \quad  \vcenter{\hbox{
\includegraphics[scale=0.4]{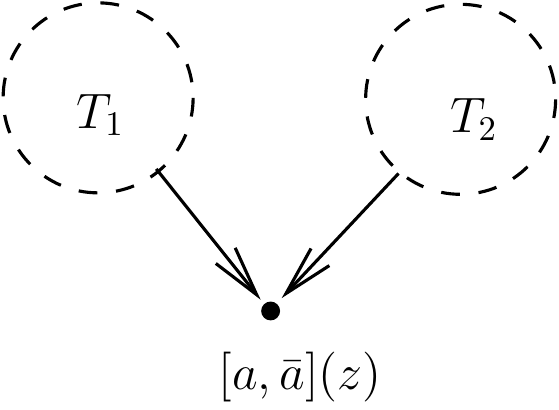}}} $$
with $T_1,T_2$ arbitrary trees with leaves decorated by the test fields. Thus, graphically, integrating over $u$ the delta-function arising in $\db$ of the propagator, results in chopping off the root of the tree. So, we obtained the identity 
\begin{equation*}  
\lan \left(\db a-\frac{g}{2}[a,\bar{a}]\right)(z)\;\; \phi_1(x_1)\cdots \phi_n(x_n) \ran = 0
\end{equation*}
Here the field in the brackets vanishes by classical equations of motion (\ref{eom in cx fields}). Our result here is that it holds in the quantum world: correlators of this field with any collection of test fields vanish.  This graphic argument for equations of motion under the correlator appeared in \cite{AI}.

A point related to this calculation is that the free theory OPE
\begin{equation}\label{O^2 db a OPE}
\OO^{(2)}(u)\, \db a(z) \underset{g=0}{\sim} -2\pi d^2u\; [a,\bar{a}](u)\,\delta(u-z)+\reg 
\end{equation}
contains a \textit{contact term}\footnote{By contact terms we generally mean terms containing delta-functions (or derivatives of delta-functions) in positions of fields} singularity. Normally when considering OPEs we require the fields to be at non-coinciding points. However, non-abelian theory is constructed as abelian theory with arbitrarily many insertions of $\OO^{(2)}$ which can hit other observables. Therefore, when talking about OPEs involving $\OO^{(2)}$, we should allow it to hit the other field, and we should care about contact terms.

Generally, we say that a composite field $\Xi$ is a \textit{quantum equation of motion} if it vanishes under the  correlator with an arbitrary collection of test fields inserted away from $\Xi$.
$$ \lan \Xi(z) \; \phi_1(x_1)\cdots \phi_n(x_n)
  \ran =0  $$
Thus, we just showed that $\Xi=\db a -\frac{g}{2} [a,\bar{a}]$ is a quantum equation of motion.

Similarly to (\ref{db a under correlator}), for the correlator of $\db \gamma(z)$ with test fields we find
\begin{equation}\label{db gamma correlator}
\lan \db\gamma(z)\;  \phi_1(x_1)\cdots \phi_n(x_n)\ran =  \frac{g}{2}\lan \big(-[\bar{a},\gamma-\bar\gamma]+[c,\db b]\big) \;  \phi_1(x_1)\cdots \phi_n(x_n)\ran 
\end{equation}
Graphically:
$$  \vcenter{\hbox{\includegraphics[scale=0.4]{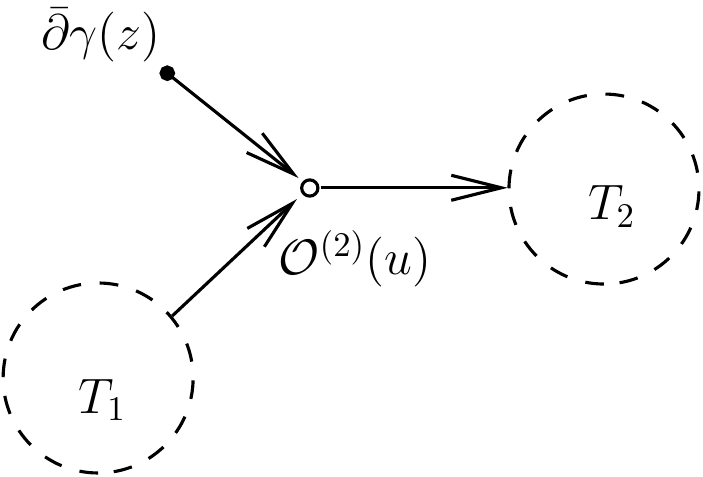}}} \quad = \quad  \vcenter{\hbox{\includegraphics[scale=0.4]{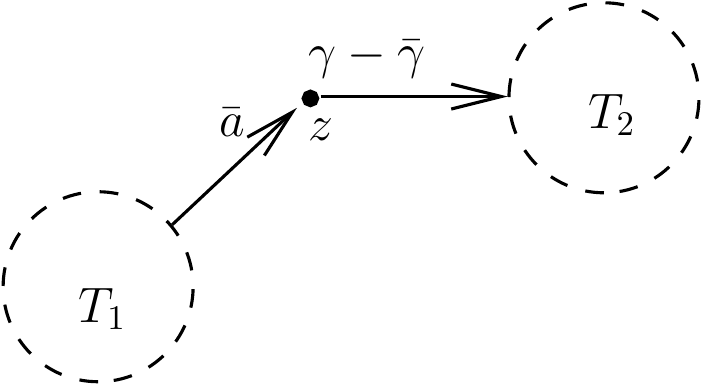}}} + \vcenter{\hbox{\includegraphics[scale=0.4]{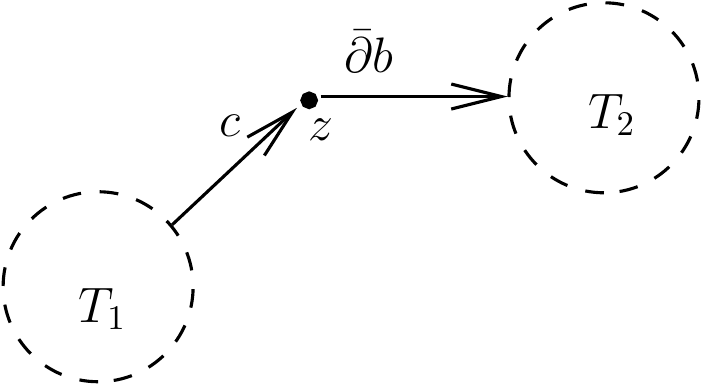}}} $$
Therefore, the classically vanishing expression 
$$\db\gamma+\frac{g}{2}[\bar{a},\gamma-\bar\gamma]-\frac{g}{2}[c,\db b]$$
vanishes under the correlator.

Likewise, we obtain
\begin{equation} \label{Laplace c correlator}
\lan \dd\db c(z)\;  \phi_1(x_1)\cdots \phi_n(x_n)
\ran = -\frac{g}{2} \lan \big(\db [a,c]+\dd[\bar{a},c]\big) \;  \phi_1(x_1)\cdots \phi_n(x_n)
\ran  
\end{equation}
and
\begin{equation}\label{Laplace b correlator}
\lan \dd\db b(z)\;  \phi_1(x_1)\cdots \phi_n(x_n)\ran = -\frac{g}{2} \lan \big( [a,\db b]+[\bar{a},\dd b]\big) \;  \phi_1(x_1)\cdots \phi_n(x_n)\ran  
\end{equation}
Ultimately, we see that all the expressions (\ref{eom in cx fields})  vanishing by classical equations of motion also vanish under the correlator. We further note that if a field $\Xi(z)$ vanishes under the correlator, then its product with any other field $\Phi(z)$ (the product 
is understood as renormalized in the sense of (\ref{normal ordering})) also vanishes under the correlator, since 
\begin{multline*} 
\lan \big(\Phi\Xi\big)(z)\;\cdots  
\ran = \lim_{z'\ra z} \lan \big(\Phi(z')\,\Xi(z)-[\Phi(z')\,\Xi(z)]_\mr{sing}\big)\;  \cdots
\ran \\
=  \lim_{z'\ra z}\Big( \underbrace{\lan \Phi(z')\,\Xi(z)\;\cdots \ran}_{\mr{I}} -   \underbrace{\lan [\Phi(z')\,\Xi(z)]_\mr{sing}\;\cdots \ran}_{\mr{II}} \Big)
=0 
\end{multline*}
Here $\cdots$ are test fields inserted away from $z$. Term $\mr{I}$ vanishes as a correlator of $\Xi(z)$ with insertions away from $z$ and $\mr{II}$ is, by definition of OPE, the singular part of $\mr{I}$ at $z'\ra z$ and thus also vanishes. The same argument applies if choose the opposite order of merging in $\Phi\Xi$, i.e., if we merge $\Xi$ onto $\Phi$. 

In summary, we have the following 
\begin{lemma}\label{lemma: eom under corr}
\begin{enumerate}[(a)]
\item 
Expressions (\ref{eom in cx fields}) viewed as composite fields are quantum equations of motion.
\item If $\Xi$ is a quantum equation of motion, then any derivative $\dd^p\db^q\Xi$ is also a quantum equation of motion.\footnote{
Here we understand that the order of merging for the derivative is inferred from the order of merging for $\Xi$ via (\ref{derivative of a renormalized product}).
}
\item If $\Xi$ is a quantum equation of motion and
and $\Phi$ is any composite field, then 
the renormalized products $\displaystyle{\til{\lim_{z'\ra z}}}\Phi(z')\Xi(z)$, $\displaystyle{\til{\lim_{z'\ra z}}}\Phi(z)\Xi(z')$ are also quantum equations of motion.
\end{enumerate}
\end{lemma}

Counterparts of the free theory OPE (\ref{O^2 db a OPE}) corresponding to (\ref{db gamma correlator}), (\ref{Laplace c correlator}), (\ref{Laplace b correlator}) are:
\begin{equation}
\begin{aligned}
\OO^{(2)}(u) \,\db \gamma(z) & \underset{g=0}{\sim} & 2\pi d^2u\; \big([\bar{a},\gamma-\bar{\gamma}]-[c,\db b]\big)(u)\;\delta(u-z)+\reg \\
\OO^{(2)}(u) \,\dd\db c(z) & \underset{g=0}{\sim} & 2\pi d^2u\; \big(\db[a,c]+\dd[\bar{a},c]\big)(u)\;\delta(u-z)+\reg \\
\OO^{(2)}(u) \,\dd\db b(z) & \underset{g=0}{\sim} & 2\pi d^2u\; \big([a,\db b]+[\bar{a},\dd b]\big)(u)\;\delta(u-z)+\reg
\end{aligned}
\end{equation}

Thus, for each  $\xi\in \{\db a, \dd \bar{a},\db\gamma,\dd\bar\gamma, \dd\db c,\dd\db b\}$ a derivative of a fundamental field vanishing by equations of motion in free theory, we have an OPE similar to (\ref{O^2 db a OPE}), of form
\begin{equation}\label{O^2 xi OPE}
 \OO^{(2)}(u)\, \xi(z)\underset{g=0}{\sim}4\pi d^2 u\, \delta(u-z)\, r_\xi(u) +\reg 
\end{equation}
with $r_\xi$ some composite field. 
Then the expression
\begin{equation}\label{Xi=xi+g r}
\Xi=\xi+g\, r_\xi
\end{equation}
vanishes under the correlator in the deformed theory. Thus, the deformation of equations of motion $\xi\ra \Xi$ from abelian to non-abelian theory is given by the coefficient $r_\xi$ of the contact term in the OPE of $\xi$ with the deforming $2$-observable $\OO^{(2)}$.

\begin{remark}
The OPE $T_0(z)\OO^{(2)}(u)$, see (\ref{T_0 O^2 ab OPE}), in fact contains a contact term:
\begin{equation}\label{T_0 O^2 with contact term}
T_0(z)\OO^{(2)}(u) \underset{g=0}{\sim} \frac{\OO^{(2)}(u)}{(z-u)^2}+\frac{\dd \OO^{(2)}(u)}{z-u}+4\pi\, d^2u\, \delta(z-u) \underbrace{\frac12 \lan \dd b,[a,c] \ran}_{T_1}(u)+\reg
\end{equation}
Observe that the composite field arising as the coefficient of the delta-function in the contact term is precisely $T_1$, the deformation of the stress-energy tensor induced by the non-abelian deformation of the theory, cf. (\ref{T explicit}), (\ref{Q,T,J deformation}). 
\end{remark}

\subsection{Quantum conservation laws: holomorphicity of $G$ and $T$. Quantum BRST operator}\label{ss: quantum conservation laws}
Using Lemma \ref{lemma: eom under corr}, we can prove the following quantum counterpart of the classical conservation laws (\ref{d J^tot = 0}), (\ref{G,T conservation}).
\begin{proposition}
We have 
\begin{gather}
\label{G conservation under corr} \lan \db G(z)\cdots  \ran = 0 , \quad \lan \dd \bar{G}(z)\cdots  \ran = 0,\\
\label{T conservation under corr}\lan \db T(z)\cdots  \ran =  0, \quad  \lan \dd \bar{T}(z)\cdots  \ran =  0, \\
\label{J conservation under corr}\lan d J^\tot(z)\cdots  \ran = 0
\end{gather}
with $\cdots$ any collection of test fields.
\end{proposition}
\begin{proof}
We start by considering $\db T$. Classically, we have
\begin{multline}\label{dbar T via eom}
\db T = \langle \dd (\underline{\db \gamma+\cdots}),a \rangle -
\langle\underline{\dd\bar\gamma+\cdots}, \dd\bar{a}\rangle + \langle \underline{\dd\db b+\cdots}, \dd c \rangle \\
+ \langle \dd b, \underline{\dd \db c+\cdots} \rangle  
+ \langle \dd\gamma, \underline{\db a+\cdots}\rangle + \langle \dd\bar\gamma, \underline{\dd \bar{a}+\cdots}\rangle
\end{multline}
Here the underlined terms are the expressions (\ref{eom in cx fields}) -- the classical equations of motion. For brevity, we write explicitly only the top derivative term in each equation, thus $\db a+\cdots$ stands for $\db a-\frac{g}{2}[a,\bar{a}]$ and similarly for other equations.\footnote{
Expansion (\ref{dbar T via eom}) arises when we write the stress-energy as $T=\lan \dd \gamma,a \ran+\lan \dd b , \dd c\ran+\frac{g}{2}\lan \dd b ,[a,c] \ran$. If instead we use the classically equivalent expression (\ref{T explicit}), we should add to (\ref{dbar T via eom}) the term $\db\langle \underline{\dd\bar\gamma+\cdots},a \rangle$.
}

In the quantum setting, we split $T$, placing fields $\gamma,\bar\gamma,b$ or derivatives at $z$ and fields $a,\bar{a},c$ or derivatives at a point $z'\ra z$. Then, taking the derivative, we will have the same splitting rule in (\ref{dbar T via eom}). Then, using Lemma \ref{lemma: merging A-fields onto a B-field}, we can equivalently re-assign fields $a,\bar{a},c$ or derivatives in underlined terms in (\ref{dbar T via eom}) to a point $z''$, so that we have $\displaystyle{\til{\lim_{z'\ra z}}}\db T^\mr{split}(z,z')= \displaystyle{\til{\lim_{z'\ra z}}}\,\displaystyle{\til{\lim_{z''\ra z}}} \db T^\mr{split}(z,z',z'') $. The latter expression vanishes under the correlator by Lemma \ref{lemma: eom under corr}.

One proves vanishing of $\db G$ and $d J^\tot$ under the correlator by the same reasoning. In particular, one has
$$ \db G =\langle \underline{\db a+\cdots},\dd b \rangle + \langle a, \underline{\dd\db b+\cdots} \rangle $$
and\footnote{
We remark that formula (\ref{d J^tot via eom}) has the structure $d J^\tot = \sum_i \pm\frac{\delta S}{\delta \phi_i}Q(\phi_i)$ where the sum runs over the species of fundamental fields $\phi_i\in \{a,\bar{a},\gamma,\bar\gamma,b,c\}$. Similarly, (\ref{dbar T via eom}) has the structure $\db T = \sum_i \pm\frac{\delta S}{\delta \phi_i}\dd\phi_i+\dd(\cdots)$.
}
\begin{multline}\label{d J^tot via eom}
dJ^\tot = 4d^2z \Big(
\langle \gamma+\bar\gamma, \underline{\dd \db c+\cdots} \rangle +
\langle \underline{\db\gamma+\cdots},\dd c +g[a,c] \rangle \\+ 
\langle \underline{\dd\bar\gamma+\cdots},\db c +g[\bar{a},c] \rangle
-
\frac{g}{2} \langle \underline{\dd \db b+\cdots}, [c,c] \rangle+
\frac{g}{2}\langle [c,\gamma-\bar\gamma] ,\underline{\db a-\dd \bar{a}+\cdots} \rangle
\Big)
\end{multline}
\end{proof}

\subsubsection{Quantum BRST operator}
We define the quantum BRST operator $Q_q$ acting on a composite field $\Phi(z)$ as
\begin{equation}\label{Q_q definition}
Q_q: \Phi(z) \mapsto \frac{1}{4\pi}\oint_{C_z\ni w} J^\tot(w) \Phi(z)
\end{equation}
-- this is understood as an equality under a correlator with test fields. Here $C_z$ is a simple closed contour going around $z$ in positive direction and  not enclosing any of the test fields. Note that the conservation law (\ref{J conservation under corr}) implies that the result is independent under deformations of the contour.

It turns out that quantum BRST operator essentially coincides with the classical BRST operator. More precisely, taking care of the order-of-merging issue for the composite fields, one has the following.
\begin{proposition} \label{Prop Q} Quantum BRST operator $Q_q$  satisfies the following properties.
\begin{enumerate}[(a)]
\item \label{Prop Q (a)} For $\Phi$ a fundamental field, quantum and classical BRST operators agree:
\begin{equation}\label{Q_q=Q}
Q_q \Phi(z)=Q\Phi(z)
\end{equation}
\item \label{Prop Q (b)} $Q_q$ commutes with derivatives:  
$$Q_q (\dd\Phi)=\dd (Q_q \Phi)\quad ,\quad Q_q (\db\Phi)=\db (Q_q \Phi)$$
for any composite field $\Phi$.
\item \label{Prop Q (c)} $Q_q$ acts as an  odd derivation on renormalized products:
\begin{equation}\label{Q_q derivation property}
Q_q \til{\lim_{z'\ra z}}\Phi_1(z')\Phi_2(z) = 
\til{\lim_{z'\ra z}} (Q_q \Phi_1)(z') \Phi_2(z) + \til{\lim_{z'\ra z}}(-1)^{|\Phi_1|} \Phi_1(z')(Q_q\Phi_2)(z)
\end{equation}
for any composite fields $\Phi_1,\Phi_2$.
\item \label{Prop Q (d)} For a general composite field $\Phi=(\phi_1\cdots\phi_n)_\mu$, with $\phi_i$ fundamental fields (or their derivatives) and $\mu$ the order-of-merging data (see Section \ref{sss: order of collapse ambiguity}), we have
\begin{equation}
Q_q \Phi = \sum_{i=1}^n \pm (\phi_1 \cdots (Q \phi_i) \cdots \phi_n )_\mu
\end{equation}
with $\pm=(-1)^{\sum_{j=1}^{i-1}|\phi_j|}$ the Koszul sign.
\item \label{Prop Q (e)} $Q_q$ squares to zero: 
\begin{equation}
Q_q^2 \Phi=0
\end{equation}
for any composite field $\Phi$.
\end{enumerate}
\end{proposition}
\begin{proof}
Property (\ref{Prop Q (b)}) is immediate from the definition of $Q_q$ (\ref{Q_q definition}), by applying a derivative to both sides.

Derivation property (\ref{Prop Q (c)}) is proven as follows.
We have
\begin{equation} \label{J derivation}
 \oint_{C_{12}\,\ni w} J^\tot(w)  \Phi_1(z') \Phi_2(z)  
=  \oint_{C_{1}} J^\tot(w)  \Phi_1(z') \Phi_2(z)+ 
\oint_{C_{2}} J^\tot(w)  \Phi_1(z') \Phi_2(z)  
\end{equation}
under a correlator with test fields away from $z,z'$. Here $C_{12}$ is a contour enclosing $z$ and $z'$, $C_1$ encloses only $z'$ and $C_2$ encloses only $z$. The equality corresponds to splitting an integral over $C_{12}$ into integrals over $C_1$ and $C_2$. Taking the limit $z'\ra z$ while subtracting singular terms as $z'\ra z$, yields the left and right sides of (\ref{Q_q derivation property}).\footnote{
More explicitly, let $\sum_{i} \sigma_{i}(z'-z)\til\Phi_{i}(z)$ be the singular part of the OPE $\Phi_1(z')\Phi_2(z)$ with $\sigma_{i}(z'-z)$ the basis singular coefficient functions (\ref{sigma_pqr}) and with $\til\Phi_{i}$ some composite fields (only finitely many of them nonzero).
Then the singular part of the l.h.s. of (\ref{J derivation}) is ${\sum_{i} \sigma_{i}(z'-z)Q_q\til\Phi_{i}(z)}$. Two integrals on the r.h.s. of (\ref{J derivation}) have singular parts $\sum_{i} \sigma_{i}(z'-z) \til\Phi^{(\alpha)}_{i}(z)$ with $\alpha=1,2$ and $\til\Phi^{(\alpha)}_{i}$ are some composite fields. Since l.h.s. and r.h.s. of (\ref{J derivation}) are equal under the correlator, the singular parts (and thus, coefficients of $\sigma_{i}$) must be equal: $Q_q\til\Phi_{i}=\til\Phi^{(1)}_{i}+\til\Phi^{(2)}_{i}$. In particular, expressions which we need to subtract from l.h.s. and r.h.s. of (\ref{J derivation}) to obtain l.h.s. and r.h.s. of (\ref{Q_q derivation property}) are the same, which proves (\ref{Q_q derivation property}).
}

Property (\ref{Prop Q (d)}) is an immediate consequence of properties (\ref{Prop Q (a)}),  (\ref{Prop Q (b)}),  (\ref{Prop Q (c)}). Property (\ref{Prop Q (e)}) follows from (\ref{Prop Q (d)}) by applying $Q_q$ twice to a composite field $(\phi_1\cdots \phi_n)_\mu$ and using that the classical BRST operator $Q$ squares to zero.

Lastly, consider property (\ref{Prop Q (a)}). 
For a fundamental field $\Phi$, we prove (\ref{Q_q=Q}) by a direct computation of the OPEs $J^\tot\Phi$. For instance, we compute
\begin{multline}\label{J^tot a OPE}
J^\tot(w)a(z)\sim \\
\sim  \underbrace{-(\dd c+g[a,c])}_{Qa} \frac{2i dw}{w-z}\;\; 
- \;\; 2i g [\db c+g[\bar{a},c],a] 
 d_w \Big((\bar{w}-\bar{z})\log|w-z| \Big)+\reg
\end{multline}
where fields on the r.h.s. are at $z$. The first term is a pole and gives a contribution $Qa$ to the contour integral (\ref{Q_q definition}); the second term is a milder (logarithmic) singularity and vanishes under the contour integral (as do regular terms). Likewise, we find
\begin{align}
J^\tot(w) c(z) &\sim \frac{g}{4}[c,c] \Big(\frac{-2idw}{w-z}+\frac{2id\bar{w}}{\bar{w}-\bar{z}}\Big) +
\cdots  \label{J^tot c OPE} \\ 
J^\tot(w) b(z) &\sim  \gamma\frac{-2i dw}{w-z}+\bar\gamma\frac{2id\bar{w}}{\bar{w}-\bar{z}}+\cdots  \label{J^tot b OPE} \\
\label{J^tot gamma OPE} J^\tot(w) \gamma(z) &\sim 
\frac{ig^2}{2} K( c(z),-)\, d_w\log^2|w-z|+\\
\nonumber & +
\frac{g}{2} [c,\gamma] \frac{2idw}{w-z}+\frac{g}{2} [c,\bar\gamma]\frac{2id\bar{w}}{\bar{w}-\bar{z}} +\cdots
\end{align}
where $\cdots$ stands for milder singular (e.g. logarithmic and $O(\frac{\bar{w}-\bar{z}}{w-z})$) terms,\footnote{Individually, a logarithmic or an $O(\frac{\bar{w}-\bar{z}}{w-z})$ term could contribute to an integral over a finite contour, but, due to (\ref{J conservation under corr}), such terms  always combine into $d_w$-closed expressions (cf. the second term in (\ref{J^tot a OPE})), thus one can take the contour to be very small -- in this limit, singular terms milder than a first order pole clearly vanish when integrated.} not contributing to the contour integral. In fact, (\ref{J^tot a OPE}) and (\ref{J^tot c OPE}) are computed easily from (\ref{OPE = tree OPE}), using the results of Section \ref{sss: OPE of fundamental fields}. For (\ref{J^tot b OPE}) one could have $1$-loop diagrams, but they vanish/are non-singular. OPE (\ref{J^tot gamma OPE}) is more complicated (see Remark \ref{rem: Q_q phi ansatz} below for a shortcut to computing $Q_q\gamma$); the first term on the l.h.s. comes from $1$-loop diagrams which contain potentially dangerous terms 
proportional to $g^2\frac{\log|w-z|}{w-z} K(c(z),-)$ (or the conjugate), with $K$ the Killing form on $\g$; these terms add up to a $d_w$-exact term when summed in (\ref{J^tot gamma OPE}). Ultimately, the first term on the l.h.s. of (\ref{J^tot gamma OPE}) vanishes under the contour integral over $w$ and does not contribute to $Q_q\gamma$.

This finishes the proof of (\ref{Q_q=Q}).
\end{proof}

\begin{remark}\label{rem: Q_q phi ansatz}
The following trick allows one to simplify the computation of $Q_q\Phi$ for fundamental fields $\Phi$, and in particular provides an alternative way to calculate $Q_q\gamma$,  avoiding the direct computation of the OPE (\ref{J^tot gamma OPE}). 
 Writing 
\begin{equation}\label{Q_q phi = sum_k g^k Phi}
Q_q\Phi=\sum_{k\geq 0}g^k \Phi_k
\end{equation} 
one can restrict the form of possible composite fields $\Phi_k$ appearing on the right by analyzing various degrees on the left and right side of (\ref{Q_q phi = sum_k g^k Phi}) -- the ghost number, $\mc{AB}$-charge, weight $(h,\bar{h})$ and the number of constituent fundamental fields of the composite field $\Phi_k$. We see that $\Phi_k$ must have the following properties

\begin{tabu}{c|c|c|c}
ghost degree & $\mc{AB}$-charge & weight $(h,\bar{h})$ & $\#$ fund. fields \\ \hline
$|\Phi|+1$ & $\mc{AB}(\Phi)-k$ & $(h_\Phi,\bar{h}_\Phi)$ & $k+1-2\#\mbox{loops}$
\end{tabu}

Here $\#\mbox{loops}$ is the number of loops in the Feynman diagram giving the contribution to OPE. These properties immediately imply that, writing the r.h.s. of $Q_q\phi$ schematically, up to numeric factors and indices, we have 
$$Q_q a \approx \dd c +g ac, \;\; Q_q c \approx g c c,\;\;  Q_q b \approx\gamma+\bar\gamma+g bc+g\kappa,\;\; Q_q \gamma \approx g\gamma c+g^2 c $$
\end{remark}
Here $\kappa$ in $Q_q b$ is a component of constant vector in $\g^*$; it must be zero due to global $\g$-invariance (on the level of Feynman diagrams, it vanishes due to unimodularity of $\g$).    Furthermore, one can exclude the $bc$ structure from $Q_q b$, since bare $b$ ghost cannot appear on a leaf of a Feynman diagram for $J^\tot b$ OPE. Finally, once $Q_q b$ is known, one can prove that $Q_q\gamma$ does not contain the $1$-loop correction term $g^2 c$, by probing it with the correlator with a test field $\dd b(x)$:
\begin{equation}\label{Q_q gamma testing for c term}
\lan  (Q_q \gamma)(z)\; \dd b(x)\ran = 
\langle \frac{1}{4\pi}\oint_{C_{zx}\ni w} J^\tot(w)  \gamma(z) \dd b(x)\rangle+
\underbrace{\langle  \gamma(z)\; Q_q(\dd b(x)) \rangle}_{0}
\end{equation}
Here on the left $Q_q$ acts by a contour integral of $J^\tot$ around $z$ -- we present it on the right as an integral over a large contour $C_{zx}$ encircling both $z$ and $x$, minus a term with $J^\tot$ encircling only $x$. The second term is the correlator of already known $Q_q(\dd b)=\dd(\gamma+\bar\gamma)$ with $\gamma(z)$ and vanishes trivially, since $2$-point functions $\lan \gamma(x) \gamma(z) \ran$, $\lan \bar\gamma(x)\gamma(z) \ran$ are zero. If $Q_q \gamma$ would contain a $g^2 c$ term, the correlator (\ref{Q_q gamma testing for c term}) would behave as $O\big(\frac{1}{z-x}\big)$. Considering the asymptotics $z\ra x$, we see that the r.h.s. does not behave this way, since the OPE $\gamma \dd b$ (\ref{OPE db gamma}) does not -- it behaves as $O\big(\frac{\bar{z}-\bar{x}}{z-x}\big)$. Hence, the coefficient of $g^2 c$ in $Q_q \gamma$ must be zero.

\subsubsection{OPEs of $G$ with fundamental fields}
Since field $G$ is holomorphic under the correlator, its OPE with any composite field $\Phi$ must have the form
\begin{equation}\label{G Phi OPE}
G(w) \Phi(z) \sim \sum_{k=1}^p (w-z)^{-k} \Phi_k(z) +\mr{reg}^{(\infty)}
\end{equation}
with $\Phi_k$ some composite fields and some $p\geq 0$. For instance, such an OPE cannot contain terms like $\log|w-z|$ or $\frac{\bar{w}-\bar{z}}{w-z}$ which we have seen in other OPEs. The remainder in (\ref{G Phi OPE}) is holomorphic at $w\ra z$; in particular this OPE can be differentiated arbitrarily many times. Similarly, one has that the singular part $\bar{G}(w)\Phi(z)$ is a Laurent polynomial in $\bar{w}-\bar{z}$. Since the stress-energy tensor $T$ is also holomorphic, same observation applies to $T(w) \Phi(z)$: one has
\begin{equation}
T(w) \Phi(z) \sim \sum_{k=1}^q (w-z)^{-k} \til\Phi_k(z) +\mr{reg}^{(\infty)}
\end{equation}
and similarly for $\bar{T}(w)\Phi(z)$.

Another observation is that for $\Phi$ a fundamental field, the OPE  $G(w)\Phi(z)$ does not have admissible decorations for $1$-loop diagrams, and hence this OPE satisfies (\ref{OPE = tree OPE}). Explicitly, we obtain:
\begin{equation}\label{G (fund field) OPEs}
\begin{aligned}
G(w) a(z) &\sim  \mr{reg}^{(\infty)}, &
G(w) \bar{a}(z) &\sim   \mr{reg}^{(\infty)},\\
G(w) \gamma(z) &\sim  \frac{\dd b(z)}{w-z}+\mr{reg}^{(\infty)}, &
G(w) \bar\gamma(z) &\sim  \mr{reg}^{(\infty)}, \\
G(w) c(z) &\sim -\frac{a(z)}{w-z} + \mr{reg}^{(\infty)}, &
G(w) b(z) &\sim \mr{reg}^{(\infty)}
\end{aligned}
\end{equation}
By complex conjugation, one obtains OPEs of $\bar{G}$ with fundamental fields. 
Also, the OPE of $G$ with itself is trivial:
\begin{equation}
G(w)G(z) \sim \mr{reg}^{(\infty)}\qquad , \qquad G(w)\bar{G}(z) \sim \mr{reg}^{(\infty)}
\end{equation}

Computing the OPE between $G$ and 
the BRST current, one gets
\begin{equation}
\begin{aligned}
G(w)J(z) &\sim \frac{-\dim\g}{(w-z)^3}+\frac{-\lan \gamma,a \ran}{(w-z)^2}+\frac{T-\dd \lan \gamma,a \ran
}{w-z}+\mr{reg}^{(\infty)},\\
G(w) \bar{J}(z) & \sim \frac{\db\lan \gamma,a \ran}{w-z}+\pi\delta(w-z)\lan \bar\gamma, a \ran+\mr{reg}^{(\infty)}
\end{aligned}
\end{equation}
All the fields on the right are at $z$. Here the cubic pole comes from a 1-loop diagram. Thus, for $G J^\tot$, one has
\begin{equation}\label{G J^tot OPE}
G(w) J^\tot(z)\sim -\frac{2i Tdz}{w-z}+2i d_z\left(\frac{\frac12 \dim\g}{(w-z)^2}+\frac{\lan \gamma,a \ran}{w-z}\right)+\mbox{(contact term)} +\mr{reg}^{(\infty)}
\end{equation}

Integrating $J^\tot$ around $G$ in (\ref{G J^tot OPE}), we find
\begin{equation}\label{Q_q G = T}
Q_qG=T
\end{equation}
-- the quantum counterpart of (\ref{T=Q(G)}).

\subsubsection{Quantum stress-energy tensor. Examples of primary fields. $TT$ OPE}\label{sss: T quantum, primary fields}
Recall that a field $\Phi$ is called \textit{primary}, of conformal dimension $(\Delta,\bar{\Delta})$, if its OPEs with $T,\bar{T}$ are of the form 
\begin{equation}\label{T Phi primary OPE}
 T(w) \Phi(z)\sim \frac{\Delta\Phi(z)}{(w-z)^2} + \frac{\dd \Phi(z)}{w-z}+\rg \quad , \quad 
\bar{T}(w) \Phi(z)\sim \frac{\bar{\Delta}\Phi(z)}{(\bar{w}-\bar{z})^2} + \frac{\dd \Phi(z)}{\bar{w}-\bar{z}}+\rg 
\end{equation}

\begin{proposition} \label{Prop T}\leavevmode
\begin{enumerate}[(a)]
\item \label{Prop T (a)} Fundamental fields $a,\bar{a},\gamma,\bar\gamma,b,c$ are all primary (each component of these fields), 
 of conformal dimension $(1,0)$ for $a$, $(0,1)$ for $\bar{a}$ and $(0,0)$ for the rest. 
\item \label{Prop T (b)}Fields $G,\bar{G}$ are primary, of conformal dimension $(2,0)$ and $(0,2)$, respectively.
\item \label{Prop T (c)}Stress-energy tensor satisfies the OPE
\begin{equation}\label{TT OPE}
T(w) T(z)\sim \frac{2T(z)}{(w-z)^2}+ \frac{\dd T(z)}{w-z}+\rg\quad, \quad T(w)\bar{T}(z)\sim 
\rg
\end{equation}
Thus, $T$ has the standard OPE of a conformal field theory with central charge $\mathsf{c}=\bar{\mathsf{c}}=0$ (since we do not have a $4$-th order pole in $TT$ and $\bar{T}\bar{T}$ OPEs).\footnote{Recall that in a conformal theory with (holomorphic) 
central charge $\mathsf{c}$, the stress-energy tensor satisfies the OPE $T(w)T(z)\sim\frac{\mathsf{c}/2}{(w-z)^4}+\frac{2T(z)}{(w-z)^2}+\frac{\dd T(z)}{w-z}+\rg$ and similarly 
for $\bar{T}\bar{T}$ and anti-holomorphic central charge $\bar{\mathsf{c}}$.
} 
Put another way, $T$ and $\bar{T}$ themselves are primary fields of dimensions $(2,0)$ and $(0,2)$, respectively.
\end{enumerate}
\end{proposition}

\begin{proof}
First note that for any field $\Phi$ we have
\begin{equation}\label{T Phi via G Phi}
T(w) \Phi(z) = Q_qG(w) \Phi(z) = Q_q\big(G(w) \Phi(z)\big) + G(w)\, Q_q\Phi (z)
\end{equation}
-- by combining the BRST-exactness of the stress-energy tensor (\ref{Q_q G = T}) with the contour-switching argument (\ref{J derivation}); $Q_q$ in the first term on the r.h.s. means ``integrate $J^\tot$ over a contour enclosing both $w$ and $z$.'' Thus, computing the OPE $T\Phi$ reduces to computing OPEs of $G$ with $\Phi$ or $Q_q\Phi$; computing of $Q_q$ on any field is straightforward (reduces to computing the classical BRST operator) by Proposition \ref{Prop Q}. 

Next, we make the following remark: if $\Phi$ is at most linear in fundamental fields of $\mc{AB}$-charge $+1$, then
\begin{equation}\label{G Phi 1-loop vanishing}
G(w) \Phi(z) = [G(w)\Phi(z)]_\mr{tree}\qquad \mbox{unless $\Phi$ contains $c$ and $\gamma$ or $\bar\gamma$}
\end{equation}
This is a special case of the remark of Section \ref{sss: a remark on OPEs of composite fields}: $1$-loop graphs (\ref{1-loop graphs for Phi_1 Phi_2 OPE linear in B}) involving $G$ and $\Phi$ have no admissible decorations unless $\Phi$ contains $c$ and $\gamma$ or $\bar\gamma$ (possibly with derivatives).

For $\Phi$ a fundamental field, we calculate $G(w)\, Q\Phi(z)$ (recall that $Q_q=Q$ on fundamental fields) using (\ref{G Phi 1-loop vanishing}):
\begin{align*}
G(w)\, Qa(z) &\sim \frac{a}{(w-z)^2}+\frac{\dd a}{w-z} +\rg, & 
G(w)\, Q\bar{a}(z) &\sim \frac{\dd \bar{a}}{w-z}+\rg \\
G(w) \, Qc(z)&\sim -g\frac{[a,c]}{w-z}+\rg, & 
G(w) \, Qb(z) &\sim \frac{\dd b}{w-z}+\rg \\
G(w)\, Q\gamma(z) &\sim 
-\frac{\dd\bar\gamma}{w-z}+\rg
,&
G(w)\, Q\bar\gamma(z) &\sim \frac{\dd\bar\gamma}{w-z}+\rg
\end{align*}
where all fields in the r.h.s. are at $z$.\footnote{
Note that the field $Q\gamma = \frac{g}{2}[c,\gamma-\bar\gamma]$ does contain both $c$ and $\gamma$ or $\bar\gamma$, so there is a possibility of a $1$-loop correction to the OPE $G(w) Q\gamma(z)$. However, this $1$-loop diagram vanishes by unimodularity. The same applies to $G(w) Q\bar\gamma(z)$. Another remark is that these OPEs are written modulo equations of motion and modulo contact terms. -- In fact, there is a contact term arising in $G(w) Q\bar{a}(z)$. It corresponds to a contraction between $\dd\bar\gamma(w)$ and $\bar{a}(z)$ in $T\bar{a}$ OPE when the stress-energy tensor is written in the form (\ref{T explicit}). Similarly, $G\bar{T}$ OPE (\ref{G Tbar}) below would contain a contact term if $\bar{T}$ were written as complex conjugate of (\ref{T explicit}) instead of (\ref{T equivalent}). 
} Combining these OPEs with $Q_q$ applied to OPEs (\ref{G (fund field) OPEs}), as in (\ref{T Phi via G Phi}), we obtain the OPEs of the standard primary form (\ref{T Phi primary OPE}) between $T$ and any fundamental field; OPEs between $\bar{T}$ and fundamental fields are complex conjugates of the ones we already found. This proves item (\ref{Prop T (a)}).

Next, we calculate $G(w)T(z)$, which is straightforward using (\ref{G Phi 1-loop vanishing}) and (\ref{G (fund field) OPEs}):
\begin{multline*}
G(w) \underbrace{\big(\lan\dd\gamma, a \ran+\lan \dd b, \dd c \ran+\frac{g}{2}\lan \dd b,[a,c] \ran \big)(z)}_{T(z)} \sim \\
\sim
\frac{\lan \dd b,a \ran(z)}{(w-z)^2} +\frac{\lan \dd^2 b, a\ran(z)}{w-z}+ \frac{\lan\dd b,a\ran(z)}{(w-z)^2}+ \frac{\lan \dd b, \dd a\ran(z)}{w-z}+\rg \sim \\ \sim
\frac{2\lan \dd b,a\ran(z)}{(w-z)^2}+\frac{\dd \lan \dd b,a\ran(z)}{w-z}+\rg
\sim \frac{2G(w)}{(z-w)^2}+\frac{\dd G(w)}{z-w}+\rg
\end{multline*}
where in the last step we re-expanded the fields in the r.h.s. at $w$ instead of $z$. Similarly, one finds
\begin{multline}  \label{G Tbar}
G(w) \underbrace{ \big(\lan \db\bar\gamma,\bar{a} \ran+ \lan \db b, \db c \ran+\frac{g}{2}\lan\db b,[\bar{a},c]\ran\big)(z)}_{\bar{T}(z)} \sim
\\ \sim
\frac{\lan \db b, \db a \ran(z)}{w-z}+\frac{g}{2}\frac{\lan\db b,[\bar{a},a]\ran(z)}{w-z}+\rg \quad \sim\quad \rg
\end{multline}
Thus, $G$ is indeed a primary field of dimension $(2,0)$ (note that we do not see the pole $\frac{1}{\bar{z}-\bar{w}}$ in $\bar{T}G$ OPE, since its coefficient $\db G$ vanishes under the correlator). By complex conjugation, we get that $\bar{G}$ is $(0,2)$-primary. This proves item (\ref{Prop T (b)}).

Finally, item (\ref{Prop T (c)}) follows immediately from (\ref{Prop T (b)}) by applying $Q_q$ to the $GT$, $G\bar{T}$ OPEs and using the fact that $T,\bar{T}$ are $Q_q$-closed:
\begin{multline*} 
T(w) T(z)=Q_q(G(w) T(z))+G(w) \underbrace{Q_qT(z)}_0\sim \\
\sim Q_q\Big(\frac{2G(z)}{(w-z)^2}+\frac{\dd G(z)}{w-z}+\rg\Big) \sim \frac{2T(z)}{(w-z)^2}+\frac{\dd T(z)}{w-z}+\rg 
\end{multline*}
and
$$ T(w) \bar{T}(z) = Q_q(\underbrace{G(w) \bar{T}(z)}_{\rg})+G(w)\underbrace{Q_q\bar{T}(z)}_0 \sim \rg$$
\end{proof}

\begin{example}
If a field $\Phi$ is primary of conformal dimension $(0,\bar{\Delta})$, then $\dd \Phi$ is also primary, of dimension $(1,\bar{\Delta})$. This follows from (\ref{T Phi primary OPE}) with $\Delta=0$, by applying $\dd_z$. Similarly, for $\Phi$ primary of dimension $(\Delta,0)$, $\db\Phi$ is primary of dimension $(\Delta,1)$.  This implies in particular that derivatives of fundamental fields 
$$\db a, \dd \gamma, \db \gamma, \dd\db\gamma, \dd b, \dd\db b, \dd c, \dd\db c$$ 
and complex conjugates are primary (but higher derivatives are non-primary).
\end{example}


\section{Examples of fields with a quantum correction to dimension (``vertex operators'')}\label{ss: vertex operators}
In this section we will present vertex operators with anomalous dimensions, i.e., with the actual conformal dimension different from the naive one, defined in Section \ref{sss: conf dimension}. We obtain these dimensions in two ways: first, by considering OPEs with $T$ in Subsection \ref{sss: V and W}. The second way is  due to singular  subtractions in renormalized products in Subsection \ref{sss: V and W corr}.

These two ways in the standard case of the free scalar field are:
\begin{enumerate}
\item OPE of the vertex operator $V_\alpha=\nl e^{i\alpha\phi}\nr$ 
with the energy-momentum tensor $T=-\frac12\nl\dd \phi\dd\phi\nr$ is
$$ T(w) V_\alpha(z)\sim  \frac{\frac{\alpha^2}{2}V_\alpha(z)}{(w-z)^2}+ \frac{\dd V_\alpha(z)}{w-z}+\reg
$$
Together with the similar OPE $\bar{T}(w)V_\alpha(z)$, this implies that $V_\alpha$ is primary, of conformal dimension $(\Delta=\frac{\alpha^2}{2},\bar{\Delta}=\frac{\alpha^2}{2})$.
\item 
Recall that the renormalized (normally ordered) field depends on the choice of local coordinate. In particular,
under the infinitesimal change of local coordinate 
$z\ra z'=(1+\epsilon)z$, the renormalized field $\nl\phi^k(0)\nr$ transforms as
\begin{equation}\label{phi^k transformation}
\nl\phi^k\nr_{z} \ra \quad \nl\phi^k\nr_{z'}=\quad\nl\phi^k\nr_{z}+ \epsilon\, k(k-1) \nl\phi^{k-2}\nr_{z}  
\end{equation}
up to $O(\epsilon^2)$ terms, as proven by induction in $k$ using 
$$\nl\phi^{k+1}(0)\nr_z=\lim_{p\ra 0}\Big(\phi(p)\nl\phi^{k}(0)\nr_z +2 k\nl\phi^{k-1}(0)\nr_z\log|z(p)|\Big)$$
in local coordinate $z$. Here $p$ is a point of insertion of an observable and here we take care to distinguish between a point $p$ and its coordinate $z(p)$.
Summing (\ref{phi^k transformation}) over $k$ with coefficients $\frac{(i\alpha)^k}{k!}$, we obtain the transformation law for the vertex operator:
\begin{equation}\label{Free scalar V scaling}
(V_\alpha)_z \ra (V_\alpha)_{z'}=\left(1-\epsilon \alpha^2\right)(V_\alpha)_{z} 
\end{equation}
This is consistent with scaling dimension $\Delta+\bar{\Delta}=\alpha^2$.
\end{enumerate}

\subsection{New vertex operators $V$ and $W$ and their conformal dimensions}
\label{sss: V and W}

Let us fix $X\in\g$ a Lie algebra element,  fix $Y\in\g$ an eigenvector of $\mr{ad}_X$ with eigenvalue $\alpha$ and fix $\rho\in\g^*$ an eigenvector of the coadjoint action $\mr{ad}^*_X$ with eigenvalue $-\alpha$, i.e.:
$$ \ad_X Y=\alpha Y,\qquad \ad^*_X \rho = -\alpha \rho$$

Consider the following composite fields (``vertex operators''):\footnote{
We understand $V_{X,\rho}$ as $\sum_{n\geq 0} \frac{1}{n!}\lan \rho,a \ran\, \lan X,\bar\gamma \ran^n$. Each term in the sum is a composite field understood as the renormalized product $\displaystyle{\til{\lim_{z'\ra z}}}\,\frac{1}{n!}\lan \rho,a(z') \ran\, \lan X,\bar\gamma(z) \ran^n$. Here we can safely put all $\bar\gamma$'s into the same point as they are regular with each other. In fact, since $a(z') \bar\gamma^n(z)$ OPE contains only powers of logs (cf. (\ref{3-point a (bargamma ... bargamma) gamma})), there is no order-of-merging ambiguity in the renormalized product above. Similar remarks apply to $W_{X,Y}$.
}
\begin{equation}\label{vertex operators}
V_{X,\rho}=\lan \rho,a \ran e^{\lan X,\bar\gamma \ran} , \qquad W_{X,Y}= \lan \gamma-\bar\gamma,Y \ran e^{\lan X,\bar\gamma \ran}
\end{equation}
Note that $V$ depends only on the vectors $X,\rho$ and $W$ depends only on $X,Y$.

\begin{proposition} \label{Prop: vertex operators}
Fields $V_{X,\rho}$, $W_{X,Y}$ are  primary, of conformal dimensions 
$$(\Delta=1-\frac{\alpha g}{2},\bar{\Delta}=-\frac{\alpha g}{2})\; \mbox{ for }V,\qquad
(\Delta=\frac{\alpha g}{2},\bar{\Delta}=\frac{\alpha g}{2})\; \mbox{ for }W$$
\end{proposition}

\begin{proof}
First, note that for any composite field $\Phi$, the residue of the first order pole in the OPE $T(w)\Phi(z)$ (commonly denoted $L_{-1}\Phi=\frac{1}{2\pi i}\oint T(w) \Phi(z)$) is $\dd\Phi$. This follows from the fact that $L_{-1}=\dd$ on fundamental fields (from Proposition \ref{Prop T}), and hence for derivatives of fundamental fields (as $L_{-1}$ commutes with derivatives, by the same logic as (\ref{Prop Q (b)}) of Proposition \ref{Prop Q}). Finally, $L_{-1}$ is a derivation of the renormalized product, by the same logic as (\ref{Prop Q (c)}) of Proposition \ref{Prop Q} (contour switching argument).

Consider the field 
\begin{equation}\label{V^n}
V^{(n)}=\frac{1}{n!}\lan \rho,a \ran \lan \bar\gamma,X\ran^n
\end{equation} 
Its OPE with $T$ must be of the form
\begin{equation}\label{TV^n}
T(w) V^{(n)}(z)  \sim \frac{\cdots}{(w-z)^3}+\frac{\cdots}{(w-z)^2}+\frac{\dd V^{(n)}(z)}{w-z}+\rg
\end{equation}
Here we cannot get a pole higher than third order, because l.h.s. has weight $(3,0)$ (and we don't have fields of negative weight to accompany a pole of order $>3$). The coefficient of the third order pole must be of weight $(0,0)$ and in fact there are no such contributing diagrams.\footnote{
Indeed, in such a diagram external half-edges would need to be decorated by either $\gamma$ or $\bar\gamma$ (since $a,\bar{a},\dd b, \db b$ have nonzero weight and presence of $c$ on an external half-edge would require, by conservation of ghost number, another external half-edge decorated by $\dd b$ or $\db b$). Thus, the diagram must consist of $\geq 2$ trees rooted at $a$ from $V^{(n)}$, at $a$ or $c$ from $T$ and at external half-edges. Leaves of the trees are decorated jointly by $n$ fields $\bar\gamma$ and one $\dd\gamma$ or $\dd b$ (from $T$). Therefore, there must be a tree whose leaves are decorated only by $\bar\gamma$'s. Such a tree vanishes.
} 
Looking for the second order pole, we look for diagrams producing a field of weight $(1,0)$. There are three families of such diagrams:
$$ \vcenter{\hbox{\includegraphics[scale=0.7]{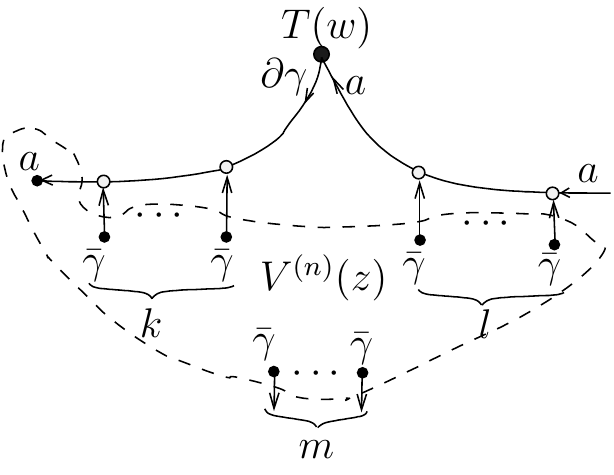}}}\quad, \quad
\underbrace{\vcenter{\hbox{\includegraphics[scale=0.7]{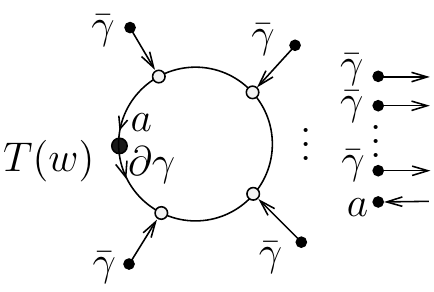}}}\quad, \quad
\vcenter{\hbox{\includegraphics[scale=0.7]{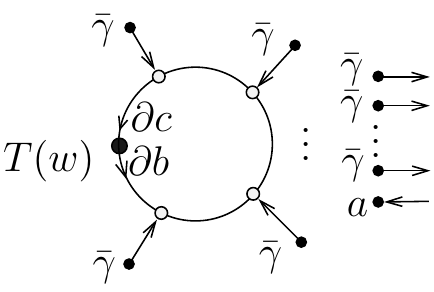}}}}_{\mbox{cancel out}}
 $$
Diagrams of second and third type cancel each other by the mechanism of Lemma  \ref{lemma: boson-fermion cancellation}.
Diagrams of the first type, evaluated using the computations of Section \ref{sss: correlators with bargamma...bargamma},  jointly give the following contribution to the OPE (\ref{TV^n}):
\begin{multline*}
\sum_{k,l,m\geq 0, k+l+m=n} \frac{1}{n!}\left(
\begin{array}{c}
n \\ k,  l,  m
\end{array}
\right)
g^k(-g)^l  \lan \rho, \ad_X^{k+l} a \ran \dd_w \frac{\log^k|z-w|}{z-w} \cdot \\ \cdot  \log^l|z-w|\cdot \lan \bar\gamma,X\ran^m \\
=\frac{1}{(z-w)^2}\left(\frac{1}{n!}\lan\rho, a \ran \lan\bar\gamma,X \ran^n-\frac{g\alpha}{2(n-1)!}
\lan\rho, a \ran \lan\bar\gamma,X \ran^{n-1}\right)
\end{multline*}
Here all fields are at $z$; $\left(
\begin{array}{c}
n \\ k,  l,  m
\end{array}
\right)$ is the multinomial coefficient.
Note that all the terms involving positive powers of $\log$ have cancelled out (as expected) -- in fact, all diagrams except ones with $k+l\leq 1$ cancel out when summed with $k+l$ fixed. Thus,  we obtained the explicit form of the OPE (\ref{TV^n}):\footnote{
This OPE implies that fields $\left(-\frac{g\alpha}{2}\right)^{-n} V^{(n)}$ comprise a Jordan cell of infinite rank of the Virasoro operator $L_0$, see Section \ref{sss: log phenomena}.
}
\begin{equation}\label{TV explicit}
T(w) V^{(n)}(z) \sim \frac{\big(V^{(n)} -\frac{g \alpha}{2} V^{(n-1)}\big)(z)}{(w-z)^2}+\frac{\dd V^{(n)}(z)}{w-z}+\rg
\end{equation}
Here by convention $V^{(-1)}=0$. Summing over $n\geq 0$, we find that our field $V_{X,\rho}=\sum_{n\geq 0}V^{(n)}$ satisfies the standard primary OPE with $T$, with holomorphic conformal dimension $\Delta=1-\frac{\alpha g}{2}$. A similar computation yields the OPE
$$ \bar{T}(w) V^{(n)}(z)\sim \frac{-\frac{g \alpha}{2} V^{(n-1)}(z)}{(\bar{w}-\bar{z})^2}+\frac{\db V^{(n)}(z)}{\bar{w}-\bar{z}}+\rg$$
where the relevant diagrams are
\begin{equation*}  \label{barT V diagrams}
\vcenter{\hbox{\includegraphics[scale=0.7]{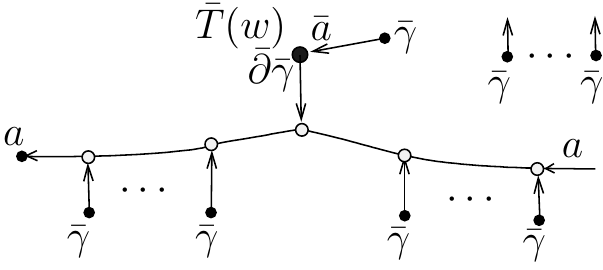}}} ,\quad,\quad
\underbrace{\vcenter{\hbox{\includegraphics[scale=0.7]{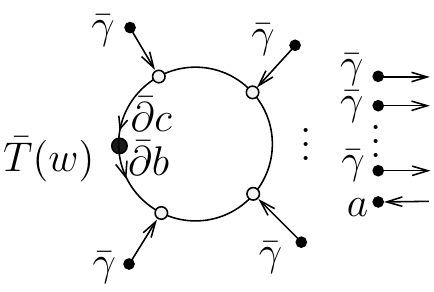}}} ,\quad,\quad
\vcenter{\hbox{\includegraphics[scale=0.7]{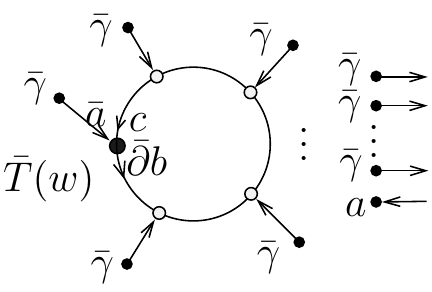}}}}_{\mbox{cancel out}}
\end{equation*}
Thus, the anti-holomorphic dimension of the field $V_{X,\rho}$ is $\bar\Delta=-\frac{\alpha g}{2}$.

Computation of the OPEs $TW$, $\bar{T}W$ is similar, with the following relevant diagrams (we omit the families cancelling by boson-fermion cancellation in the loop):
$$ \vcenter{\hbox{\includegraphics[scale=0.7]{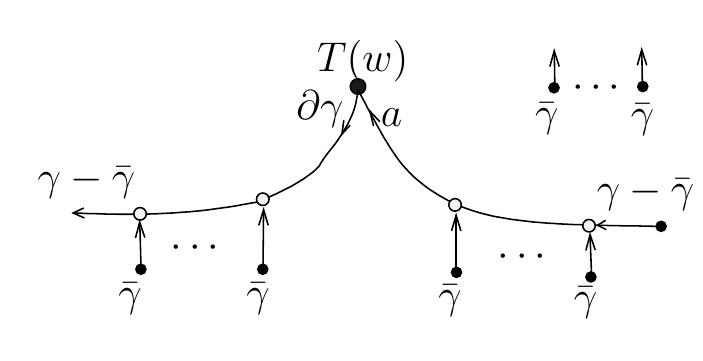}}}\qquad, \qquad  \vcenter{\hbox{\includegraphics[scale=0.7]{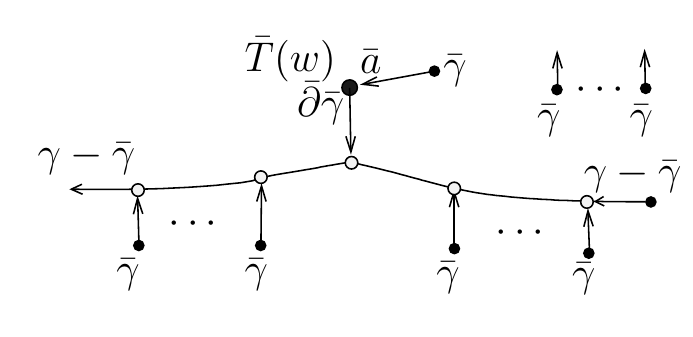}}} $$
\end{proof}

\subsection{Another view on conformal dimensions and examples of correlators}
\label{sss: V and W corr}
Ultimately, the source of the shift of the 
conformal dimension 
is in singular subtractions -- powers of logs -- needed in the renormalized products when we build the vertex operators from fundamental fields. These subtractions depend on the local coordinate and ultimately lead to the anomalous scaling behavior.

Explicitly: consider the field (\ref{V^n}) viewed as renormalized product
\begin{equation}\label{V^n renormalized}
V^{(n)}(0)=\lim_{p\ra 0}\left(V^{(n)}_{\mr{split}}(p,0)-\Big[V^{(n)}_{\mr{split}}(p,0)\Big]_\mr{sing}\right) 
\end{equation}
where
$$V^{(n)}_{\mr{split}}(p,0)=\frac{1}{n!}\lan \rho,a(p) \ran \lan\bar\gamma(0),X \ran^n$$
If we make an infinitesimal change of local coordinate 
$z\ra z'=(1+\epsilon) z$, 
we have the following:
\begin{enumerate}[(a)]
\item \label{vertex conf dim (0)} The split field transforms as
$$ V^{(n)}_{\mr{split}}(p,0)_z\ra V^{(n)}_{\mr{split}}(p,0)_{z'}=(1-\epsilon) V^{(n)}_{\mr{split}}(p,0)_z $$
up to $O(\epsilon^2)$ terms; the subscript $z$, $z'$ refers to the coordinate system.
Here the $\epsilon$-correction comes from the transformation of $a$.
\item \label{vertex conf dim (a)}
The singular subtraction in  (\ref{V^n renormalized})
is transformed  as 
$$\Big[V^{(n)}_{\mr{split}}(p,0)\Big]^\mr{sing}_z \ra
\Big[V^{(n)}_{\mr{split}}(p,0)\Big]^\mr{sing}_{z'} =
(1-\epsilon)\Big[V^{(n)}_{\mr{split}}(p,0)\Big]^\mr{sing}_{z}-\epsilon\, \alpha g V^{(n-1)}_z(0) $$
\item \label{vertex conf dim (b)} The field $V^{(n)}$ is transformed as
\begin{equation} \label{V^n scaling transf}
V^{(n)}_z \ra V^{(n)}_{z'}= (1-\epsilon) V^{(n)}_z+ \epsilon\,  \alpha g V^{(n-1)}_z 
\end{equation}
where all the fields are at the origin. 
\end{enumerate}
One proves this by induction in $n$: (\ref{vertex conf dim (0)}) is staightforward,\footnote{\label{footnote: bargamma^n no anomalous dim}
Note that in the composite field $\lan\bar\gamma,X \ran^n$  there are no singular subtractions (logarithmic or otherwise),
thus there is no anomalous dimension.}
 (\ref{vertex conf dim (b)}) follows from (\ref{vertex conf dim (0)}) and (\ref{vertex conf dim (a)}) immediately. In turn, (\ref{vertex conf dim (a)}) follows from the (\ref{vertex conf dim (b)}) for smaller  $n$ and from the OPE 
\begin{equation}\label{V^n split OPE}
V^{(n)}_{\mr{split}}(p,0)
\sim  \sum_{k=0}^n \frac{(-\alpha g)^k}{k!}{\log^k|z(p)|}\,  
V^{(n-k)}(0)+o(1)_{p\ra 0}
\end{equation}
obtained similarly to (\ref{3-point a (bargamma ... bargamma) gamma}). 
Here terms $1\leq k\leq n$ give the singular part of the OPE and $k=0$ is the regular part (modulo terms which are continuous and vanishing at $p\ra 0$). 
To see (\ref{vertex conf dim (a)}) explicitly, we compute from (\ref{V^n split OPE}):
\begin{multline*}
\Big[V^{(n)}_{\mr{split}}(p,0)\Big]^\mr{sing}_{z'}=\sum_{k=1}^n  \frac{(-\alpha g)^k}{k!}{\log^k|z'(p)|}\,  
V^{(n-k)}_{z'}
=(1-\epsilon)\Big[V^{(n)}_{\mr{split}}(p,0)\Big]^\mr{sing}_{z}+\\
+\epsilon \underbrace{\left(
\sum_{k=1}^n \frac{(-\alpha g)^k}{k!} k\log^{k-1}|z(p)|\,V^{(n-k)}_{z}+\sum_{k=1}^{n-1} \frac{(-\alpha g)^k}{k!} \alpha g \log^{k}|z(p)|\,V^{(n-k-1)}_{z}
\right)}_{-\alpha g V^{(n-1)}_{z}}
\end{multline*}
Here $k$-th term in the first sum on the r.h.s., for $k\neq 1$, is cancelled by $(k-1)$-st term in the second sum.

Summing (\ref{V^n scaling transf}) over $n$, we obtain the transformation property for the vertex operator:
$$ (V_{X,\rho})_z\ra (V_{X,\rho})_{z'}= (1-\epsilon (1-\alpha g))(V_{X,\rho})_z $$
confirming the scaling dimension $\Delta+\bar{\Delta}=1-\alpha g$ we obtained in Proposition \ref{Prop: vertex operators}. The case of the second vertex operator, $W_{X,Y}$, is treated similarly.

\begin{example}
Starting with correlators (\ref{2-point (a bargamma...bargamma) gamma}), (\ref{2-point a (bargamma...bargamma gamma)}), contracting all $\bar\gamma$'s with $X$, contracting $a$ with $\rho$ and $\gamma$ with $Y$, and summing over $n\geq 0$, we obtain the following $2$-point functions
\begin{align}
\label{2-point V  (gamma-bargamma)}\lan V_{X,\rho}(z) \;\lan (\gamma-\bar\gamma)(w),Y\ran \ran =& \lan \rho,Y \ran \frac{|z-w|^{\alpha g}}{z-w}\\
\label{2-point a W}
\lan  \lan\rho, a(z) \ran\; W_{X,Y}(w)\ran  =& \lan \rho,Y \ran \frac{|z-w|^{-\alpha g}}{z-w}
\end{align}
--  power laws consistent with the  dimensions of the  primary fields involved.
\end{example}

\begin{remark}
Soaking fields $\tsoak$ and $\delta(c)$ 
(cf. Section \ref{sec: sphere}) are primary, of dimension $(0,0)$.\footnote{The idea of proof is as follows. Since the fields $\delta(c)$ and $\tsoak$
are $Q$-closed, one can recover their OPE with $T$ from their OPE with $G$ using (\ref{T Phi via G Phi}). For the OPE with $G$, there can be only poles of orders $1$ and $2$, due to weight counting. Second order pole in fact has no contributing diagrams (as follows from Feynman diagram combinatorics and weight restrictions). Coefficient of the first order pole is easily found as $G_{-1}\Phi(z)= \frac{1}{2\pi i}\oint G(w)\Phi(z)$, for $\Phi$ a soaking field, from the fact that $G_{-1}$ is a derivation (as proven by the contour switching argument similar to (\ref{J derivation})).
Interestingly, the renormalized product $\soak=\delta(c)\tsoak $ is non-primary -- already in the abelian case -- due to logarithmic singular subtractions (cf. (\ref{soaking fields OPE})).
}
Thus, correlators (\ref{2-point V  (gamma-bargamma)}), (\ref{2-point a W}) can be extended to $4$-point functions of primary fields on a sphere, e.g.
$$ \lan V_{X,\rho}(z_1) \;\lan (\gamma-\bar\gamma)(z_2),Y\ran \;\tsoak(z_3)\; \delta(c(z_4))\ran_{\CC P^1} = \lan \rho,Y \ran \left|\frac{z_{23}}{z_{12}z_{13}}\right|^{-\alpha g}\frac{z_{23}}{z_{12}z_{13}}$$
where $z_{ij}=z_i-z_j$ (note that the r.h.s. does not depend on $z_4$). This result is consistent with the ansatz for $4$-point functions of primary fields implied by global conformal invariance on $\CC P^1$.
\end{remark}

Furthermore, we can introduce the field
$$ H_X=e^{\lan \bar\gamma,X \ran} $$
It is primary, of conformal dimension $(0,0)$ (as proven by the same technology as in the proof of Proposition \ref{Prop: vertex operators} above; cf. also footnote \ref{footnote: bargamma^n no anomalous dim}). From (\ref{3-point a (bargamma ... bargamma) gamma}), we find the $3$-point function
\begin{equation}\label{3-point a H (gamma-bargamma)}
\Big\langle  \lan \rho,a(w_1)\ran\; H_X(z)\; \lan  (\gamma-\bar\gamma)(w_2),Y \ran  \Big\rangle=
\frac{\lan \rho, Y\ran  }{w_1-w_2}\left|\frac{z-w_2}{z-w_1}\right|^{\alpha g}
\end{equation}
OPEs of $H_X$ with either $a$ or $\gamma-\bar\gamma$ yield our two vertex operators:
\begin{align}
\lan \rho,a(w)\ran\; H_X(z) \sim& V_{X,\rho}(z)|w-z|^{-\alpha g}+o(1)_{w\ra z} \label{a H_X = V OPE}\\
\lan (\gamma-\bar\gamma)(w),Y \ran\; H_X(z)   \sim & W_{X,Y}(z)|w-z|^{\alpha g}+o(1)_{w\ra z}
 \label{(gamma-bargamma) H_X = W OPE}
\end{align}
Here (\ref{a H_X = V OPE}) follows immediately from the OPE (\ref{V^n split OPE}).

We make the following remarks.
\begin{itemize}
\item Our construction of ``vertex operators'' is based on exact summation of perturbation theory in all orders in $g$. E.g., non-trivial exponents in the correlators (\ref{2-point V  (gamma-bargamma)}), (\ref{2-point a W}), (\ref{3-point a H (gamma-bargamma)}) arise from the summation of powers of logs appearing in the correlators of Section \ref{sss: correlators with bargamma...bargamma}. 
\item Vertex operators are not differential polynomials (of finite order) in fundamental fields -- we need to add infinitely many monomials to produce a field of non-trivial dimension. 
\end{itemize}

\subsection{A remark on logarithmic phenomena}\label{sss: log phenomena}
Usually, primary fields are defined by their OPE with the energy-momentum tensor (\ref{T Phi primary OPE}). However, exploring phenomena like in (\ref{TV explicit}), we recall the refinement of this definition \cite{Gurarie}. A field is called ``pseudo-primary'' if 
it has at most a second-order pole in its OPE with $T, \bar{T}$.\footnote{
Another way to say it is: $\Phi$ is pseudo-primary if $L_n\Phi=0$, $\bar{L}_n=0$ for $n\geq 1$.
}
Then pseudo-primary fields form a closed subspace w.r.t. $L_0,\bar{L}_0$. If $L_0,\bar{L}_0$ acting on the space of pseudo-primary fields are jointly diagonalizable, we get the standard definition of primary fields. If not, we have the Jordan cell structure where only the lowest component is a primary field and all the rest are only pseudo-primary (but not primary). More precisely, the space of pseudo-primary fields splits into a direct sum of filtered subspaces\footnote{
I.e., a summand of the space of pseudo-primary fields is a filtered subspace $\mathbb{F}_0\subset \cdots \subset \mathbb{F}_r $, with $\dim \mathbb{F}_k=k+1$ and with $L_0$ preserving the filtration (while $L_0-\bar{L}_0$ acts on $\mathbb{F}_r$  as a multiple of identity). For each $k$ we choose a vector $\Phi_k\in \mathbb{F}_k$ with nonzero image in the quotient $\mathbb{F}_k/\mathbb{F}_{k-1}$. We can make this sequence of choices in such a way that $L_0$ has the standard Jordan cell form in the basis $\{\Phi_0,\ldots,\Phi_r\}$.
} $\mr{Span}\{\Phi_0,\ldots,\Phi_r\}$ -- Jordan cells -- satisfying the OPEs:
\begin{align}\label{log OPE T Phi_k}
T(w) \Phi_k(z)\sim& \frac{\Delta \Phi_k(z)+\Phi_{k-1}(z)}{(w-z)^2}+\frac{\dd \Phi_k(z)}{w-z}+\reg,
\\ \label{log OPE Tbar Phi_k}
\bar{T}(w) \Phi_k(z)\sim& \frac{\bar\Delta \Phi_k(z)+\Phi_{k-1}(z)}{(\bar{w}-\bar{z})^2}+\frac{\db \Phi_k(z)}{\bar{w}-\bar{z}}+\reg
\end{align}
where by convention  $\Phi_{-1}=0$.
Here $(\Delta,\bar\Delta)$ are called conformal dimensions. Actually, the condition that the infinitesimal rotation operator $L_0-\bar{L}_0$ integrates to a representation of the group $U(1)$ is tantamount to requiring that $L_0-\bar{L}_0$ is diagonalizable, with integer eigenvalues. Thus, we must have $\Delta-\bar\Delta\in \ZZ$ 
and the upper-triangular parts of $L_0,\bar{L}_0$ must be the same
(in other words, we have the same $\Phi_{k-1}$ appearing in the OPE of $\Phi_k$ with $T$ and with $\bar{T}$). 

OPEs (\ref{log OPE T Phi_k}), (\ref{log OPE Tbar Phi_k}) imply the following behavior of fields $\Phi_k$ under a change of coordinates $z\ra z'=\Lambda z$ with $\Lambda \in \CC-\{0\}$ a scaling factor:
\begin{equation}
(\Phi_k)_z\ra \quad (\Phi_k)_{z'}=\Lambda^{-L_0}\bar\Lambda^{-\bar{L}_0}(\Phi_k)_z=\Lambda^{-\Delta}\bar\Lambda^{-\bar\Delta}\sum_{j=0}^k \frac{(-2\log|\Lambda|)^j}{j!}(\Phi_{k-j})_z
\end{equation}
where the fields are at zero.

\begin{example}
Consider the theory with Lagrangian $b\dd\db c$ (this is the ghost sector of the abelian $BF$ theory). Field $1$ is primary. Field $\nl cb\nr$ is pseudo-primary, with $L_0 \nl cb\nr=1$ (this example of a logarithmic multiplet in the free $c=-2$ ghost system was studied in \cite{GFN,Flohr} and in the context of abelian $BF$ theory in \cite{LMY}). In this case, the rank of the Jordan cell is $2$.
This example generalizes to $N$-component abelian $BF$ theory, where one has a Jordan cell with $\Delta=\bar\Delta=0$ of rank $N+1$:
\begin{equation}\label{1-bc multiplet}
1\xleftarrow{L_0} C_1 \langle c,b \rangle \xleftarrow{L_0} C_2 \langle c,b \rangle^2 \xleftarrow{L_0} \cdots \xleftarrow{L_0} C_N \langle c,b \rangle^N
\end{equation}
where $C_k=\frac{(N-k)!}{k!N!}$ are normalization constants. Moreover, the exact same Jordan cell is present in the non-abelian theory, where $N=\dim\g$.
\end{example}

\begin{example}
In abelian $BF$ theory, with $\g=\RR$, the field $G$ has a logarithmic partner $$\sg=cb\, G+\frac12 b\,\dd a$$  
(here the second term is a correction to $cb\,G$ tuned in such a way that $L_{>0}\sg=0 $). I.e., the pair $(G,\sg)$ forms a rank $2$ Jordan cell. Furthermore, the field $$t=Q(\sg)$$ 
is a logarithmic partner for the stress-energy tensor $T$.\footnote{This construction is a version of the ansatz for the logarithmic partner $t_\mr{ghost}=cb\,T_\mr{ghost} +\cdots$ of the stress-energy tensor in the free ghost system studied in \cite{FML}.} In 
non-abelian $BF$ theory, 
$G$ becomes a part of a rank $N+1$ Jordan cell $(G,\sg_1,\ldots,\sg_N)$, where $N=\dim\g$, with
\begin{multline} \sg_k= G_{-2}\big(C_k\,  \langle c,b \rangle^k + C_{k-1}\,  \langle c,b \rangle^{k-1}\big)-\frac32 G_{-1}L_{-1} \big( C_k\,   \langle c,b \rangle^k \big) =
\\ =  C_k \langle c,b\rangle^k G+\cdots 
\end{multline}
where $C_k$ are the same normalization constants as in (\ref{1-bc multiplet}) and $G_n 
$ are the mode operators of the field $G$.\footnote{The fact that fields $\sg_k$ are pseudo-primary ($L_{>0}\sg_k=\bar{L}_{>0}\sg_k=0$) and form a Jordan cell ($L_0 \sg_k= {2 \sg_k+\sg_{k-1} }$, $\bar{L}_0\sg_k=\sg_{k-1}$) is checked straightforwardly from the extended $c=0$ Virasoro commutation relations $[G_n,L_m]=(n-m)G_{n+m}$, $[L_n,L_m]=(n-m)L_{n+m}$ (read off from the OPEs $T(w)G(z)$ and $T(w)T(z)$ -- items (\ref{Prop T (b)}), (\ref{Prop T (c)}) of Proposition \ref{Prop T}) 
and from (\ref{1-bc multiplet}).
} Accordingly, the stress-energy becomes a part of a rank $N+1$ ($Q$-exact) Jordan cell $(T,t_1,\ldots,t_N)$, with $t_k=Q(\sg_k)$.
\end{example}

\begin{example}\label{ex: free boson Jordan cells}
In the free scalar field theory, fields $1$ and $\phi$ are primary. Field $\nl\phi^2\nr$ is pseudo-primary, in the same Jordan cell as $1$, with $L_0 \nl \phi^2\nr=-1$.  Pseudo-primary field $\nl\phi^3\nr$ is in the same Jordan cell as $\phi$, with $L_0 \nl\phi^3\nr=-3\phi$. Actually, due to infinite-dimensionality of the space of pseudo-primary fields, these Jordan cells are of infinite rank.
\end{example}

\begin{example}\label{ex: V^n Jordan cell}
Consider the fields $V^{(k)}$ defined in (\ref{V^n}). For $k=0$, $V^{(0)}=\lan \rho, a\ran$ is a primary field. 
Fields $V^{(k)}$ for $k\geq 1$ are pseudo-primary and they are in the same Jordan cell, see (\ref{TV explicit}). Setting $\Phi_k=\left(-\frac{g \alpha}{2}\right)^{-k}V^{(k)} $, we have a standard basis for the Jordan cell.
\end{example}

\begin{remark}
Every time when we have a Jordan cell of infinite rank, 
we can form a family of  vertex operators
$$\mathsf{V}_\varkappa=\sum_{k=0}^\infty \varkappa^k \Phi_k$$
parameterized by $\varkappa\in\RR$. Each $\mathsf{V}_\varkappa$ is a primary field of conformal dimension ${(\Delta+\varkappa,\bar\Delta+\varkappa)}$, as follows from (\ref{log OPE T Phi_k}), (\ref{log OPE Tbar Phi_k}).
Note that by this mechanism the two infinite Jordan cells of Example \ref{ex: free boson Jordan cells} give rise to the vertex operators $\nl\cos (\alpha \phi)\nr$, $\nl\sin(\alpha\phi)\nr$ -- linear combinations of the standard vertex operators $\nl e^{\pm i\alpha\phi}\nr$. Likewise, when applied to the infinite Jordan cell of Example \ref{ex: V^n Jordan cell} with $\varkappa=-\frac{g \alpha}{2}$, this mechanism produces the new vertex operator $V_{X,\rho}$ defined in  (\ref{vertex operators}).
\end{remark}

\begin{remark}
While there is a lot of recent developments in logarithmic CFT, see e.g. \cite{Simmons,FGST,AM,GRW} ($c=0$ logarithmic theories), \cite{AP,RW} ($c=2$ bosonic ghost system) 
our mechanism of getting a logarithmic theory is different. Namely, in the case treated in this paper, logs are coming out of the integration of the interaction vertex over the worldsheet.\footnote{
We also have a different, well-known, source of logs coming from the abelian ghost system  \cite{GFN,Flohr}.
} This is somewhat similar to the mechanism of getting logs in \cite{FLN}, where logs are obtained by integration over the moduli space of instantons.
\end{remark}

\newpage
\appendix

\section{Some useful plane integrals}
Let $D_R=\{u\in\CC \;|\; |u|\leq R \}$ be a disk of radius $R$ in $\CC$ centered at zero. Then we have 
\begin{equation}\label{log-integral over D_R}
\int_{D_R}\frac{d^2 u}{\pi}\, \frac{1}{(u-z)(\bar{u}-\bar{w})} = \log \left( \frac{R^2-z\bar{w}}{|z-w|^2} \right)
\end{equation}
for $z\neq w$ two points inside $D_R$. One finds this by writing the integrand as $\frac{\dd}{\dd \bar{u}} \frac{\log(\bar{u}-\bar{w})}{u-z}$, replacing the integration domain with $D_R$ with a cut from $w$ to the boundary of $D_R$ and with a small disk around $z$ removed, and applying Stokes' theorem.
Explicitly, denoting the l.h.s. of (\ref{log-integral over D_R}) by $I_R(z,w)$ and 
denoting the new integration domain $\mc{D}$, we have:
\begin{multline*}
I_R(z,w)= \int_\mc{D} d\bar{u} \frac{\dd}{\dd \bar{u}}\; \frac{du}{2\pi i}\;\frac{\log(\bar{u}-\bar{w})}{u-z} = \int_{\dd \mc{D}}  \frac{du}{2\pi i}\;\frac{\log(\bar{u}-\bar{w})}{u-z}\\
=\underbrace{-\int_{-R}^w \frac{du}{u-z}}_{I} \underbrace{-\log(\bar{z}-\bar{w})}_{II}+ \underbrace{\int_{-\pi}^\pi \frac{ d\phi}{2\pi}\; R e^{i\phi}\; \frac{\log(R e^{-i\phi}-\bar{w})}{R e^{i\phi}-z}}_{III}
\end{multline*}
Three terms here come from components of the contour $\dd D$. Term I comes from the jump of the integrand on the cut between $u=w$ and $u=-R$ and evaluates to $\log\frac{R+z}{z-w}$. Term II is the contribution of the small circle around $u=z$. Term III is the contribution of the big circle $\dd D_R$; it evaluates to
\begin{multline*}
\int_{-\pi}^\pi \frac{ d\phi}{2\pi}\; \frac{R-i\phi+\log(1-\frac{\bar{w}}{R} e^{i\phi})}{1-\frac{z}{R}e^{-i\phi}}\\ = 
\int_{-\pi}^\pi \frac{d\phi}{2\pi} \Big( \log R\; \sum_{p\geq 0}\left(\frac{z}{R}\right)^p e^{-i p \phi}  - i\phi \sum_{p\geq 0} \left(\frac{z}{R}\right)^p e^{-i p \phi}  
- \underbrace{ \sum_{p\geq 0} \left(\frac{z}{R}\right)^p e^{-i p \phi}\cdot
 \sum_{q\geq 1} \frac{1}{q}\left(\frac{\bar{w}}{R}\right)^q e^{i q \phi} }_{\mbox{only $p=q$ contributes}}
\Big)\\
=\log R +\sum_{p\geq 1}\frac{1}{p} \left(-\frac{z}{R}\right)^p -\sum_{p\geq 0} \frac{1}{p} \left( \frac{z\bar{w}}{R^2}\right)^p = \log R- \log(1+\frac{z}{R})+\log(1-\frac{z\bar{w}}{R^2})=\log \frac{R^2-z\bar{w}}{R+z}
\end{multline*}
Collecting all the terms, we get the result (\ref{log-integral over D_R}).

Similarly, one can treat the cases when one or both points $z,w$ are outside $D_R$: 
\begin{equation}
\begin{array}{ll}
I_R(z,w)= - \log\left(1-\frac{R^2}{z\bar{w}}\right) & \mbox{if } |z|,|w|>R \\
I_R(z,w)= - \log\left(1-\frac{w}{z}\right) & \mbox{if } |w|<R<|z| \\
I_R(z,w)= - \log\left(1-\frac{\bar{z}}{\bar{w}}\right) & \mbox{if } |z|<R<|w| 
\end{array}
\end{equation}

One can use (\ref{log-integral over D_R}) to evaluate integrals over $\CC$ of products of expressions $\frac{1}{u-z_i}$ and $\frac{1}{\bar{u}-\bar{z}_i}$. For example, for $z,w,x$ three distinct points in $\CC$ we have
\begin{multline}\label{3-point integral}
\int_\CC \frac{d^2 u}{\pi} \;\frac{1}{(u-z)(u-x)(\bar{u}-\bar{w})}= \frac{1}{z-x} \lim_{R\ra\infty}(I_R(z,w)-I_R(x,w))
\\= \frac{2}{z-x}\log\left|\frac{x-w}{z-w} \right|
\end{multline}
where we used the expansion $\frac{1}{(u-z)(u-x)}=\frac{1}{z-x}(\frac{1}{u-z}-\frac{1}{u-w})$ to reduce the integral to (\ref{log-integral over D_R}). Integral (\ref{3-point integral}) is crucial for the computation of $3$-point functions.

Another useful integral of this type is
\begin{equation}\label{int 1/((u-z)(u-x))}
\int_\CC \frac{d^2 u}{\pi} \;\frac{1}{(u-z)(u-x)}=-\frac{\bar{z}-\bar{x}}{z-x}
\end{equation}
One obtains it by presenting the integrand as $\frac{\dd}{\dd\bar{u}}\,\frac{\bar{u}}{(u-z)(u-x)}$ and using Stokes' theorem on the plane with two small disks around $u=z$ and $u=x$ removed.

\subsection{The dilogarithm integral}\label{appendix: dilog integral}
The following integral over a disk is useful for evaluating $4$-point functions
and can be evaluated in terms of the dilogarithm function:
\begin{multline}\label{dilog-integral over D_R}
\int_{D_R}\frac{d^2 u}{\pi}\, \frac{\log|u|}{(u-z)(\bar{u}-\bar{w})} = \\
= \log^2 R+i\DD\left(\frac{z}{w}\right)
-\log|zw|\cdot \log|z-w|+\log|z|\cdot\log|w|
+ O\left(\frac{\log R}{R^2}\right)
\end{multline}
Here $\DD(z)=\mr{Im}\,\mr{Li}_2(z)+\arg(1-z)\log|z|$ is the \textit{Bloch-Wigner dilogarithm},
see \cite{Zagier}. It is the monodromy-free variant of the standard dilogarithm $\mr{Li}_2(z)=-\int_0^z dt\frac{\log(1-t)}{t} $ -- the analytic continuation of the sum $\sum_{n\geq 1}\frac{z^n}{n^2}$ convergent on the disk $|z|\leq 1$. In particular, $\DD(z)$ is a real-analytic function everywhere on $\CC P^1$ except at $z=0,1,\infty$ where it is continuous (and vanishes) but is not differentiable. Function $\DD(z)$ satisfies the identity $\DD(1/z)=-\DD(z)$,\footnote{The more general identity is that, under a M\"obius transformation permuting points $0,1,\infty$, $\DD(z)$ changes by the sign of the permutation:  $\DD(z)=\DD(\frac{1}{1-z})=\DD(1-\frac{1}{z}) =-\DD(\frac{1}{z})=-\DD(1-z)=-\DD(\frac{z}{z-1})$.} 
thus it is clear that the r.h.s. of (\ref{dilog-integral over D_R}) conjugates when $z$ and $w$ are interchanged.
The $O\left(\frac{\log R}{R^2}\right)$ remainder term in (\ref{dilog-integral over D_R}) can be written explicitly as 
$\log R \log(1-\frac{z \bar{w}}{R^2})-\frac12 \mr{Li}_2(\frac{z\bar{w}}{R^2})$.

Starting from (\ref{dilog-integral over D_R}), similarly to (\ref{3-point integral}), one obtains
\begin{multline}\label{3-point dilog integral}
\int_\CC \frac{d^2 u}{\pi} \frac{\log|u|}{(u-z_1)(u-z_2)(\bar{u}-\bar{z}_3)}=\\
=\frac{1}{z_1-z_2}\left(i\DD\left(\frac{z_1}{z_3}\right)-\log|z_1z_3|\cdot \log|z_1-z_3|+\log|z_1|\cdot \log|z_3| -\quad  \Big(z_1\leftrightarrow z_2\Big)\right)
\end{multline}
The last term in the brackets stands for the previous terms with $z_1$ replaced by $z_2$.



\end{document}